\documentclass[10pt]{article}

\usepackage{amsfonts,amssymb,amsmath,comment,mathtools,amsthm}
\usepackage{graphicx,graphics,color}
\usepackage[permil]{overpic}
\usepackage{caption}
\usepackage{subcaption}
\usepackage{pgf,tikz}
\usetikzlibrary{arrows}
\usepackage{enumerate}
\usepackage{cancel}
\usepackage{hyperref}  
\usepackage{mathrsfs}
\usepackage[normalem]{ulem}
 
 \makeatletter
 \@addtoreset{equation}{section}
 \makeatother

 \newcounter{extralabel}[section]

 \newtheorem{ittheorem}{Theorem}
 \newtheorem{itlemma}{Lemma}
 \newtheorem{itproposition}{Proposition}
 \newtheorem{itdefinition}{Definition}
 \newtheorem{itcorollary}{Corollary}
 \newtheorem{itconjecture}{Conjecture}
 \newtheorem{itremark}{Remark}

 \newenvironment{theorem}{\addtocounter{extralabel}{1}
 \begin{ittheorem}}{\end{ittheorem}}

 \newenvironment{lemma}{\addtocounter{extralabel}{1}
 \begin{itlemma}}{\end{itlemma}}

 \newenvironment{proposition}{\addtocounter{extralabel}{1}
 \begin{itproposition}}{\end{itproposition}}

 \newenvironment{definition}{\addtocounter{extralabel}{1}
 \begin{itdefinition}}{\end{itdefinition}}

 \newenvironment{corollary}{\addtocounter{extralabel}{1}
 \begin{itcorollary}}{\end{itcorollary}}

 \newenvironment{conjecture}{\addtocounter{extralabel}{1}
 \begin{itconjecture}}{\end{itconjecture}}
 
 \newenvironment{remark}{\addtocounter{extralabel}{1}
 \begin{itremark}}{\end{itremark}}

\newcommand{\e}{\mathrm{e}}
\newcommand{\dd}{\mathrm{d}} 
\def\d{\mathrm{d}}
\def\z{\boldsymbol{z}}

\def\x{\boldsymbol{x}}
\def\y{\boldsymbol{y}}
\def\t{\boldsymbol{t}}
\def\r{\boldsymbol{r}}
\def\s{\boldsymbol{s}}
\def\rr{\boldsymbol{\rho}}
\def\tt{\boldsymbol{\theta}}
\def\vt{\boldsymbol{\vartheta}}
\def\vp{\boldsymbol{\varphi}}
\def\ZZ{\boldsymbol{Z}}

\newcommand{\eps}{\varepsilon}

\newcommand{\be}{\begin{equation}}
\newcommand{\ee}{\end{equation}}
\newcommand{\ba}{\begin{equation} \begin{aligned}}
\newcommand{\ea}{\end{aligned}\end{equation}}
\newcommand{\bes}{\begin{equation*}}
\newcommand{\ees}{\end{equation*}}

\definecolor{forestgreen}{cmyk}{0.91,0,0.88,0.12}
\newcommand{\RK}[1]{{\color{forestgreen}#1}} 

\newcommand{\dist}{{\operatorname {dist}}}

\DeclareMathOperator{\reach}{reach}
  
\DeclarePairedDelimiter\abs{\lvert}{\rvert}
\DeclarePairedDelimiter\norm{\lVert}{\rVert}

\renewcommand{\P}{\mathbb{P}}
\newcommand{\E}{\mathbb{E}}
\def \T {{\mathbb T}}
\def \Z {{\mathbb Z}}
\def \R {{\mathbb R}}
\def \N {{\mathbb N}}
\def \Q {{\mathbb Q}}

\def \cC {{\mathcal C}}
\def \cD {{\mathcal D}}

\def \cI {{\mathcal I}}

\def \cO {{\mathcal O}}  
\def \cH {{\mathcal H}}

\def \cV {{\mathcal V}}
\def \cP {{\mathcal P}}

\def\ka{\kappa}

\def\be{\begin{equation}}
\def\ee{\end{equation}}

\def\e{\mathrm{e}}

\def\Rc{R_{\mathrm{c}}}

\def\UB{\mathrm{UB}}
\def\LB{\mathrm{LB}}

\usepackage{dsfont}

\newcommand{\1}{{\mathds1}}

\begin{document}


\title{The Widom-Rowlinson model:\\ 
Mesoscopic fluctuations for the critical droplet}
\author{\renewcommand{\thefootnote}{\arabic{footnote}}
Frank den Hollander
\footnotemark[1]
\\
\renewcommand{\thefootnote}{\arabic{footnote}}
Sabine Jansen
\footnotemark[2]
\\
\renewcommand{\thefootnote}{\arabic{footnote}}
Roman Koteck\'y
\footnotemark[3]
\\
\renewcommand{\thefootnote}{\arabic{footnote}}
Elena Pulvirenti
\footnotemark[4]
}

\date{13 March 2026}

\maketitle

\begin{abstract} 
We study the critical droplet for a close-to-equilibrium Widom-Rowlinson model of interacting particles, represented by disks of radius $1$, in the two-dimensional plane at low temperature. The critical droplet is the set of macroscopic states that correspond to saddle points for the passage from a low-density supersaturated vapour to a stable high-density liquid. We analyse the mesoscopic fluctuations of the surface of the critical droplet, which turns out to be the set of particle configurations that are close to a disk of a certain deterministic radius. Our results represent the first detailed rigorous analysis of the surface fluctuations of a continuum interacting particle system exhibiting condensation and, as such, constitute a fundamental step in the study of phase separation from the perspective of stochastic geometry. At the same time, our results serve as a basis for the study of a non-equilibrium version of the Widom-Rowlinson model, to be analysed in \cite{dHJKP1}, where they lead to a correction term in the Arrhenius formula for the average vapour-liquid crossover time.
\end{abstract}

\vglue1cm
\noindent
{\it AMS} 2010 {\it subject classifications.} 
60J45, 
60J60, 
60K35; 
82C21, 
82C22, 
82C27. 

\medskip\noindent
{\it Acknowledgment.} FdH was supported by NWO Gravitation Grant 024.002.003-NETWORKS, FdH and EP by ERC Advanced Grant 267356-VARIS, and RK by Czech Science Foundation Grant 25-16267S. EP was supported by the German Research Foundation in the Collaborative Research Centre 1060 \emph{The Mathematics of Emergent Effects}. The authors acknowledge support by The Leverhulme Trust through International Network Grant \emph{Laplacians, Random Walks, Bose Gas, Quantum Spin Systems}.  SJ thanks Christoph Th\"ale for fruitful discussions.

\medskip\noindent
{\it Data availability.} No datasets were generated or analyzed during the current study.

\footnotetext[1]{Mathematical Institute, Leiden University, Einsteinweg 55,
2333 CC Leiden, The Netherlands,\\
\emph{denholla@math.leidenuniv.nl}}

\footnotetext[2]{
Mathematisches Institut, Ludwig-Maximilians-Universit\"at, Theresienstrasse 39, 
80333 M\"unchen, Germany,\\
\emph{jansen@math.lmu.de}}

\footnotetext[3]{
UTIA, Czech Academy of Sciences  and Center for Theoretical Study, Charles University, Prague, Czech Republic and
Mathematics Institute, University of Warwick, Coventry CV4 7AL, United Kingdom,\\
\emph{r.kotecky@icloud.com}}

\footnotetext[4]{Delft Institute of Applied Mathematics, 
Delft University of Technology, Mekelweg 4, 2628 CD Delft, The Netherlands,\\ 
\emph{e.pulvirenti@tudelft.nl}}

\newpage

\small
\tableofcontents
\normalsize


\section{Model, results and background}
\label{S:intro}


\subsection{Introduction and outline}

In this paper we study the boundary fluctuations of a continuum interacting particle system exhibiting condensation. The microscopic model that we consider is the \emph{Widom-Rowlinson model}, introduced in \cite{WR} to describe a vapour-liquid phase transition, where particles are unit disks in the plane that interact with each other when they overlap. This model is one of the few models in the continuum for which a vapour-liquid phase transition has been treated rigorously (see Ruelle~\cite{Ru2}, Chayes, Chayes and Koteck{\'y}~\cite{CCK}). The phase transition occurs as the chemical potential, controlling the density of the particles, crosses a threshold value.

Our objective is to give a precise description of the so-called \emph{critical droplet}, i.e., macroscopic states constituting the saddle points with minimal free energy connecting the vapour state and the liquid state. In order to do so, we study the finite-volume model close to and above the phase transition line in the limit of \emph{low temperature}. It turns out that the critical droplet is close to a disk of a certain deterministic radius, with a boundary that is random and consists of a large number of unit disks that stick out by a small distance. We provide an analysis of the surface fluctuations on a mesoscopic level. In the physics literature, density fluctuations at the surface of macroscopic droplets in the continuum are studied in terms of \emph{capillary waves} (see Rowlinson and Widom~\cite{RW} and Stillinger and Weeks~\cite{SW*}). Mathematically, we identify the asymptotics for the partition function of the Widom-Rowlinson model on those configurations where the union of the unit disks is close to the critical disk both in volume  and in Hausdorff distance.

Our results not only are of interest in the study of phase separation in continuum interacting particle systems, but also serve as a building block in the study of \emph{metastability} when we consider a dynamical version of the Widom-Rowlinson model where particles are allowed to enter or leave the system, in the regime of supersaturated vapour where the temperature is low and the chemical potential above the coexistence line. A detailed understanding of the shape and the fluctuations of the critical droplet is crucial for the computation of the \emph{metastable crossover time} in the dynamic model.

Section~\ref{WRstatic} introduces the Widom-Rowlinson model. Section~\ref{target} introduces the main object of our investigation, namely, the critical droplet. Section~\ref{LDiso} states a \emph{large deviation principle} for the shape (Theorem~\ref{thm:ldp-halo}) and the volume (Theorem~\ref{thm:ldp-volume}) of the \emph{halo} of the particle configuration in the Widom-Rowlinson model, and show that the corresponding rate functions are linked via an \emph{inequality of isoperimetric type} (Theorem~\ref{thm:isope}). Section~\ref{zoomin} states the main theorem (Theorem~\ref{thm:zoom1}) about \emph{moderate deviations} for the halo volume, and states a conjecture for sharp asymptotics for volumes that are close to the volume of the critical droplet (Conjecture~\ref{thm:zoom2}). Section~\ref{outline} contains an outline of the remainder of the paper. For background on large deviation theory, see e.g.\ Dembo and Zeitouni~\cite{DZ} or den Hollander~\cite{dH}.

\subsection{The Widom-Rowlinson model}
\label{WRstatic}

The Widom-Rowlinson model is an interacting particle system in $\R^2$ where the particles are the centres of unit disks with an attractive interaction. It was introduced in Widom and Rowlinson~\cite{WR} to model liquid-vapour phase transitions, and is one of the rare models in the continuum for which a phase transition has been established rigorously. In the present paper we place the particles on a finite torus in $\R^2$.    

\begin{figure}[htbp]
\vspace{0.5cm}
\begin{center}
\includegraphics[width=0.3\linewidth]{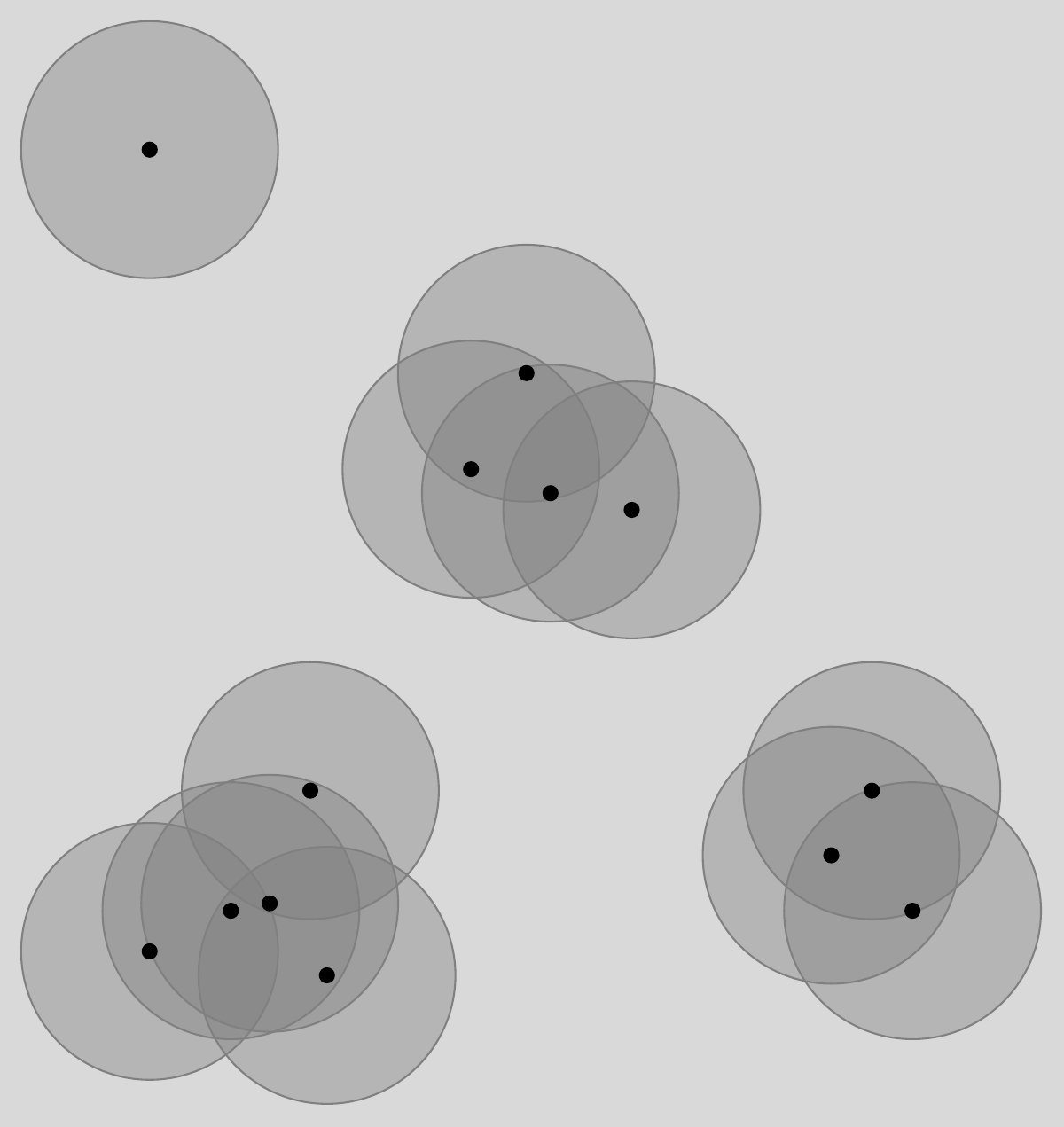}
\hspace{2cm}\includegraphics[width=0.3\linewidth]{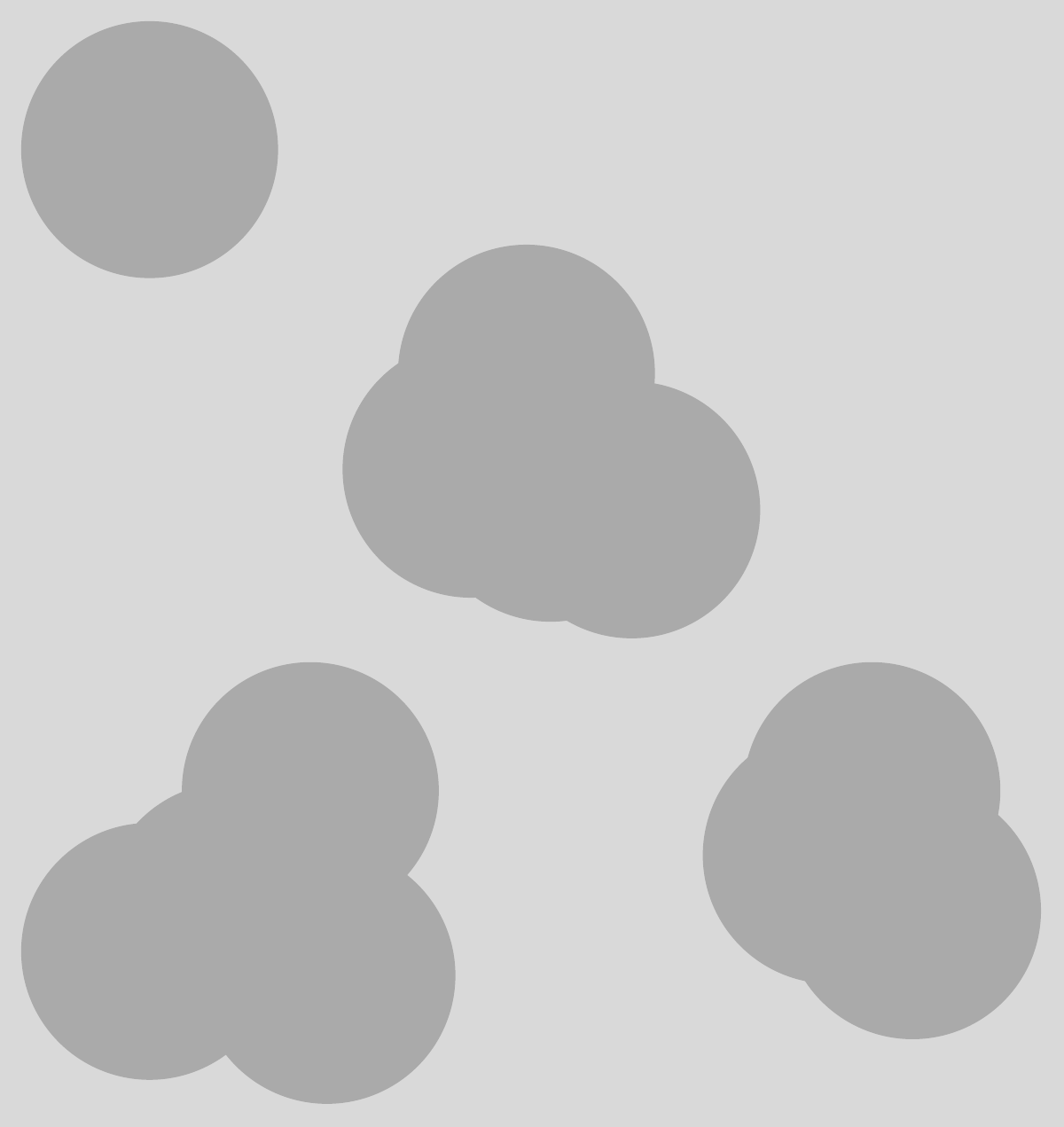}
\end{center}
\caption{\small 
A particle configuration $\gamma\in\Gamma$  (with a penetrable unit disk $B(x)$ around each particle $x\in\gamma$) and its halo $h(\gamma)$ (determining the energy of $\gamma$).}
\label{fig:ninecircles}
\end{figure}

Fix $L \in (4,\infty)$ and let $\T=\T_L=\R^2/(L\Z)^2$ be the torus of side-length $L$. We can identify $\T$ with the set $[-\tfrac12 L, \tfrac12 L)^2$ after we redefine the distance by
\be
\dist(x,y) = \inf_{k\in\Z} \abs{x-y+kL}, \qquad x,y\in\R^2.
\ee
The set $\Gamma=\Gamma_\T$ of finite particle configurations in $\T$ is 
\be
\Gamma = \{ \gamma \subset {\T}\colon\,N(\gamma)\in\N_0\},
\ee
where $N(\gamma)$ denotes the cardinality of $\gamma$, i.e., particles are viewed as non-coinciding points that are indistinguishable.  Particles are surrounded by freely penetrable disks of radius $1$, serving only as a tool for calculating the configurational energy. The \emph{halo} of a configuration $\gamma\in\Gamma$ is defined as (see Fig.~\ref{fig:ninecircles})
\be
\label{halodef}
h(\gamma) = \bigcup_{x\in \gamma} B(x),
\ee 
where $B(x)=B_1(x)$ is the closed disk of radius $1$ centred at $x\in {\T}$. The \emph{energy} $E(\gamma)$ of a configuration $\gamma\in\Gamma$ is defined as 
\be
\label{Udef}
E(\gamma) = V(\gamma)-V_0N(\gamma) 
= \Bigl| \bigcup_{x\in\gamma} B(x)\Bigr| - \sum_{x\in\gamma} \abs{B(x)},
\ee
where $V(\gamma)=\abs{h(\gamma)}$ and $V_0=|B(0)|=\pi$. The energy vanishes when the unit disks do not overlap, and reaches its minimal value when all unit disks coincide. Since $0 \geq E(\gamma) \geq -\pi[N(\gamma)-1]$, the interaction is attractive and stable (Ruelle~\cite[Section 3.2]{Ru1}). 

We define the \emph{grand-canonical Gibbs measure} $\mu_\beta=\mu_{{\T},z,\beta}$ as the probability measure on $\Gamma$ defined as
\be
\label{gibbs1}
\mu_{\beta,z}(\d\gamma) = \frac{1}{\Xi_{\beta,z}}\, \e^{-\beta  E(\gamma)}\,z^{N(\gamma)}\,\Q(\d\gamma), \qquad \gamma \in \Gamma,
\ee
where $\beta \in (0,\infty)$ is the inverse temperature, $z\in (0,\infty)$ is the activity,  $\Q$ is the law of the homogeneous Poisson point process on ${\T}$ with intensity $1$, and $\Xi_{\beta,z}=\Xi_{{\T},z,\beta}$ is the normalisation
\be
\label{Xidef}
\Xi_{\beta,z} = \int_\Gamma \Q(\d\gamma)\,z^{N(\gamma)}\,\e^{-\beta E(\gamma)}.
\ee

\begin{figure}[htbp]
\vspace{0.4cm}
\begin{center}
\setlength{\unitlength}{0.4cm}
\begin{picture}(10,8)(0,0)
{\thicklines
\qbezier(0,-1)(0,3)(0,7)
\qbezier(-1,0)(6,0)(12,0)
\qbezier(5,3)(5.5,2.5)(5.8,2)
\qbezier(5.8,2)(7,0.4)(11,0.3)
}
\qbezier[30](0,0)(1.75,3.3)(3.3,3.5)
\qbezier[15](3.3,3.5)(4.5,3.5)(5,3)
\qbezier[20](5,0)(5,2)(5,3)
\put(12.5,-.25){$\beta$}
\put(-.8,7.5){$z_t(\beta)$}
\put(4.7,-1){$\beta_c$}
\put(5.9,0.4){\tiny {\rm vapour}}
\put(6.9,1.4){\tiny {\rm liquid}}
\put(5,3){\circle*{0.35}}
\end{picture}
\end{center}
\vspace{0.5cm}
\caption{\small 
Graph of the coexistence line $\beta \mapsto z_t(\beta)$.}
\label{fig-zc}
\vspace{0.2cm}
\end{figure}

In the thermodynamic limit, i.e., when $L \to \infty$,  the liquid-vapour \emph{phase transition} occurs at the coexistence line $z=z_t(\beta)$ with (see Fig.~\ref{fig-zc})
\begin{equation}
\label{phtr}
z_t(\beta) = \beta\,\e^{-\beta \pi}, \qquad \beta>\beta_c \in (0,\infty) 
\end{equation}
(Ruelle~\cite{Ru2}, Chayes, Chayes and Koteck{\'y}~\cite{CCK}). No closed form expression is known for the critical inverse temperature $\beta_c$. We place ourselves in the \emph{regime of the supersaturated vapour} at low temperature,
\be
\label{metaregalt}
z = \ka z_t(\beta), \quad \ka \in (1,\infty), \qquad \beta\to\infty.
\ee
In other words, we set our system to be in the vapour phase and choose the activity $z$ to lie in the region corresponding to the liquid phase, above the phase coexistence line in Fig.~\ref{fig-zc} representing the phase transition in the thermodynamic limit, and we let $\beta\to\infty$ and $z \downarrow 0$ in such a way that we keep close to the phase coexistence line by a fixed factor $\kappa$. For this choice the Gibbs measure in \eqref{gibbs1} with $z=\ka z_t$ becomes
\be
\label{gibbs1alt}
\mu_\beta(\d\gamma) = \frac{1}{\Xi_\beta}\,(\kappa\beta)^{N(\gamma)}\,\e^{-\beta V(\gamma)}\,
\Q(\d\gamma), \qquad \gamma \in \Gamma,
\ee
where $\Xi_{\beta}=\Xi_{{\T},\beta}$ is the normalisation
\be
\label{Xidefalt}
\Xi_{\beta} = \int_\Gamma \Q(\d\gamma)\,(\kappa\beta)^{N(\gamma)}\,\e^{-\beta V(\gamma)}.
\ee
Note that \eqref{gibbs1alt} favours large particle numbers and disfavours large halos. 

Let $\Pi_{\kappa\beta}$ be the homogeneous Poisson point process on $\mathbb T$ with intensity $\kappa\beta$. Then $\mu_\beta$ is absolutely continuous with respect to its law $\Q_{\kappa\beta}$ with Radon-Nikodym derivative 
\be
\frac{\dd \mu_\beta}{\dd \Q_{\kappa \beta}}(\gamma) 
= \frac{\exp( - \beta V(\gamma))}{\int_\Gamma \exp( - \beta V)\,\dd \Q_{\kappa\beta}},
\qquad \gamma \in \Gamma.
\ee


\subsection{Key target: shape and fluctuations of the critical droplet}
\label{target}

For $\kappa \in (1,\infty)$, we define the following function and its maximum point (see Fig.~\ref{fig-Rc})
\be
\label{USkapdef}
\Phi_\kappa(R) = \pi R^2 - \kappa \pi (R-1)^2, \quad R \in [1,\infty),
\qquad \Rc(\kappa) = \frac{\kappa}{\kappa-1}.
\ee
The role of $\Rc(\kappa)$ is that it is the radius of a \emph{critical disk} representing the critical droplet that the particle configuration closely resembles. Throughout the paper, $L$ and $\kappa \in (1,\infty)$ are fixed so that $1<\Rc(\kappa)< \frac{L}{\pi}  + \frac12$ (recall that  $L$ is the linear size of the torus $\T=\T_L$). Define
\be
\label{GammaK}
\Phi(\kappa) = \Phi_\kappa(\Rc(\kappa)) = \frac{\pi\kappa}{\kappa-1},
\qquad \Psi(\kappa) = \frac{c^*\kappa^{2/3}}{\kappa-1},
\ee
where $c^*$ is a constant that will be identified below and that does not depend on $\kappa$. 

\begin{figure}[htbp]
\vspace{0.2cm}
\begin{center}
\setlength{\unitlength}{0.4cm}
\begin{picture}(12,5)(8,2)
{\thicklines
\qbezier(1,0)(6,0)(11,0)
\qbezier(2,-1)(2,2)(2,6)
\qbezier(3,2)(6,8)(10,-2)
}
\qbezier[20](2,4.2)(3.7,4.2)(5.4,4.2)
\qbezier[30](5.4,0)(5.4,2.5)(5.4,4.2)
\qbezier[15](3,0)(3,1)(3,2)
\put(11.5,-.15){$R$}
\put(1,6.5){$\Phi_\kappa(R)$}
\put(-.2,4){$\Phi(\kappa)$}
\put(2.85,-1){$1$}
\put(4.7,-1){$\Rc(\kappa)$}
\put(5.4,4.25){\circle*{0.35}}
{\thicklines
\qbezier(17.2,5)(18,1)(25,1.2)
\qbezier(16,-1)(16,3)(16,6)
\qbezier(15,0)(20,0)(26,0)
}
\qbezier[50](16,1)(20,1)(25,1)
\qbezier[40](17,0)(17,3)(17,6)
\put(26.5,-.15){$\kappa$}
\put(15,6.5){$\Rc(\kappa)$}
\put(16.8,-1){$1$}
\put(15.3,.8){$1$}
\end{picture}
\end{center}
\vspace{1.5cm}
\caption{{\small Picture of $R \mapsto \Phi_\kappa(R)$ for fixed $\ka\in(1,\infty)$ and $\kappa \mapsto \Rc(\kappa)$.}}
\label{fig-Rc}
\vspace{0.2cm}
\end{figure}

Fix $C \in (0,\infty)$, abbreviate $\delta(\beta) = \beta^{-2/3}$, and define
\be
\label{Diri6}
\begin{aligned}
I(\kappa,\beta;C) 
&= \int_\Gamma \Q(\d\gamma)\,(\kappa\beta)^{N(\gamma)}\,\e^{-\beta V(\gamma)}\,
\1_{\{|V(\gamma) - \pi \Rc(\kappa)^2| \leq C\delta(\beta)\}}\\
&= \Xi_\beta \,\mu_\beta \Bigl( |V(\gamma) - \pi \Rc(\kappa)^2| \leq C\delta(\beta) \Bigr). 
\end{aligned}
\ee
The main objective of the present paper is to obtain sharp upper and lower bounds on $I(\kappa,\beta;C)$. Our result is a  statement about a \emph{restricted equilibrium}: we provide sharp bounds on the probability that the halo has a volume that lies inside an interval of width $C\beta^{-2/3}$ around the volume of the critical disk $B_{\Rc(\kappa)}(0)$. In addition, we prove that this restricted equilibrium represents the formation of a \emph{critical droplet}. A critical droplet in the context of metastability refers to a cluster of particles that has reached the minimum size required for a phase transition to occur. Therefore the significance of $I(\kappa,\beta;C)$ arises also in the analysis of the metastable behaviour of a dynamical version of the Widom-Rowlinson model where particles are created or annihilated according to a Glauber dynamics. In this context, the special role of the critical disk $B_{\Rc(\kappa)}(0)$ becomes apparent through the fact that the set $\{\gamma \in \Gamma\colon\,|V(\gamma) - \pi \Rc(\kappa)^2| \leq C\delta(\beta)\}$ forms the \emph{gate} for the metastable transition from the vapour phase (`$\T$ empty') to the liquid phase (`$\T$ full') in the metastable regime \eqref{metaregalt}.

\begin{figure}[htbp]
\vspace{0.2cm}
\begin{center}
\includegraphics[width=7cm]{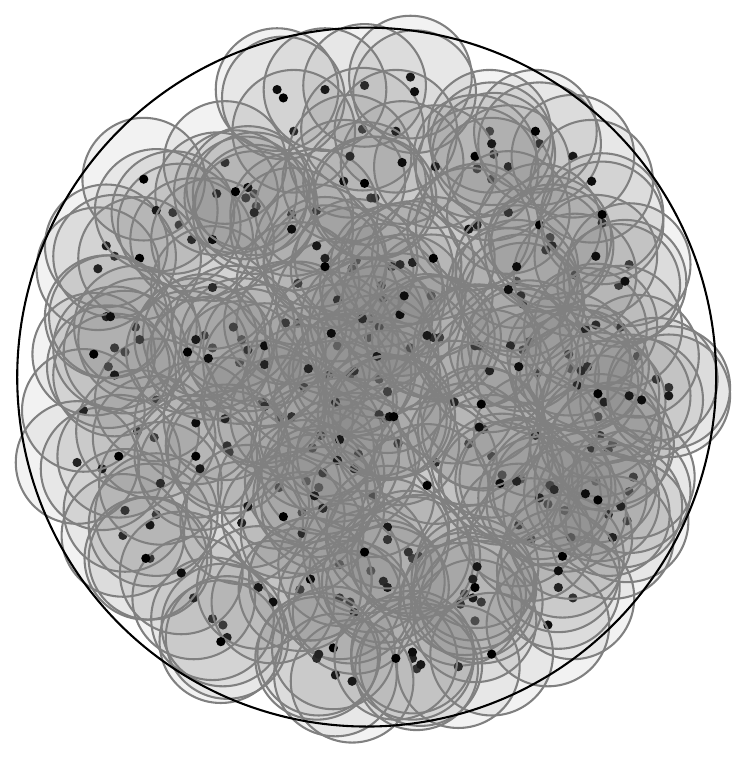}
\end{center}
\caption{\small A picture of the critical droplet in the metastable regime \eqref{metaregalt}. The critical droplet is given by a cluster of particles close to $B_{\Rc(\kappa)}$, a disk of radius $\Rc(\kappa)$. It has a random boundary that fluctuates within a narrow annulus whose width shrinks to zero as $\beta\to\infty$. In the course of the paper it will become clear that the number of unit disks lying fully inside $B_{\Rc(\kappa)}$ is $\kappa\beta$ (modulo a constant), while the number of unit disks that touch the boundary $\partial B_{\Rc(\kappa)}$ is $\frac{\kappa^{2/3}}{\kappa-1}\beta^{1/3}$ (modulo a constant).}
\label{fig:critdrop}
\end{figure}

\medskip\noindent
In the approximation of $I(\kappa,\beta;C)$, we will see that the two quantities defined in \eqref{GammaK} play a pivotal role. Namely, $\Phi(\kappa)$ is the minimum of the rate function in the large deviation principle with rate $\beta$ for the \emph{volume of the halo}, and determines the \emph{shape of the critical droplet}, which is close to a disk with radius $\Rc(\kappa)$, while $\Psi(\kappa)$ comes from the \emph{fluctuations of the critical droplet} and is obtained via a \emph{weak} large deviation principle for the random variable $\beta^{2/3} |V-\pi \Rc^2(\kappa)|$ under the Gibbs measure in \eqref{gibbs1alt}. The rate is $\beta^{1/3}$ and the rate function is \emph{degenerate}, being equal to the constant $\Psi(\kappa)$. This degeneracy reflects the fact that the radius of the critical droplet is close to $\Rc(\kappa)$, for which $\Phi'_\kappa(\Rc(\kappa)) = 0$. Intuitively, we may think of $\Phi(\kappa)$ as (a leading order approximation of) the \emph{free energy} of the critical droplet, consisting of the bulk free energy and the surface tension, and of $\Psi(\kappa)$ as (a leading order approximation of) the \emph{entropy} associated with the fluctuations of the surface of the critical droplet, which plays the role of a correction term to get the full free energy. We will see that there are order $\beta$ unit disks inside the critical droplet and order $\beta^{1/3}$ unit disks touching its boundary (see Fig.~\ref{fig:critdrop}). We remark that $\beta$ is to be viewed as a dimensionless quantity, i.e., the inverse temperature divided by a unit of energy. Otherwise, its fractional powers would not make sense.

In Section~\ref{LDP-WR} we prove the following asymptotics for the partition function defined in \eqref{Xidefalt}:
\be
\label{partsum}
\Xi_\beta = \e^{[(\kappa-1)\beta-1]|\mathbb T|}[1+o(1)], \qquad \beta \to \infty.
\ee
This reflects the leading-order approximation of the bulk free energy in the liquid phase.


\subsection{Shape of the critical droplet: isoperimetric inequalities and large deviations}
\label{LDiso}


\paragraph{Admissible sets.}

Let $\mathcal{F}_\T$ be the family of non-empty closed (and hence compact) subsets of the torus $\T$. We equip $\mathcal{F}_\T$ with the Hausdorff metric
\be
\label{E:dH}
\begin{aligned}
d_{\text{\rm{H}}}(F_1,F_2) 
&= \max\Bigl\{\max_{x\in F_1}\dist (x,F_2),\max_{x\in F_2}\dist (x,F_1)\Bigr\}\\
&= \min\Bigl\{\eps\geq 0\colon\, F_1\subset F_2+\eps B(0), 
F_2\subset F_1+\eps B(0)\Bigr\}, \qquad F_1,F_2 \neq \emptyset,
\end{aligned}
\ee
where $\dist(x,F)=\min_{y\in F} \dist(x,y)$. This turns $\mathcal{F}_{\T}$ into a compact metric space (Matheron~\cite[Propositions 12.2.1, 1.4.1, 1.4.4]{Ma}, Schneider and Weil~\cite[Theorems 12.2.1, 12.3.3]{SW}). Let $\mathcal{S}\subset\mathcal{F}_{\T}$ be the collection of all sets that are $({\T}$-)\emph{admissible}, i.e., 
\be
\label{adm}
\mathcal{S}_{\T} =\{S \subset {\T} \colon\,\exists\,F \text{ such that } h(F)=S\},
\ee
where $h(F) = \cup_{x \in F} B(x)$ is the halo of $F$. In Section~\ref{SS:admissible} we will see that there is a unique maximal $F$ such that $h(F)=S$, which we denote by $S^-$ and which equals $S^- = \{x \in S\colon\,B(x) \subset S\}$. 

Obviously, not every closed set is admissible. For example, when we form $1$-halos we round off corners, and so a shape with sharp corners cannot be in $\mathcal{S}$. Also note that  $S^-\neq\emptyset$ whenever $S$ is admissible: $S$ necessarily contains at least one unit disk $B(x)$ with $x\in S$. In the following, we typically omit the subscript referring to the torus $\T$. 


\paragraph{Large deviation principles.} 
Define 
\be \label{jhalo}
J(S) =|S|- \kappa |S^-|, \qquad S\in \mathcal{S},
\ee
and 
\be \label{ihalo} 
I(S) = J(S) - \inf_ {\mathcal S} J. 
\ee
We view the halo $h(\gamma)$ as a random variable with values in the space $\mathcal S$, topologized with the Hausdorff distance. Note that $\inf_ {\mathcal S} J = -(\kappa-1) |\mathbb T|$ because $\kappa \in (1,\infty)$, $S\subset \mathbb T$ for any $S\in \mathcal S$, and $\mathbb T^-=\mathbb T$.

\begin{theorem}[Large deviation principle for the halo shape] 
\label{thm:ldp-halo}
$\mbox{}$\\
The family of probability measures $(\mu_\beta(h(\gamma) \in \cdot\,))_{\beta \geq 1}$ satisfies the LDP on $\mathcal{S}$ with speed $\beta$ and with good rate function $I$.
\end{theorem}

\noindent
Informally, Theorem \ref{thm:ldp-halo} says that 
\be
\mu_\beta\bigl( h(\gamma) \approx S\bigr) \approx \exp\bigl( - \beta I(S)\bigr),
\qquad \beta \to \infty.
\ee

The contraction principle suggests a large deviation principle for the halo volume. To formulate this, we first state a minimisation problem. The condition $R\in (1, \frac{L}{\pi}  + \frac12 )$ below ensures that the periodic boundary conditions on the torus $\T$ are not felt.

\begin{theorem}[Minimisers of the rate function for the halo shape] 
\label{thm:isope}
$\mbox{}$\\
For every $R\in (1, \frac{L}{\pi}  + \frac12 )$,
\begin{itemize}
\item[{\rm (1)}]  
\be 
\label{isope}
\min\big\{|S|- \kappa |S^-|\colon\, S\in \mathcal{S}, |S|= \pi R^2\big\} = \pi R^2 - \kappa \pi (R-1)^2 
\ee
and the minimisers are the disks of radius $R$. 
\item[{\rm (2)}] 
The minimisers are stable in the following sense: There exists an $\eps_0>0$ such that if $R-1\ge \eps_0$ and $S \in \mathcal S$ satisfies
\be
\label{E:Bonnassumpt}
\bigr(|S|-\kappa |S^-|\bigr) - \bigl( \pi R^2 -  \kappa\pi(R-1)^2\bigr)  \leq \pi\kappa\eps 
\text{ with } |S| = \pi R^2 \text{ and  } \eps\in (0, \eps_0),
\ee 
then $S^-$ is connected with connected complement (simply connected as a subset of $\R^2$), and
\be
\label{Bonnesen}
d_{\text{\rm{H}}}(\partial S,\partial B_R) \leq \sqrt{5R\eps},  
\ee
where $d_{\text{\rm{H}}}$ denotes the Hausdorff distance.
\end{itemize}
\end{theorem}

\noindent 
Theorem~\ref{thm:isope} is a powerful tool because it shows that \emph{the near-minimers of the halo rate function are close to a disk and have no holes inside}. In particular, it tells us that
\be
I(B_R) = \Phi_\kappa(R)+(\kappa-1)|\T|,
\ee
and allows us to describe the large deviations of the halo volume. 
 
\begin{theorem}[Large deviation principle for the halo volume] 
\label{thm:ldp-volume} 
$\mbox{}$\\
The family of probability measures $(\mu_\beta(V(\gamma) \in \cdot\,))_{\beta \geq 1}$ satisfies the LDP on $[0,\infty)$ with speed $\beta$ and with good rate function $I^*$ given by
\be \label{eq:ivol} 
I^*(A) = \inf \{I(S)\colon\, |S| = A\}, \qquad A \in [0,\infty).	
\ee
\end{theorem}	

\noindent
Informally, Theorem \ref{thm:ldp-volume} says that	
\be
\mu_\beta\bigl( V(\gamma) \approx A\bigr) \approx \exp\bigl( - \beta I^*(A)\bigr).
\ee	
For every $R\in (1, \frac{L}{\pi}  + \frac12 )$, we have 
\be
I^* (\pi R^2) = I(B_R). 
\ee
 

\subsection{Fluctuations of the critical droplet: moderate deviations} 
\label{zoomin}


\paragraph{Fluctuations of the halo volume.}

The function $R \mapsto I(B_R)$ is maximal at $\Rc = \Rc(\kappa)$, defined in \eqref{USkapdef}, that plays the role of the radius of the \emph{critical droplet} as $\beta\to\infty$. We  now  zoom in on a neighbourhood of the critical droplet. The large deviation principle yields the statement 
\be
\mu_\beta \Bigl( |V(\gamma) - \pi \Rc^2| \leq \eps \Bigr) 
= \exp\Bigl( - \beta \min_{ {A \in [0,\infty):} \atop {|A- \pi \Rc^2|\leq \eps} } I^*(A)  
+ o(\beta) \Bigr), \qquad \beta \to \infty,
\ee
for $\eps>0$ fixed. We would like to take $\eps = \eps(\beta) \downarrow 0$, for which we need a stronger property. This is going to be the statement of our main theorem. Let us set
\be
\label{E:Gkappa}
G_\kappa = \frac{\kappa^{2/3}}{\kappa-1},
\ee
and let $\tau^* \in \R$ be the unique solution to the equation 
\be 
\label{tstar1}
\int_0^\infty \sqrt{2\pi u} \exp\Bigl(-\tau^* u - \frac{u^3}{24}\Bigr)\, \dd u =1. 
\ee
It is easy to show that $\tau^*>0$. By numerically approximating the integral we find that $\tau^*\in [1.60,1.61]$.

The following result provides rigorous bounds on the cost of the moderate deviations.

\begin{theorem}[Moderate deviation bounds] 
\label{thm:zoom1}
For $C$ large enough, 
\be
\begin{aligned}
\limsup_{\beta\to \infty}\frac{1}{ \beta^{1/3}} \log \Bigl\{ \e^{\beta I(B_{\Rc})}
\mu_\beta \Bigl( |V(\gamma) - \pi \Rc^2| \leq C \beta^{-2/3}\Bigr) \Bigr\} 
& \leq 2\pi G_\kappa \tau^*,\\
\liminf_{\beta\to \infty}\frac{1}{ \beta^{1/3}} \log \Bigl\{ \e^{\beta I(B_{\Rc})}
\mu_\beta \Bigl( |V(\gamma) - \pi \Rc^2| \leq C \beta^{-2/3}\Bigr) \Bigr\} 
&\geq 0.
\end{aligned}
\ee 
\end{theorem} 

We conjecture that the bounds in Theorem \ref{thm:zoom1} can be sharpened so that they match. To formulate this, let us first introduce, for $(p,q) \in \R \times \R$, the integral operator $\textbf{K}_{p,q}$ acting from $L^2(\R_+ \times \R)$ into itself (with respect to Lebesgue measure) given by 
\be\label{kpoper}
(\mathbf K_{p,q} f)(\xi) = \int_{\R_+ \times \R} \d\xi' \,K_{p,q}(\xi,\xi') f(\xi'). 
\ee
The integral kernel $K_{p,q}\colon\,(\R_+ \times \R)^2 \to [0,\infty)$ is given by
\be \label{kpdef}
\begin{aligned}
&K_{p,q}\bigl((x_1,y_1),(x_2,y_2)\bigr)\\ 
&= \exp\Bigl( \frac12 \Bigl[ p x_1+qy_1 - \frac{x_1^3}{24} - \frac{y_1^2}{2x_1}\Bigr]\Bigr) 
\, \1_{\big\{ \frac{2}{x_1+x_2}(\frac{y_{2}}{x_{2}} - \frac{y_1}{x_1})<1\big\}}
\, \exp\Bigl( \frac12 \Bigl[ p x_2+qy_2 - \frac{x_2^3}{24} - \frac{y_2^2}{2x_2}\Bigr]\Bigr). 
\end{aligned}
\ee
Let $\lambda(p,q)$ be the principal eigenvalue of the operator $\mathbf K_{p,q}$.

\begin{conjecture}[Moderate deviations] 
\label{thm:zoom2}
For $C$ large enough and $\beta\to\infty$,
\be
\mu_\beta \Bigl( |V(\gamma) - \pi \Rc^2|\leq C \beta^{-2/3}\Bigr) 
= \e^{- \beta I(B_{\Rc}) + \beta^{1/3}\Psi(\kappa) + o(\beta^{1/3})}, 
\ee
where 
\be \label{eq:scaling}
I(B_{\Rc}) = \Phi(\kappa) + (\kappa-1)|\T|, \qquad 
\Phi(\kappa) = \frac{\pi\kappa}{\kappa-1}, \quad	
\Psi(\kappa) = 2\pi G_{\kappa}(\tau^*- p^*).
\ee
Here $\tau^*$ is defined in \eqref{tstar1} and  $p^*$ is the unique solution of the equation $\lambda(p,0)=1$. 
\end{conjecture} 

\begin{remark}[Three conditions]
{\rm Conjecture \ref{thm:zoom2} provides a sharp asymptotics up to second order. In Section \ref{conds} we prove that this conjecture is true subject to \emph{three conditions} related to the \emph{microscopic} fluctuations of the surface of the critical droplet. It turns out that these fluctuations are related to a microscopic model that we call the \emph{parabolic interface model} (PIM), consisting of a sequence of interacting downward parabolas arranged in linear order and pinned at both ends of the sequence. We analyse this model and settle the three conditions in \cite{dHJKP3}.} 
\end{remark}

\begin{remark}[Metastability]
{\rm As shown in \cite{dHJKP1}, Conjecture~\ref{thm:zoom2} also provides a detailed description of the \emph{mesoscopic} fluctuations of the surface of the critical droplet in the Widom-Rowlinson model in the \emph{metastable regime} \eqref{metaregalt}, in terms of a certain constrained Brownian bridge and quantifies the cost of moderate deviations for the surface free energy of droplets, and makes rigorous the heuristic arguments for capillary waves put forward in Stillinger and Weeks~\cite{SW*}.}
\end{remark}

\begin{remark}[Broader context]
{\rm Our results open up a new window in the area of \emph{stochastic geometry for interacting particle systems}.
\begin{itemize}
\item[(a)] 
There is a large literature in stochastic geometry on fluctuations in point processes. General overviews are Chiu, Stoyan, Kendall and Mecke~\cite{CSKM} and Georgii, H\"aggstr\"om and Maes~\cite{GHM}. Schreiber~\cite{Sc2,Sc3} derives limit laws for high-intensity point processes and extremal points. Schreiber and Yukich~\cite{SY} and Calka, Schreiber and Yukich~\cite{CSY}) focus on \emph{random convex polytopes} and their relation to the \emph{paraboloid growth process}. While we marginally touch on these concepts in the present paper, they play an important role in \cite{dHJKP3}, where we discuss the microscopic fluctuations of the surface of the critical droplet. There we prove that, upon rescaling of the random variables describing the boundary of the critical droplet, the effective microscopic model is given by a modification of the \emph{paraboloid hull process}, which is connected to the \emph{paraboloid growth process}. 
\item[(b)] 
The literature on stochastic interfaces is extensive, especially for phase boundaries separating two coexisting phases. In statistical mechanics, interface analyses have been successfully carried out for various lattice models. For the two-dimensional Ising model, Higuchi~\cite{H} proved that, in a properly chosen diffusive scaling limit, the interface converges to a Brownian bridge. This result was later adapted by Higuchi, Murai and Wang~\cite{HMW} to the two-dimensional lattice version of the two species Widom-Rowlinson model.
\item[(c)] 
Several variations on the Widom-Rowlinson model have been considered in the literature. Examples are Baddeley and van Lieshout~\cite{BvL} and Kendall, van Lieshout and Baddeley~\cite{KvLB}, who introduced \emph{quermass interaction processes} by adding to the energy $E(\gamma)$ defined in \eqref{Udef} terms like the perimeter length of the halo $h(\gamma)$ defined in \eqref{halodef} or other Minkowski functionals, and Georgii and H\"aggstr\"om \cite{GH}, who added a short-range repulsive pair interaction between the particles. 
\end{itemize}
}
\end{remark}


\subsection{Outline of the paper}
\label{outline}

Section~\ref{sec:ldp} is devoted to the proof of the two large deviation principles in Theorems~\ref{thm:ldp-halo} and \ref{thm:ldp-volume},  
as well as to the identification of the minimisers of the rate function $S \mapsto |S| - \kappa |S^-|$, under the conditions $S \in \mathcal{S}$ and $|S|= \pi R^2$, and their stability, as stated in Theorem~\ref{thm:isope}. Section~\ref{overallheur} provides the heuristics behind Theorem~\ref{thm:zoom1} and Conjecture~\ref{thm:zoom2}, whose proof is carried out in Sections~\ref{app:geometry}--\ref{proofmoddev}, which form the technical core of the present paper. Section~\ref{app:geometry} focusses on approximations of certain key geometric functionals that are crucial for the analysis of the moderate deviations. Section~\ref{sec:surface} represents moderate deviation probabilities in terms of geometric surface integrals and introduces auxiliary random processes that are needed for the description of the fluctuations of the surface of the critical droplet. Section~\ref{sec:prep} contains various preparatory statements involving exponential functionals of the auxiliary random variables. Section~\ref{proofmoddev} uses these preparations, in combination with the geometric properties derived in Sections~\ref{app:geometry}--\ref{sec:prep}, to prove the moderate deviations for the halo volume close to the critical droplet stated in Theorem~\ref{thm:zoom1}. In the same section, some conditions are given under which Conjecture~\ref{thm:zoom2} is true. Appendix~\ref{appA} contains several lemmas that are needed in Section~\ref{sec:ldp} and that involve admissible sets, isoperimetric inequalities and stability of the minimisers, with the rate function playing the role of perimeter. The stability of the minimisers plays a crucial role in converting probabilistic restrictions about volumes of halos into statements about geometric distributions of particle configurations.


\section{Shape of the critical droplet} 
\label{sec:ldp}

In this section we prove Theorems~\ref{thm:ldp-halo}--\ref{thm:ldp-volume}. Section~\ref{SS:admissible} takes a closer look at the properties of admissible sets (Lemmas~\ref{lem:admissible}--\ref{lem:boundary}). Section~\ref{isoper} gives the proof of Theorem~\ref{thm:isope}. The proof requires two isoperimetric inequalities (Lemmas~\ref{lem:isopequiv}--\ref{L:isope}), which are analogues of the classical isoperimetric and Bonnesen inequalities, and imply that the minimisers of $I$ in \eqref{ihalo} are disks and that the difference of $I$ with its minimum can be quantified in terms of the Hausdorff distance to these disks. Section~\ref{LDP-WR} proves the large deviation principle for the centres of the unit disks in the Widom-Rowlinson model (Proposition~\ref{prop:ldp-centers}), and uses this to prove Theorems~\ref{thm:ldp-halo} and \ref{thm:ldp-volume}. 


\subsection{Properties of admissible sets} 
\label{SS:admissible}

Write
\be 
\label{F+def}
F^+ = F + B(0) = \bigcup_{x\in B(0)} (F+x) = h(F)
\ee
for the $1$-halo of $F\in\mathcal{F}_{\T}$ (Minkowski addition) and 
\be 
\label{F-def}
F^- = F \ominus B(0) = \bigcap_{x\in B(0)} (F+x) =\{x\in F\colon\, B(x)\subset F\},
\ee
for the 1-interior of $F\in\mathcal{F}_{\T}$ (Minkowski subtraction). In integral geometry, the sets $F^+$ and $F^-$ are called the \emph{dilation} and the \emph{erosion} of $F$, respectively. Note that the erosion and subsequent dilation of a set $F$ is contained in $F$, i.e.,
\be
\label{E:F-+}
(F^-)^+= \bigcup_{B(x) \subset F} B(x) \subset F
\ee
(Matheron~\cite[Section 1.5]{Ma}). See Lemma~\ref{lem:admissible}(\ref{F-+}) below. 

Note that $(F^-)^+$ is not necessarily equal to $F$. We use $\mathcal{S}\subset\mathcal{F}_{\T}$ to denote the collection of all sets for which the equality holds and call them (${\T}$-)\emph{admissible} (\emph{open with respect to $B(0)$} in the terminology of integral geometry), i.e.,
\be
\mathcal{S}_{\T} =\{S \subset {\T} \colon\,    (S^-)^+=S\}=  \{S \subset {\T} \colon\,  
S=(S\ominus B(0)) + B(0)\}. 
\ee
In the following, we typically omit the subscript referring to the torus $\T$, by writing $\mathcal{F}_{\T}=\mathcal{F}$, $\mathcal{S}_{\T}=\mathcal{S}$. There is another useful characterisation of admissible sets: $ S\in\mathcal{S}$ if and only if it is the 1-halo of some $F\in\mathcal{F}$, $S=F^+$ (see Lemma~\ref{lem:admissible}(\ref{adm-comp}) below). Obviously, not every closed set is admissible. For example, when we form 1-halos we round off corners, and so a shape with sharp corners cannot be in $\mathcal{S}$. Also note that  $S^-\neq\emptyset$ whenever $S$ is admissible: $S$ necessarily contains at least one unit disk $B(x)$ with $x\in S$.

In this section we summarise some known properties of admissible sets in a setting that will be needed later. The proofs of these properties rely on various sources. In the proof of Lemma~\ref{lem:admissible}, which is deferred to Appendix~\ref{appA}, we only quote appropriate references, and when instructive supply a short proof.

A key property is that for any set $S\in\mathcal S$ such that $S^-$ is connected and $S$ is simply connected, the set $S^-$ is of reach at least 1. Recall that the reach of a set $F\in\mathcal F$ is 
\be
\reach(F)=\sup\big\{r \geq 0\colon \text{ for any } x\in F+B_r(0) \text{ there exist a unique } 
y\in F\text{ nearest to } x\big\}.
\ee

\begin{lemma}
\label{lem:admissible}
$\mbox{}$
\begin{enumerate}[{\rm (1)}]
\item
\label{F-+}
If $F\in \mathcal F$, then $(F^-)^+\subset F$.
\item
\label{adm-comp}
$S\in\mathcal{S}$ if and only if $S$ is the $1$-halo of some $F\in\mathcal{F}$, i.e., 
\be
\{S \subset {\T} \colon\, S=(S^+)^-\} 
= \{S \subset {\T} \colon\, \exists\,F\in \mathcal{F} \text{ such that } S = F^+\}. 
\ee
\item
\label{contFtoF+}
Both $F \mapsto F^+ = h(F)$ and $F\mapsto |F^+| = |h(F)|$ are continuous with respect to the Hausdorff metric. 
\item
\label{connS-toS}
If $S\in\mathcal{S}$ and $S^-$ is connected, then also $S$ is connected. 
\item
\label{convS-toS}
If $F\in\mathcal{F}$ is convex, then $F^+$ and $F^-$ are convex and $F=(F^+)^-$.
If $F_1,F_2\in\mathcal{F}$ are convex and $F_1^+=F_2^+$, then also $F_1=F_2$. 
\item
\label{FtoS}
The set $\mathcal{S}$ is the closure in $\mathcal{F}$ of the set $\mathcal S^{\textrm{fin}}\subset \mathcal{S}$, where $\mathcal S^{\textrm{fin}}$ consists of all $S$ of the form $S=h(\gamma)$ with $\gamma \subset \T$ finite. 
\item
\label{Sreach} 
If $S\in\mathcal S$, then $\reach(S^-)\geq 1$, provided the following condition is satisfied:

\text{\rm (C)}\ \ \     
$S^-$ is connected and each component of  $\T\setminus S^-$ contains exactly one component of $\T\setminus S$.
\item
\label{noholes}
If $S\in\mathcal S$ and 
\be
\label{E:epsBonnesen}
d_{\text{\rm{H}}}(S,B_R(x)) \leq \eps
\ee
for a ball $B_R(x)$ with $x\in \T$, $\eps$  sufficiently small, and $R\ge 1+\frac{\eps}{1-2\eps}$,  then $d_{\text{\rm{H}}}(S^-,B_{R-1}(x)) \leq \eps$ and $S^-$ is connected and simply connected.
\item
\label{SLip} 
For any $S\in\mathcal S$ such that $\reach(S^-)>0$, the boundary $\partial S^-$ is $1$-rectifiable. If $S \in \mathcal S^{\textrm{fin}}$, then the boundary $\partial S^-$ is Lipschitz.
\end{enumerate}
\end{lemma}
 
We will also need the Steiner formula for sets of positive reach (Federer~\cite{Fe1}). 

\begin{lemma} 
\label{lem:boundary}
Let $S\in \mathcal{S}$ be an admissible set with $S^-$ of reach at least $1$. Then 
\be
|S \setminus S^-| = \mathcal S\mathcal M(S^-)+  \pi  \chi(S^-),
\ee
where $\mathcal S\mathcal M(S^-)$ is the outer Minkowski content of $S^-$ and $\chi(S^-)$ is the Euler-Poincar\'e characteristic of $S^-$ (= the number of connected components minus the number of holes). If the boundary $\partial S^-$ is Lipschitz, then $\mathcal S
\mathcal M(S^-) = \mathcal H^1(\partial S^-)$, where $\mathcal H^1$ is the $1$-dimensional Hausdorff measure. 
\end{lemma} 

\begin{proof}
Reformulating the Steiner formula for sets of positive reach as defined by Federer~\cite[Theorem 5.5, Theorem 5.19]{Fe1}, we get, for $S\subset \R^2$ and $S\in\mathcal S$,
\be
\label{E:Steiner}
\abs{S^-+B_r(0)}=\abs{S^-} +  \mathcal S\mathcal M({S^-})r +\chi(S^-)\abs{B_1(0)}r^2
\ee
for any $0<r<1$ and by continuity also for $r=1$. For continuity of the left-hand side, see Sz.-Nagy~\cite{Na}. The last claim is the same as Ambrosio, Colesanti and Villa~\cite[Corollary 1]{ACV}. 
\end{proof}


\subsection{Minimisers of the shape rate function and their stability} 
\label{isoper}

In this section we prove Theorem~\ref{thm:isope}.

\medskip\noindent
(1) The proof relies on the Brunn-Minkowski inequality and on Lemma~\ref{lem:isopequiv} below, which provides three reformulations of the isoperimetric inequality in \eqref{isope}. 

\begin{lemma} 
\label{lem:isopequiv}
Let $S \in \mathcal{F}$. If  $R>1$ and $|S|= \pi R^2$, then the following three statements are equivalent: 
\begin{itemize}
\item[{\rm (a)}] 
 $|S|- \kappa |S^-|\geq \pi R^2 - \kappa\pi (R-1)^2$. 	
\item[{\rm (b)}] 
$4 \pi |S| \leq (|S\setminus S^-|+ \pi )^2$. 
\item[{\rm (c)}] 
$4 \pi |S^-|\leq (|S\setminus S^-| - \pi)^2$.  
\end{itemize}
Moreover, equality holds in (a), (b), and (c) simultaneously, or in none. 
\end{lemma}
 
\begin{proof} 
The equivalence of (b) and (c) is an immediate consequence of the fact that $|S| = |S^-|+ |S\setminus S^-|$. For the equivalence of (a) and (b), we observe that $|S|- \kappa |S^-|= \kappa |S\setminus S^-|-(\kappa-1)|S|$ and therefore, with $|S| = \pi R^2$, (a) is equivalent to 
\be
|S \setminus S^-| \geq \pi R^2 - \pi (R-1)^2 = 2 \pi R -  \pi. 
\ee
We add $\pi $ to both sides and take the square to find that (a) is equivalent to (b). 
\end{proof} 

\begin{proof}[Proof of Theorem~\ref{thm:isope}]\

\noindent
(1) Armed with Lemma~\ref{lem:isopequiv}, we employ the Brunn-Minkowski inequality 
\be
\label{E:BM}
\abs{F+B}^{1/2}\ge \abs{F}^{1/2} + \abs{B}^{1/2},
\ee
which is valid for any non-empty measurable $F$, $B$ and $F+B$ (see Lusternik~\cite{Lu} and Federer~\cite[3.2.41]{Fe2}). Indeed, \eqref{E:BM} with $B=B(0)$ implies
\be
\label{E:F+}
\abs{F^+} - \abs{F} \ge 2\abs{F}^{1/2} \pi^{1/2}+ \pi,
\ee
and yields inequality (c) with $F=S^{-}$ and $F^+=S$, and thus also \eqref{isope} by (a). In $\R^2$, the equality in \eqref{E:F+} occurs only if $F=S^{-}$ is a disk or a point (see Burago and Zalgaller~\cite[Section 8.2.1]{BZ}). 

The role of the condition $R\in (1, \frac{L}{\pi}  + \frac12)$ is to disqualify the alternative of a strip $S^-=a\times L$ with $a=\frac{\pi R^2}L -2$ wrapped around the torus. It yields $|S| = \pi R^2$ as well as the minimum $|S|- \kappa |S^-|= -(\kappa-1) \pi R^2 +2\kappa L< \pi R^2 - \kappa \pi (R-1)^2$ once $R> \frac{L}{\pi} + \frac12$.

\medskip\noindent
(2) We first prove that if $S$ is close to a minimiser, then necessarily $S^-$ is connected and simply connected.

\begin{lemma}
\label{L:isope}
There exist  a function $\eps\mapsto\xi(\eps)$ satisfying $\lim_{\eps \downarrow 0} \xi(\eps)= 0$ such that if $S \in \mathcal S$ satisfies \eqref{E:Bonnassumpt} with $R-1 \geq \max\{\eps_0,\frac{\xi(\eps)}{1-2\xi(\eps)}\}$, then
\be
\label{E:xiBonnesen}
d_{\text{\rm{H}}}(S,B_R) \leq \xi(\eps)
\ee
for sufficiently small $\eps$,  and $S^-$ is connected and simply connected.
\end{lemma}

\noindent
The proof is given in Appendix~\ref{appA}. 

Now, to finish the proof of (2) we use that $S^-$ is connected and simply connected once \eqref{E:Bonnassumpt} is satisfied and $R-1\geq \frac{\xi(\eps)}{1-2\xi(\eps)}$. For the latter, similarly as in the proof of Lemma~\ref{lem:admissible}(\ref{Sreach}), we may assume that $S \in S^{\textrm{fin}}$. For fixed $R \in (1,\tfrac{L}\pi +\frac12)$ we choose $\eps_0$ such that $R-1\geq \frac{\xi(\eps)}{1-2\xi(\eps)}$ for any $0<\eps\leq\eps_0$. Using now that $\reach(S^-) \geq 1$ according to Lemma~\ref{lem:admissible}(\ref{Sreach}), we will rely on the Bonnesen inequality, which is more precise than the provisional claim in \eqref{E:xiBonnesen} with the bound $\xi(\eps)$ whose dependence on $\eps$ is not explicitly specified. For connected and simply connected $S^-$, its boundary $\partial S^-$ is a Jordan curve and according to the Bonnesen inequality the difference of the radii $r_\mathrm{out}(\partial S^-)$ and $r_\mathrm{in}(\partial S^-)$ of the outer and the inner circle of the curve $\partial S^-$ can be bounded in terms of the isoperimetric defect $\mathcal H^1(\partial S^-)^2-4\pi \abs{S^-}$:
\be 
\label{Bon}
\pi^2 \big(r_\mathrm{out}(\partial S^-)-r_\mathrm{in}(\partial S^-)\big)^2 
\leq \mathcal H^1(\partial S^-)^2-4\pi \abs{S^-}.
\ee
Assuming that $|S|=\pi R^2$, we can use \eqref{E:defect}  combined with the Steiner formula (Lemma~\ref{lem:boundary}),
\be
\label{Bon1}
\big(|S|-\kappa|S^-|\big) - \big(\pi R^2 - \kappa\pi (R-1)^2\big)
= \kappa \big[\mathcal H^1(\partial S^-) - 2\pi(R-1)]
\ee
and reformulate the isoperimetric defect as
\be
\label{Bon2}
\mathcal H^1(\partial S^-)^2-4\pi |S^-| 
= \mathcal H^1(\partial S^-)^2 - 4\pi[|S|-\pi-\mathcal H^1(\partial S^-)] 
= (\mathcal H^1(\partial S^-)+2\pi)^2 - (2\pi R)^2.
\ee
If the left-hand side of \eqref{Bon1} is $\leq \pi\kappa\eps$, then the right-hand side of \eqref{Bon2} is bounded from above by
\be
\pi \eps \left( \pi \eps +4\pi R\right)
= \pi^2 \eps\,  (4R+\eps)<  5\pi^2 R \eps.
\ee
It follows from \eqref{Bon} that 
\be
r_\mathrm{out}(\partial S^-)-r_\mathrm{in}(\partial S^-) \le \sqrt{5 R\eps}.
\ee
Since $r_\mathrm{out}(\partial S)-r_\mathrm{in}(\partial S) = r_\mathrm{out}(\partial S^-)-r_\mathrm{in}(\partial S^-)$, we get the claim in \eqref{Bonnesen}.  
\end{proof}


\subsection{Large deviation principle for the Widom-Rowlinson model}
\label{LDP-WR}

In this section we prove Theorems~\ref{thm:ldp-halo} and~\ref{thm:ldp-volume}, and the claim made in \eqref{partsum}.

Recall that $\mu_\beta$ denotes the Gibbs measure \eqref{gibbs1alt} at inverse temperature $\beta$ and activity $z = \kappa\, z_t(\beta)$. This is a probability measure on the space $\Gamma$ of particle configurations, which we may view as a subset of $\mathcal F$ equipped with the Hausdorff topology. By a slight abuse of notation, we identify $\mu_\beta$ on $\Gamma$ with the measure on $\mathcal F$ supported on $\{h(\gamma)\colon\,\gamma\in\Gamma\}$. Theorem~\ref{thm:ldp-halo} builds on the large deviation principle for the Gibbs measure $\mu_\beta$ itself, i.e., the large deviation principle for the set of particle locations. The LDP for $\mu_\beta$ is summarised in the following proposition (recall that a rate function is called good when it is lower semi-continuous and has compact level sets). This proposition is in the spirit of Schreiber~\cite{Sc1}, \cite[Theorem 1]{Sc2}. The latter is stated in a slightly different setting, but the main ideas of the proof carry over.

\begin{proposition}[Large deviation principle for the Widom-Rowlinson model]  
\label{prop:ldp-centers}
The family of probability measu\-res $(\mu_\beta)_{\beta \geq 1}$ on $\mathcal{F}$, supported on $\Gamma\subset \mathcal F$, satisfies the LDP with rate $\beta$ and good rate function $I^\mathrm{WR}$ given by
\be 
I^\mathrm{WR}=J^\mathrm{WR}-\inf_{\mathcal{F}} J^\mathrm{WR},
\quad J^\mathrm{WR}(F) =  |F^+|- \kappa |F|, \quad F \in \mathcal{F}.
\ee
\end{proposition} 

\begin{proof} 
Let $\Pi_{\kappa\beta}$ be the homogeneous Poisson point process on $\T$ with intensity $\kappa\beta$. We may view $\Pi_{\kappa\beta}$ as a random variable on a probability space $(\Omega,\mathcal{B},\P)$ taking values in $\mathcal{F}$. Since $\P(\Pi_{\kappa\beta} \subset F) = \P(\Pi_{\kappa\beta} \cap ({\T}\setminus F) = \emptyset) = \e^{-\kappa\beta|{\T} \setminus F|}$, $F \in \mathcal{F}$, it follows that the family $(\Pi_{\kappa\beta})_{\beta \geq 1}$ satisfies the large deviation principle with rate $\beta$ and good rate function $I(F) = \kappa |{\T} \setminus F|$, $F\in\mathcal{F}$. Note that $F\mapsto |F|$ is upper semi-continuous (Schneider and Weil~\cite[Theorem 12.3.6]{SW}), but not continuous: sets $F$ of positive measure can be approximated by finite sets, which have measure zero. It follows that $F \mapsto I(F) = \kappa(|{\T}|- |F|)$ is lower semi-continuous, but not continuous. Nevertheless, the map $F\mapsto |F^+| = |h(F)|$ is continuous with respect to the Hausdorff metric (see Lemma~\ref{lem:admissible}(\ref{contFtoF+})). Therefore, since
\be
\label{Gibbsid} 
\mu_\beta(\mathcal{A}) 
= \frac{\e^{-|{\T}|}}{\Xi_{\beta}}\,\e^{\kappa \beta |{\T}|} \E \left(\e^{-\beta |h(\Pi_{\kappa \beta})|} 
\1_{\{ \Pi_{\kappa \beta} \in \mathcal{A} \}} \right)
\quad \text{ for measurable } \mathcal{A} \subset \mathcal{F},  
\ee
the claim follows from the LDP for the Poisson point process $(\Pi_{\kappa\beta})_{\beta \geq 1}$ and Varadhan's lemma.  
\end{proof}

With the help of Proposition~\ref{prop:ldp-centers}, the proof of Theorems~\ref{thm:ldp-halo} and~\ref{thm:ldp-volume} becomes straightforward via the contraction principle.

\begin{proof}[Proof of Theorem~\ref{thm:ldp-halo}]
As mentioned above, the map $\mathcal{F}\to \mathcal{S}$ defined by $F\mapsto S = F^+$ is continuous with respect to the Hausdorff metric. Proposition~\ref{prop:ldp-centers} and the contraction principle therefore imply that the LDP for the law of $h(\gamma)$ under $\mu_\beta$ holds with rate $\beta$ and good rate function $I$ given by 
\be \label{eq:ldp_functional}
I=J-\inf_{\mathcal{F}} J, 
\quad J(S) = \inf \big\{ |F^+|- \kappa |F|\colon\,F\in \mathcal{F}, 
F^+=S\big\}, \quad S\in\mathcal S.
\ee
We show that $J(S) = |S|-\kappa |S^-|$. Indeed, if $F^+=S$, then $F\subset S^-$ and $|F^+| - \kappa |F|\geq |S|- \kappa |S^-|$ yielding $J(S) \geq |S|-\kappa |S^-|$. On the other hand, taking $F = S^-$ with  $F^+= (S^-)^+ = S$ in view of admissibility of $S$, we get  $J(S) \leq |F^+|- \kappa |F| = |S| - \kappa |S^-|$.
\end{proof}

\begin{proof}[Proof of Theorem~\ref{thm:ldp-volume}]
Theorem~\ref{thm:ldp-volume} follows from Proposition~\ref{prop:ldp-centers}, the continuity of the map $F\mapsto |F^+|= |h(F)|$, and the contraction principle. The rate function $I^*$ is given by 
\be
\begin{aligned}
I^*(A) &= \inf \{ I(S)\colon\ S\in \mathcal S,\, |S| = A\}\\ 
&= \inf\{|S|- \kappa|S^-|\colon\, S\in \mathcal S,\, |S| = A \} 
- \inf\{|S|- \kappa |S^-|\colon\, S \in \mathcal S \}.
\end{aligned}
\ee
In the difference of the two infima, the first infimum (when $A=\pi R^2$ with $R\in (1,\frac{L}{\pi}+\frac12 )$) is equal to $\pi R^2 - \kappa \pi (R-1)^2$ by Theorem~\ref{thm:isope}(1). For the second infimum, we note that 
\be
|S|- \kappa |S^-| \geq (1- \kappa)|S|\geq (1-\kappa) |\mathbb T|
\ee
with equality for $S = \mathbb T$.
\end{proof} 

We conclude the section with the following lemma.

\begin{lemma}[Asymptotics of the partition function]
\label{lem:partfun}
Let $\Xi_{\beta}$ be the partition function defined in \eqref{Xidefalt}. Then
\be
\Xi_{\beta} = \e^{[(\kappa-1)\beta-1] |{\T}|}(1+o(1)), \qquad \beta\to\infty.
\ee
\end{lemma}

\begin{proof}
From \eqref{Gibbsid}, by picking $\mathcal{A}=\mathcal{F}$, we get 
\be
\Xi_{\beta}= \e^{-|{\T}|}\,\e^{\kappa \beta |{\T}|} \E \left(\e^{-\beta |h(\Pi_{\kappa \beta})|}  \right).
\ee
Therefore it suffices to show that
\be
 \E \left(\e^{\beta | \T \setminus h(\Pi_{\kappa \beta})|}\right) = 1+o(1).
\ee
We split the expectation as
\be
\label{eq:three_for_exp}
\E \left(\e^{\beta |\T \setminus h(\Pi_{\kappa \beta})|}\right) =  1- \P\left(h(\Pi_{\kappa \beta}) \neq \T \right)
+ \E \left(\e^{\beta | \T \setminus h(\Pi_{\kappa \beta})|} \1_{\{ h(\Pi_{\kappa \beta}) \neq \T\}} \right).
\ee
The second term in the right-hand side of \eqref{eq:three_for_exp} is small. Indeed, this is immediate from Proposition \ref{prop:ldp-centers}, since there exists a non-empty $F \in \mathcal F$ such that
\be
\P\left(h(\Pi_{\kappa \beta}) \neq \T\right) = \P\left(\Pi_{\kappa \beta} \subset F\right)
= \e^{-\kappa\beta|\T\setminus F|},
\ee
which is equal to $\e^{-\kappa\beta c}$ for some $c>0$.
 The third term in the right-hand side of \eqref{eq:three_for_exp} is also small. Indeed, from Theorem~\ref{thm:ldp-halo} and \eqref{eq:ldp_functional} we obtain that
\be
\E \left(\e^{\beta |\T \setminus h(\Pi_{\kappa \beta})|} \1_{\{h(\Pi_{\kappa \beta}) \neq \T\}} \right)
= \exp\left(\beta\sup_{F\colon\,F^+\,\neq\,\T} \left\{|\T\setminus F^+|-\kappa |\T\setminus F|\right\} + o(\beta)\right).
\ee
Therefore, setting $F = S^-$ with $F^+= (S^-)^+ = S$ in view of admissibility of $S$, we see that it suffices to show that there exists a $C>0$ such that
\be\label{eq:goal_now}
(|S| - \kappa |S^-|) - (|\T| - \kappa |\T|) \geq C \qquad \forall\,\mathcal S \ni S \neq \T.
\ee
To that end, we observe that if $ S \neq \T$, then $|\T \setminus S^-| \geq \pi$. 
Indeed, if $S\neq \T$, then there must be at least one point $p \in \T\setminus S$, and hence the disk of radius $1$ centered at $p$ is a forbidden region for the centers, i.e., is contained in $\T\setminus S^-$. 
 The claim in \eqref{eq:goal_now} now follows by noting that if $\mathcal S \ni S \neq \T$, then
\be
(|S|-\kappa|S^-|) - (|\T|-\kappa |\T|) = |S \setminus S^-| + (\kappa-1)|\T\setminus S^-|
\geq (\kappa-1) \pi
\ee
which completes the proof.
%
%
\end{proof}


\section{Intermezzo: heuristics for fluctuations of the critical droplet} 
\label{overallheur}

In this section we provide the main ideas behind the proof of Theorem~\ref{thm:zoom1} given in Sections~\ref{app:geometry}--\ref{proofmoddev}. While these ideas are heuristic in nature, they offer valuable insight into the proof of the moderate deviation bounds, which is both lengthy and complex. In Section~\ref{sec:heur-surfred} we explain how the moderate deviation probability for the halo volume can be expressed in terms of a certain \emph{surface integral}. In Section~\ref{approxsurface} we explain how the weight in this surface integral can be approximated in terms of the \emph{polar coordinates} of the boundary points. In Section~\ref{orders} we provide a quick guess of what the \emph{orders of magnitude} of the angles and the radii of the boundary points are as $\beta\to\infty$. In Section~\ref{sec:heur-aux} we introduce \emph{auxiliary random processes} that allow us to transform the surface integral into an expectation of a certain exponential functional, capturing the \emph{global} (= mesoscopic) scaling of the boundary of the critical droplet. In Section~\ref{effint} we perform a further change of variable to rewrite the expectation in terms of an \emph{effective interface model}, capturing the \emph{local} (= microscopic) scaling of the boundary of the critical droplet. 


\subsection{Reduction to a surface integral} 
\label{sec:heur-surfred}

The starting point for the proof of Theorem~\ref{thm:zoom1} is the following. In view of the large deviation principles in Theorems~\ref{thm:ldp-halo} and~\ref{thm:ldp-volume} and the quantitative isoperimetric inequality~in Theorem~\ref{thm:isope}, the dominant contribution to the event $|V(\gamma) - \pi \Rc^2|\leq C\beta^{-2/3}$ should come from approximately disk-shaped halos (``droplets''). 

Consider the event 
\be 
\label{eventDeps}
\cD_\eps(x) = \Bigl\{\gamma \in \Gamma\colon\, d_\mathrm{H}(\partial h(\gamma), \partial B_{\Rc}(x))\leq\eps\Bigr \}
\ee
that $h(\gamma)$ is close to a disk $B_{\Rc}(x)$ (with no holes). Because of translation invariance, we may focus on $\cD_\eps(0)$. For $\gamma\in \cD_\eps(0)$, the boundary $\partial h(\gamma)$ of the droplet is a union of circle arcs centred at points $z_1,\ldots, z_n\in \gamma$, called \emph{boundary} points. Each boundary point is \emph{extremal} in the sense that $h(\gamma\setminus x) \subsetneq h(\gamma)$. We call a collection of points $\z=\{z_1,\ldots,z_n\}$ a \emph{connected outer contour} if there exists a halo $S$ with a simply connected $1$-interior $S^-$ having exactly these boundary points. The halo $S$, if it exists, is unique, and we denote it by $S(\z)$. The set of connected outer contours is denoted by $\mathcal O$. 

\begin{figure}[htbp]
\vspace{-0.2cm}
\begin{center}
\includegraphics[width=7cm]{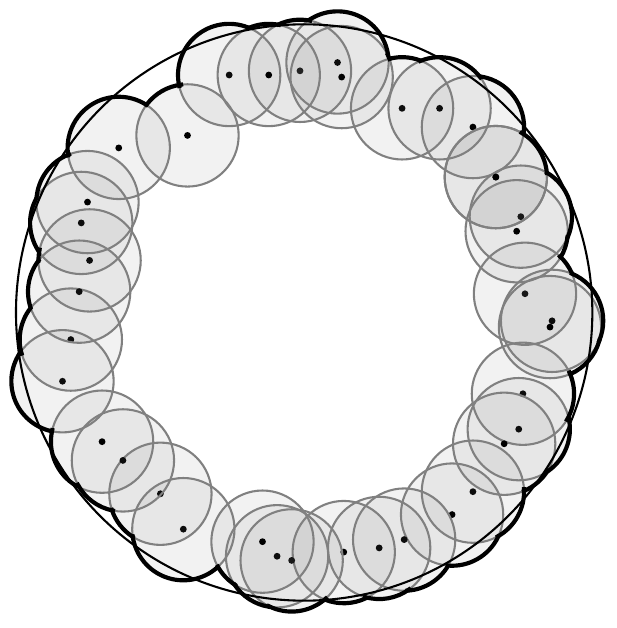}
\raisebox{-0.54cm}{\includegraphics[width=8.08cm]{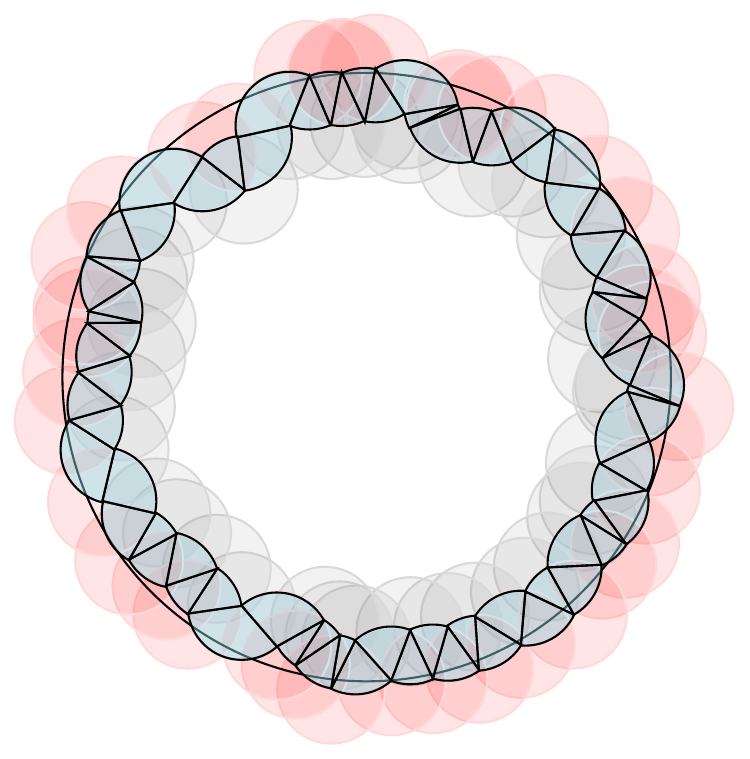}}
\end{center}
\caption{\small The set  $\z(\gamma)$ consisting of the boundary points of the configuration $\gamma$ from Fig.~\ref{fig:critdrop} is shown on the left. The thick line is the  boundary $\partial h(\gamma)$ of the halo $h(\gamma)$. On the right, the boundary layer $S(\z)\setminus S(\z)^-$ is shown in blue with the outer boundary $\partial h(\gamma)$ the same as on the left picture and the inner boundary $\partial h(\gamma)^-$ consisting of circular arcs of unit disks in pale red with centres at the cusps of the outer boundary.}
\label{fig:cleansausage}
\end{figure}

For $\gamma \in \cD_\eps(0)$, both $h(\gamma)$ and $V(\gamma)$ are uniquely determined by the boundary points, since $h(\gamma) = S(\z)$ and $V(\gamma) = |S(\z)|$. Given the above equivalence, it is natural to use the same notation $\cD_\eps(0)$ for both $\gamma \in \cD_\eps(0)$ and the corresponding $\z \in \cD_\eps(0)$. It also makes sense to compute probabilities by conditioning on the boundary points. Abbreviate 
\be 
\label{Ezdefs}
\!\!\!\! \cH(\z) = \bigl(|S(\z)| - \kappa |S(\z)^-|\bigr) - \pi\bigl(\Rc^2 - \kappa (\Rc-1)^2\bigr)
= \Delta(\z) - (2\pi \Rc-\pi)  +(1-\kappa)(|S(\z)^-|-\pi(\Rc-1)^2),
\ee
where $\Delta(\z)=\abs{S(\z)\setminus S(\z)^-}$. We will see that the following is true: For all measurable set $\mathcal{A} \subset \mathcal S$, as $\beta \to \infty$,
\begin{equation}
\label{eq:surface}
\begin{aligned}
&\mu_\beta \bigl( h(\gamma) \in \mathcal{A}, \gamma \in \cD_\eps(0)\bigr)\\
&\geq \bigl(1-\pi(\Rc-2+\eps)^2(\kappa\beta)^{2/3} \e^{-\kappa\beta}\bigr)\,
e^{-\beta I^*(\pi \Rc^2)} 
\sum_{n \in \N_0} \frac{(\kappa\beta)^n}{n!} \!\int_{\mathbb T^n} \dd \z\,
 \e^{-\beta \cH(\z)} \1_{\{ S(\z)\in\mathcal{A}\}} \1_{\cD_\eps(0)}(\z). 
\end{aligned}
\end{equation}
Here, $I^*$ is the rate function defined in \eqref{eq:ivol}. In view of this geometric constraint, the only contributions to the right-hand side of \eqref{eq:surface} are from connected outer contours $\z\in \mathcal O$ that lie in an annulus: $|z_i - (\Rc-1)|\leq \eps$, $1 \leq i \leq n$. Hence we may think of~\eqref{eq:surface} as a \emph{surface integral}. 


\subsection{Approximation of the surface term}
\label{approxsurface}

In view of~\eqref{eq:surface}, our next task is to evaluate $\cH(\z)$. We choose polar coordinates for the boundary points and write 
\be 
\label{polar}
z_i = ( r_i \cos t_i, r_i \sin t_i), \qquad 1 \leq i \leq n.
\ee
Upon relabelling the centres, we may without loss of generality assume that $0\leq t_1\leq \cdots \leq  t_n < 2\pi $. We set $t_{n+1} =t_0 + 2\pi$ and $r_{n+1} = r_1$, and define angular increments 
\be 
\label{thetaj}
\theta_i = t_{i+1}- t_i, \qquad 1 \leq i \leq n. 
\ee
Note that $\theta_i\geq 0$ and $\sum_{i=1}^n \theta_i = 2\pi$. We will show that the function $\cH(\z)$ admits an expansion after we put 
\be 
\label{rhoi}
\rho_i= r_i - (\Rc-1), \qquad \bar \rho_i =\frac{\rho_i + \rho_{i+1}}{2}. 
\ee
Namely,
\begin{equation} 
\label{eq:cH}
\cH(\z) = \frac{\kappa-1}{2}\sum_{i=1}^n \Bigl\{ \frac{(\rho_{i+1} - \rho_i)^2}{\theta_i} 
- {\bar \rho_i}^2\theta_i \Bigr\} +C_1\sum_{i=1}^n \theta_i^3  + \text{error terms},
\end{equation}
as summarised in Proposition~\ref{prop:expansion}, with
\be
C_1= \frac{\Rc^2(\Rc-1)}{24} = \frac{\kappa^2}{24\, (\kappa-1)^3} =\, \frac1{24}\, G_\kappa^3,
\ee
where we use that $\kappa-1 = 1/(\Rc-1)$ and $G_\kappa$ is defined in \eqref{E:Gkappa}.

In the following, in analogy with $\z$, we use bold $\r$, $\t$, $\rr$ and $\tt$, and later also $\s$, $\x$, $\y$, $\vp$,and $\vt$, for the corresponding vectors  $(r_1,\dots,r_n), \dots, (\vartheta_1,\dots,\vartheta_n)$.


\subsection{Orders of magnitude}
\label{orders}

Neglecting higher order terms in $\cH(\z)$ in~\eqref{eq:cH}, we see that the weight $\exp(-\beta\cH(\z))$ involves several terms. The factor $\exp(-\beta C_1\theta_i^3)$ suggests that the typical angular increment $\theta_i$ is of order $\beta^{-1/3}$, and that the typical number of boundary points $n$ is of order $\beta^{1/3}$. The factor 
\be 
\label{gaussian}
\exp\Bigl( -\beta \frac{\kappa-1}{2}\frac{(\rho_{i+1}- \rho_i)^2}{\theta_i}\Bigr)
\ee
suggests that $\rho_{i+1} - \rho_i$ is approximately normal with variance proportional to $\theta_i /\beta$. Hence, we expect that the radial increment $\rho_{i+1}- \rho_i$ is of order $\beta^{-2/3}$. Combining these observations, we expect that 
\be
\beta\sum_{i=1}^n \Bigl\{  \frac{\kappa-1}{2} \frac{(\rho_{i+1} - \rho_i)^2}{\theta_i} +  C_1 \theta_i^3\Bigr\} 
\approx \mathrm{const}\, \beta^{1/3}, 
\ee
which explains the exponent $\beta^{1/3}$ in Theorem~\ref{thm:zoom1}. Furthermore, it will be natural to think of $\rho_i$ as 
\be
\rho_i = \frac{m +  B_{t_i}}{\sqrt{(\kappa - 1)\beta}}
\ee
with $m$ some unknown mean value and $(B_{t})_{t\geq 0}$ the mean-centred Brownian bridge (see \eqref{Btildedef} below). The consistency with the guessed order of magnitude of the radial increment is guaranteed by the fact that $B_{t_{i+1}} - B_{t_i} \approx \sqrt{\theta_i} \approx \beta^{-1/6}$ and the observation that $\beta^{-1/2-1/6}= \beta^{-2/3}$. Finally, we note that 
\be
\beta \frac{\kappa -1}{2} \sum_{i=1}^n \bar \rho_i^2 \theta_i \approx \pi m^2 + \frac12 \int_0^{2\pi} B_t^2 \dd t, 
\ee
which should not contribute on the scale $\beta^{1/3}$ we are interested in (unless $m$ is large). Nevertheless, we will need to treat this term carefully, because
\be
\E\Bigl[\exp\Bigl( \frac12 \int_0^{2\pi} B_t^2 \dd t\Bigr)\Bigr] = \infty,
\ee
and extra arguments will be needed to cure this divergence. 

For later usage, let us also have a closer look at the volume constraint $|V(\gamma) - \pi \Rc^2| \leq C\beta^{-2/3}$. If we substitute the expansion of $V(\gamma)=\abs{S(\z)}$  (cf.\ Proposition~\ref{prop:expansion}) and neglect higher order terms, then the volume constraint becomes 
\be
\Bigl|  \sum_{i=1}^n \Bigl\{ \frac12 \bigl[ (\Rc+ \bar \rho_i)^2 - \Rc^2\bigr] \theta_i 
+ \frac12 \frac{(\rho_{i+1}- \rho_i)^2}{(\Rc - 1)\theta_i} - C_1 \theta_i^3  \Bigr\}   \Bigr| \lesssim C \beta^{-2/3}. 
\ee
Making a few leaps of faith, we may approximate
\begin{align}
\sum_{i=1}^n  \tfrac12 \bigl[ (\Rc+ \bar \rho_i)^2 - \Rc^2\bigr] \theta_i  	
&\approx \frac12 \int_0^{2\pi} \Bigl( \bigl( \Rc+ [(\kappa-1)\beta]^{-1/2} (m+B_t)\bigr)^2 
- \Rc^2\Bigr) \dd t \notag \\
& = \pi\Bigl( \bigl( \Rc + [(\kappa-1)\beta]^{-1/2} m \bigr)^2 - \Rc^2 \Bigr) 
+ \frac{1}{2(\kappa-1)\beta} \int_0^{2\pi} B_t^2 \dd t,
\end{align}
where we use that $\int_0^{2\pi} B_t \dd t =0$. From the considerations above we should expect the sum overall to be of order $\beta^{-2/3}$. Hence $[(\kappa-1)\beta]^{-1/2} m$ should also be of order $\beta^{-2/3}$, i.e., $|m| = O(\beta^{-1/6})$. Later we will only prove that $|m| = O(\beta^{1/6})$, but this will turn out to be enough for our purpose. 


\subsection{Global scaling: auxiliary random processes} 
\label{sec:heur-aux}

If we substitute the approximation~\eqref{eq:cH} for $\cH(\z)$ into the surface integral in~\eqref{eq:surface} and drop error terms and indicators, we are naturally led to the investigation of expressions of the type
\begin{multline} 
\label{heursurf}
\sum_{n\in\N_0} (\kappa \beta)^n \int_{[0,2\pi)^n} \d \t\,\,\1_{\{t_1\leq \cdots \leq t_n\}} 
\int_{\R^n}\dd \rr\,\prod_{i=1}^n (\Rc-1+\rho_i)^n \\
\times \exp\Bigl( - \beta (\kappa-1) \Bigl\{\sum_{i=1}^n \frac{(\rho_{i+1}-\rho_i)^2}{2\theta_i} 
- \sum_{i=1}^n \tfrac12 {\bar \rho}_i^2 \theta_i\Bigr\} - \beta C_1 \sum_{i=1}^n \theta_i^3 \Bigr) f(\z),
\end{multline}
where $f$ is a non-negative test function, and we recall~\eqref{polar} and~\eqref{rhoi} (by convention the summand with $n=0$ equals $1$). The Gaussian term (see also \eqref{gaussian}) is conveniently expressed with the heat kernel 
\be 
\label{heatkernel}
P_\theta(x-y) = \frac{1}{\sqrt{2\pi \theta}} \exp\Bigl( - \frac{(x-y)^2}{2\theta}\Bigr),
\ee
and we have 
\be
\begin{aligned}
&\exp\Bigl( - \beta (\kappa-1) \sum_{i=1}^n \frac{(\rho_{i+1}-\rho_i)^2}{2\theta_i} 
- \beta C_1 \sum_{i=1}^n \theta_i^3 \Bigr) \\
&\qquad = \prod_{i=1}^n P_{\theta_i}\Bigl( \sqrt{(\kappa -1)\beta}(\rho_{i+1}- \rho_i)\Bigr) 
\times  \prod_{i=1}^n  \sqrt{2\pi \theta_i}\,  \e^{- \beta C_1 \theta_i^3}.
\end{aligned}
\ee
Let us approximate $\Rc-1+\rho_i \approx \Rc-1$, drop the term $\sum_i {\bar \rho_i}^2\theta_i$, and change variables as $x_i = \sqrt{(\kappa -1)\beta}\, \rho_i$. Then the integral in~\eqref{heursurf} becomes
\be 
\label{heursurf2}
\sum_{n\in\N_0} \Bigl(\frac{\kappa \beta(\Rc-1)}{\sqrt{(\kappa-1)\beta}}\Bigr)^n 
\int_{[0,2\pi)^n} \d \t \,\,\1_{\{t_1\leq \cdots \leq t_n\}} \int_{\R^n}\dd \x\, 
\prod_{i=1}^n \Bigl( P_{\theta_i} (x_{i+1}- x_i) \sqrt{2\pi \theta_i}\,  \e^{- \beta C_1 \theta_i^3} \Bigr) f(\z).
\ee
With the help of
\be
\frac{\kappa \beta(\Rc-1)}{\sqrt{(\kappa-1)\beta}}\, \sqrt{2\pi \theta_i} 
= \frac{\kappa \sqrt{\beta}}{(\kappa -1)^{3/2}}\, \sqrt{2\pi \theta_i} 
= \sqrt{\beta}\, G_\kappa^{3/2} \, \sqrt{2\pi \theta_i} 
= \beta^{1/3} G_\kappa \sqrt{2\pi G_\kappa \beta^{1/3} \theta_i},
\ee
and $C_1 = G_\kappa^3/24$, the expression in \eqref{heursurf2} can be rewritten as 
\be 
\label{heursurf2alt}
\sum_{n\in\N_0} \bigl( \beta^{1/3} G_\kappa\bigr) ^n 
\int_{[0,2\pi)^n} \d \t \,\,\1_{\{t_1\leq \cdots \leq t_n\}}
\int_{\R^n}\dd \x\, \prod_{i=1}^n \Bigl( P_{\theta_i} (x_{i+1}- x_i) \sqrt{2\pi \beta^{1/3} 
G_\kappa \theta_i}\,  \e^{- \frac{1}{24} \beta G_\kappa^3 \theta_i^3} \Bigr) f(\z).
\ee
This expression motivates the auxiliary processes introduced in Section~\ref{sec:auxiliary}. Moreover, $\beta$ and $\kappa$ only enter in the combination $\beta^{1/3}G_\kappa$ (except possibly in the test function $f$), which explains the scaling of the surface corrections in Theorem~\ref{thm:zoom1}.

The picture that emerges of the droplet boundary is that its deviation from $\partial B_{\Rc}(0)$ should be of the order of $\beta^{-1/2}$, and that the boundary points are obtained by selecting points of a Gaussian bridge process according to an angular point process. We may call this picture \emph{mesoscopic}, since it describes the overall shape of the droplet: the Gaussian process $(B_t)_{t\in [0,2\pi]}$ does not see the \emph{microscopic} details. 

 
\subsection{Local scaling: effective interface model}
\label{effint}

Equation~\eqref{heursurf} suggests one last change of variables, namely, set 
\be 
\label{localscaling}
s_i = \beta^{1/3} G_\kappa t_i,\quad \varphi_i = \sqrt{\beta^{1/3} G_\kappa}\, x_i,
\quad \vartheta_i = s_{i+1}- s_i = \beta^{1/3} G_\kappa\, \theta_i.
\ee
Note that
\be 
\label{rhoxphi}
\rho_i = \frac{x_i}{\sqrt{(\kappa -1)\beta}} = \frac{\varphi_i}{\beta^{2/3} \kappa^{1/3}}.
\ee
Then~\eqref{heursurf2} becomes 
\be 
\label{heursurf3}
\sum_{n\in\N_0} \int_{[0,2\pi G_\kappa \beta^{1/3})^n} \dd \s \,\,\1_{\{s_1\leq \cdots \leq s_n\}}
\int_{\R^n}\dd \vp\, \prod_{i=1}^n \Bigl( P_{\vartheta_i} (\varphi_{i+1}- \varphi_i) 
\sqrt{2\pi \vartheta_i}\,  \e^{- \frac{1}{24}  \vartheta_i^3} \Bigr) f(\z),
\ee
or, equivalently, 
\be 
\label{heursurf4}
\sum_{n\in\N_0} \int_{[0,2\pi G_\kappa \beta^{1/3})^n} \dd \s \,\,\1_{\{s_1\leq \cdots \leq s_n\}}
\int_{\R^n}\dd \vp\, \exp\Biggl(- \sum_{i=1}^n \frac{(\varphi_{i+1}- \varphi_i)^2}{2\vartheta_i} 
- \sum_{i=1}^n\frac{\vartheta_i^3}{24}\Biggr) f(\z).
\ee
Let us finally return to the term $\sum_{i=1}^n \bar \rho_i^2 \theta_i$ that we had dropped from~\eqref{heursurf}. By  \eqref{localscaling} and \eqref{rhoxphi}, we have 
\be 
\label{eq:interface-background}
\beta\, \frac{\kappa -1}{2} \sum_{i=1}^n \bar \rho_i^2 \theta_i 
= \frac{1}{2G_\kappa^2 \beta^{2/3}} \sum_{i=1}^n \bar \varphi_i^2 \vartheta_i.
\ee
Taking this term into account, we see that \eqref{heursurf4} should be replaced by the more accurate integral
\be \label{heursurf5}
\int_{\R^n} \dd \vp \int_{[0,2\pi G_\kappa \beta^{1/3})^n} \dd \s\,\,\1_{\{s_1\leq \cdots \leq s_n\}}
\,\exp\Biggl( \frac{1}{2 G_\kappa^2 \beta^{2/3}} 
\sum_{i=1}^n \bar \varphi_i^2 \vartheta_i -  \sum_{i=1}^n 
\Bigl\{ \frac{(\varphi_{i+1} - \varphi_i)^2}{2\vartheta_i} 
+ \frac{\vartheta_i^3}{24}\Bigr\}  \Biggr) f(\z).
\ee
We may view the exponential, together with an additional indicator for the boundary points, as the Boltzmann weight for an effective interface model, which is studied in detail in a forthcoming paper \cite{dHJKP3}. The term $\frac{1}{2 G_\kappa^2 \beta^{2/3}} \sum_{i=1}^n \bar \varphi_i^2 \vartheta_i$ plays the role of a background potential. 

Having thus explained the \emph{heuristics} behind the proof of Theorem~\ref{thm:zoom1}, we are now ready to start with the \emph{technical core} of the paper, which is collected in Sections~\ref{app:geometry}--\ref{proofmoddev}.


\section{Stochastic geometry I: approximation of geometric functionals} 
\label{app:geometry}

This section collects a number of geometric facts that will be needed for the moderate deviations of the halo volume. In Section~\ref{sec:apriori} we prove a number of \emph{a priori estimates} on the radial and the angular coordinates of the \emph{boundary points}, i.e., the centres of the unit disks that lie at the boundary of the critical droplet (Lemma~\ref{L:boundaryorder}, Proposition~\ref{prop:apriori}, Definition~\ref{def:yzsums}, Corollary~\ref{cor:apriori} and Lemma~\ref{ordermagn1}). These estimates play a crucial role for the arguments in Sections~\ref{sec:surface}--\ref{proofmoddev}. In Section~\ref{bounddet} we show that the set of boundary points allows for a local characterisation, in the sense that whether or not a unit disk touches the boundary of the critical droplet only depends on the centre of the two neighbouring unit disks (Definition~\ref{def:bd}, Lemma~\ref{L:local} and Proposition~\ref{prop:locality}). In Section~\ref{sec:taylor} we use the results in Sections~\ref{sec:apriori}--\ref{bounddet} to derive an approximation for the \emph{volume} and the \emph{surface} of halos that are close to a critical disk, in terms of certain sums involving the radial and the angular coordinates of the boundary points (Proposition~\ref{prop:expansion}). In Section~\ref{app:cenmass}, finally, we define the geometric center of the droplet and derive \emph{a priori estimates} on its location (Definition~\ref{def:bdalt} and Lemma~\ref{lem:aprvolsurcen}).


\subsection{A priori estimates on boundary points}
\label{sec:apriori}

Theorem \ref{thm:isope} can be applied to sets of the form  $S = h(\gamma)$, the halo of the configuration $\gamma$. In particular, the condition that the boundary $\partial S$ is close to a disk $B_R$ with $R>1$ is a strong restriction on the geometry of the boundary points $\z=(z_1,\ldots,z_n)$ (recall Fig.~\ref{fig:cleansausage}). In this section we collect several \emph{a priori estimates} and constraints that follow from the fact that $S = h(\gamma)$ has a simply connected 1-interior $S^-$ and $d_{\text{\rm{H}}}(\partial S,\partial B_R) \leq \eps$. Recall the notions of boundary points, the set of connected outer contours $\mathcal O$  introduced in Section~\ref{sec:heur-surfred}, and the polar coordinates $(r_i,t_i)_{i=1}^n$ as well as angular increments  $\theta_i$ from~\eqref{thetaj}. We write $\ell_{z_i}$ or $\ell_i$ to denote the ray from the origin passing through the point $z_i$, and $A_{R,\eps}$ and $A_{R-1,\eps}$ to denote the $\eps$-annuli defined as the closures of $B_{R+\eps}(0) \setminus B_{R-\eps}(0)$ and $B_{R-1+\eps}(0) \setminus  B_{R-1-\eps}(0)$, respectively.

First we show that the intersections of the boundary circles in $\partial S$ follow the order of the corresponding boundary points:

\begin{lemma} 
\label{L:boundaryorder}
Fix $R\in(1,\frac{L}{\pi}+\tfrac12)$. If $\z\in\cO$ and $d_\text{{\rm H}}(\partial S(\z), B_R(0)) \le \eps$, then as $\eps \downarrow 0$:
\begin{itemize}
\item[{\rm (a)}] 
$z_i\in A_{R-1,\eps}$ for all $1 \leq  i \leq n$.
\item[{\rm (b)}]
The distance between any two points $x,x'\in A_{R-1,\eps}$ such that $\partial B(x) \cap \partial B(x')\cap A_{R,\eps}\neq\emptyset$ satisfies
\be
\label{E:|z-z|}
\abs{x-x'} \le 4\sqrt{\,\frac{R-1}{R}}\,\eps^{1/2} -  \frac{2\sqrt{R-1}(R-2)}{R^{3/2}}\,\eps^{3/2}  
+ O(\eps^{5/2}),
\ee
and the angle $\theta_{xx'}$ between the rays $\ell_x$ and $\ell_{x'}$ satisfies
\be
\label{E:theta}
\abs{\theta_{xx'}}\le \frac{4}{\sqrt{R(R-1)}}\,\eps^{1/2} 
+  \frac{2\bigl(3R(R-1)+2\bigr)}{3(R(R-1))^{3/2}}\, \eps^{3/2} 
+ O(\eps^{5/2}).
\ee
\item[{\rm (c)}] 
For every $1 \leq i \leq n$ there exists a unique $v_i\in \partial B(z_i)\cap  \partial B(z_{i+1})$ such that $v_i\in  A_{R.\eps}$. 
(Such a $v_i$ is referred to as a boundary cusp.)
\item[{\rm (d)}]  
The boundary $\partial S(\z)$ consists of the union of closed arcs of the circles $\partial B(z_i)$ between the boundary cusps $v_i$ and $v_{i+1}$, $1 \leq i \leq n$, contained in $A_{R,\eps}$ (with $v_{n+1}=v_1$).
\end{itemize}
\end{lemma} 
\begin{proof}
The proof is based on a series of geometric observations.
 
\medskip\noindent
(a) This claim is immediate from the fact that $\dist_\text{{\rm H}}(\partial S(\z), B_R(0))\le \eps$ and that each $z_i$ is a boundary point with  $B(z_i)\subset B_{R+\eps}(0)$ and $\partial B(z_i)\cap A_{R,\eps} \neq\emptyset$.

\medskip\noindent
(b) Let $v\in \partial B(x)\cap\partial B(x')\cap A_{R,\eps}$. To get the bound in \eqref{E:|z-z|}, we note that the maximal distance between $x$ and $x'$ and the maximal angle $\theta$ consistent with the condition $v\in A_{R,\eps}$ and $x,x'\in A_{R-1,\eps}$ occur when $v \in \partial B_{R-\eps}(0)$ and $x,x'\in \partial B_{R-1+\eps}(0)$, i.e., when the ray $0\, v$ is orthogonal to the segment $z \, z'$. Assume, without loss of generality, that $v=(R-\eps,0)$ and $x,x'=(x_0,\pm y_0)$ with $x_0^2+y_0^2 =(R-1+\eps)^2$. Then the condition $v \in \partial B(x)\cap \partial B(x')$ reads 
\be
(R-\eps-x_0)^2+ y_0^2=1,
\ee
which, in combination with the equation $x_0^2+y_0^2=(R-1+\eps)^2$, yields $x_0=\frac{R(R-1) -\eps+\eps^2}{R-\eps}$ and implies 
\be
\max \abs{x-x'}^2 = 4y_0^2 =\frac{16(R-1)}{R}\eps -\frac{16(R-1)(R-2)}{R^2}\eps^2 
+ O(\eps^3),
\ee
which settles \eqref{E:|z-z|}. For the corresponding angle, we have $\abs{\tan\frac{\theta}2}\le \frac{y_0}{x_0}$, and hence $\abs{\theta} \le 2 \arctan(\frac{y_0}{x_0})$, which settles \eqref{E:theta}.

\medskip\noindent
(c) Let $\{v_i,v_i'\}=  \partial B(z_i)\cap  \partial B(z_{i+1})$  and consider the boundary piece $\partial (B(z_i)\cup  B(z_{i+1}))$. This consists of two arcs, $C_i\subset  \partial B(z_i))$ and  $C_{i+1}\subset  \partial B(z_{i+1}))$, both ending in the points $v_i$ and $v_i'$. A necessary condition for both $z_i$ and $z_{i+1}$ to be boundary points is that both arcs $C_i$ and $C_{i+1}$ intersect the annulus $A_{R,\eps}$. If, in addition, $v_i, v_i'\notin A_{R,\eps}$, then we get a contradiction with the assumption that both $z_i$ and $z_{i+1}$ are boundary points. Indeed, consider the line $\ell$ through $v_i'$ and  $v_i$ (and through $\overline{z_i}=\tfrac12(z_i+z_{i+1})$) and the intersection point $p_1= \ell\cap \partial B_{R-\eps}(0)$ as shown in Fig.~\ref{fig:v_i}. There exists a $j\neq i, i+1$ such that $p_1\in B(z_j)$. Otherwise there would be a gap in the boundary $\partial S(\z)$ along $\ell\cap A_{R,\eps}$. Assuming, without loss of generality, that the line segment $z_i\,z_j$ intersects the ray $\ell_{i+1}$ (in view of  \eqref{E:|z-z|}, this segment cannot intersect both $\ell_{i+1}$ and $\ell_{i-1}$), we conclude that if $p_1\in B(z_j)$, then also $p_2
\in B(z_j)$, where $p_2$ is the reflection of $p_1$ with respect to the ray $\ell_{i+1}$. But then also $ \partial B(z_{i+1})\cap A_{R,\eps} \subset B(z_j)$, which is in contradiction with the fact that $z_{i+1}$ is a boundary point. Note that there is a severe restriction on the position of the point $z_j$: it has to be contained in $B(p_1)$. The allowed region is shown in Fig.~\ref{fig:v_i} in a darker shade. Thus, necessarily, $v_i\in  A_{R,\eps}$, while $v_i'\in  B_{R-1-\eps}$, because $v_i'$ is a reflection of $v_i$ with respect to $\overline{z_i}\in A_{R-1,\eps}$.

\begin{figure}[htbp]
\centering
\begin{overpic}[width=8cm,scale=.25,percent]
{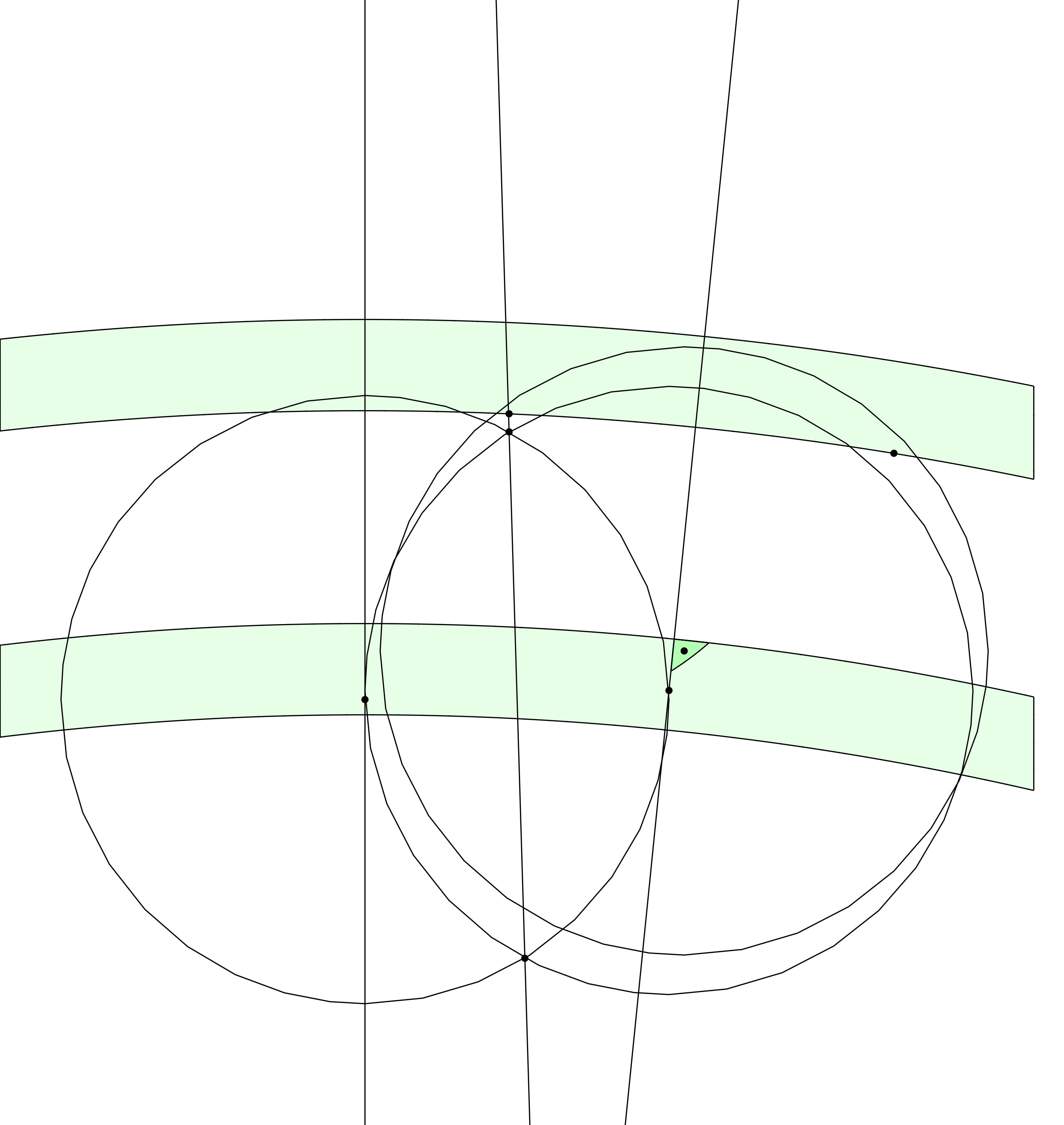}
\put(42,57) {$v_i$} 
\put(41.5,65.5) {$p_1$} 
\put(80.5,63) {$p_2$} 
\put(7,65) {$A_{R,\eps}$}
\put(7,39) {$A_{R-1,\eps}$} 
\put(27,38) {$z_i$} 
\put(51,37.5) {$z_{i+1}$} 
\put(62,39) {$z_j$} 
\put(43,11) {$v_i'$} 
\put(46,85) {$\ell$} 
\put(33.5,85) {$\ell_i$} 
\put(65.5,85) {$\ell_{i+1}$} 
\end{overpic}
\caption{\small The boundary points $z_i$ and $z_{i+1}$ and the intersection $v_i$ of their halo circles. The upper circle on the right has centre $z_j$. The dark shaded region is the intersection of the disk $B(p_1)$, the annulus $A_{R-1,\eps}$ and the halfplane to the right of $\ell_{i+1}$. The light shaded regions are the annuli $A_{R,\eps}$ and $A_{R-1,\eps}$.}
\label{fig:v_i}
\end{figure}

\medskip\noindent
(d) If the arc of the circle $\partial B(z_i)$ between the points $v_i$ and $v_{i+1}$ is intersected by a circle $\partial B(z_j)$ for some $j\notin\{i-1,i,i+1\}$, then necessarily $\{v_{i},v_{i+1}\}\cap B(z_j)\neq\emptyset$. Similarly as above, assuming that the line segment $z_i\,z_j$ intersects the ray $\ell_{i+1}$ and knowing that $v_i\in  A_{R,\eps}\cap B(z_j)$, we get that also its reflection with respect to the ray $\ell_{i+1}$ belongs to $B(z_j)$, which implies that $(\partial B(z_{i+1})\setminus B(z_i))\cap A_{R,\eps} \subset B(z_j)$, in contradiction with the fact that $z_{i+1}$ is a boundary point.
\end{proof}

Let
\be
\label{rhoidef}
\rho_i=  r_i - (R-1),
\ee 
and note that $r_{i+1}-r_i = \rho_{i+1}-\rho_i $. 
Abbreviate
\be
\label{baridef}
\begin{aligned}
&\overline{z_i} = \tfrac{1}{2}(z_i+z_{i+1}), \quad
\overline{\rho_i} = \tfrac{1}{2}(\rho_i+\rho_{i+1}), \\
&\overline{\rho_i\cos t_i} = \tfrac12(\rho_i\cos t_i+\rho_{i+1}\cos t_{i+1}), \quad
\overline{\rho_i\sin t_i} = \tfrac12(\rho_i\sin t_i +\rho_{i+1}\sin t_{i+1}).
\end{aligned}
\ee
Introduce two constants, $\overline{K}=\overline{K}(R) = \frac{4}{\sqrt{R(R-1)}}$ and $K=K(R) = 2\sqrt{R(R-1)}$.

\begin{proposition}[A priori estimates for angular and radial coordinates]
\label{prop:apriori}
$\mbox{}$\\
Fix $R\in(1,\frac{L}{\pi}+\tfrac12)$. For any $\z\in\mathcal O$ such that $d_{\text{\rm{H}}}(\partial S(\z),\partial B_{R}(0))\leq\eps$, and any $1 \leq i \leq n$,
\be
\label{apriori}
|\rho_i|\le \eps, \quad
 \theta_i \le \overline{K} \sqrt{\eps}, \quad 
\abs{\rho_{i+1}-\rho_i}/\theta_i \le K
 \sqrt{\eps}, \text{ and }
n^{-1}\le\frac{\overline{K}\sqrt{\eps}}{2\pi},
\ee
for sufficiently small $\eps$.
\end{proposition}
\begin{proof}

The first estimate in \eqref{apriori} is a trivial consequence of the inequality $d_{\text{\rm{H}}}(\partial S(\z), \partial B_{R}(0))\leq\eps$. The second estimate is the bound \eqref{E:theta} in Lemma~\ref{L:boundaryorder}(b). The fourth estimate is a consequence of the second estimate. Indeed,  $\frac{2\pi}{n} \leq  \max_{1\leq i\leq n} \theta_i\le \overline{K} \sqrt{\eps}$.

\begin{figure}[htbp]
\begin{overpic}[width=14cm,scale=.25,percent]
{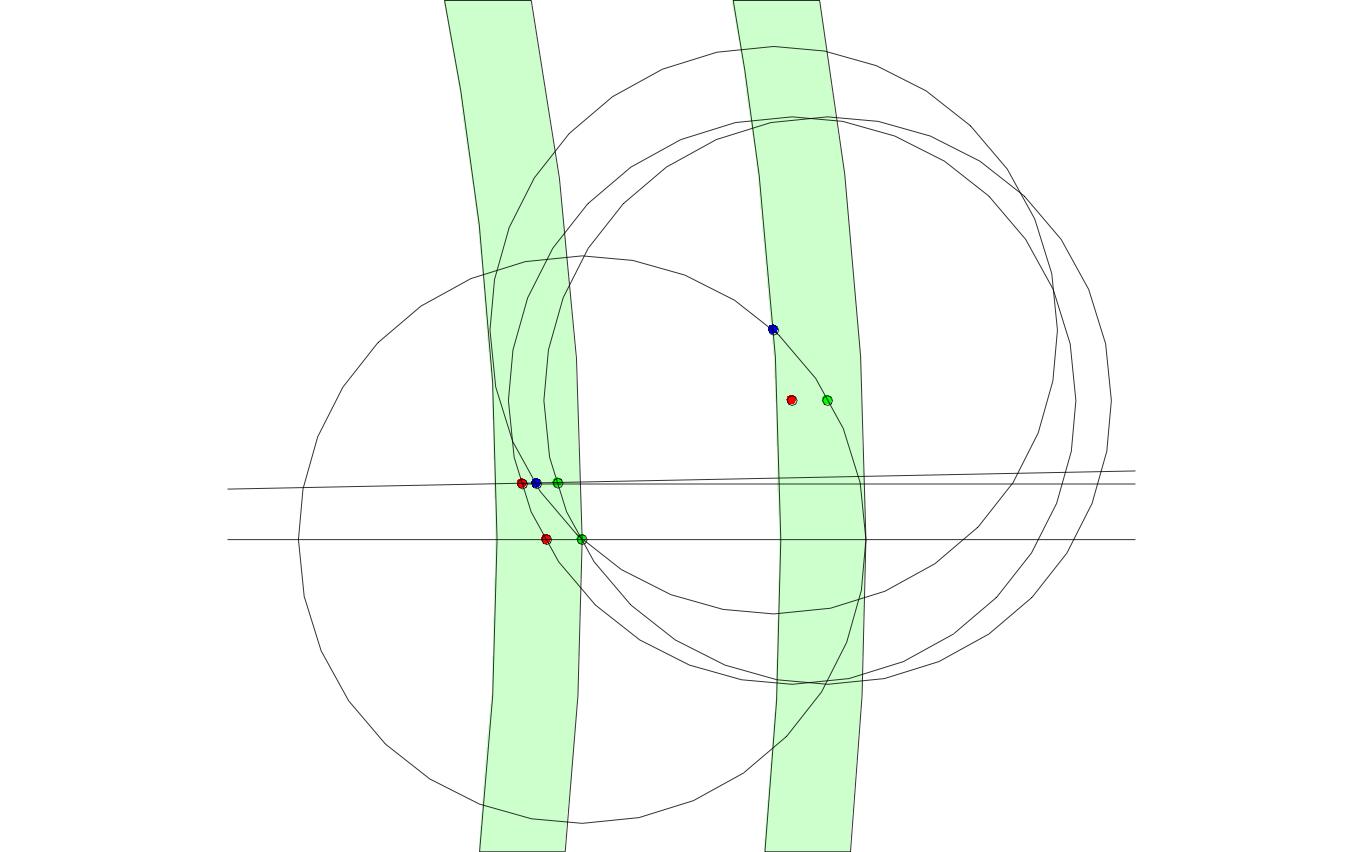}
\put(83,23.5) {$\ell_0$} 
\put(83,28.6) {$\ell_{\theta}$} 
\put(83.5,26) {$\ell_{\theta}'$} 
\put(33,56) {$A_{R-1,\eps}$} 
\put(56,56) {$A_{R,\eps}$}
\put(37.6,21.5) {$z_0$} 
\put(43.3,23.5) {$\bar z_0 $} 
\put(36.5,25.5) {$z_\theta $} 
\put(39,27.8) {$\tilde z_\theta $} 
\put(41.4,28) {$\bar z_\theta $} 
\put(41.5,25.8) {$\bar z'$} 
\put(57.5,34) {$v$}
\put(57.5,38) {$\tilde v$}
\put(60.5,34) {$\bar v$}
\put(43,13) {$\partial B(v)$}
\put(70,55) {$\partial B(\tilde v)$} 
\put(70,13) {$\partial B(\bar v)$} 
\put(16.7,13){$\partial B(\bar z_0)$}
\end{overpic}
\caption{\small The points $z_0, z_\theta \in  A_{R-1,\eps}$ lie on the rays $\ell_0$ and $\ell_\theta$, respectively. Assume that these correspond to nearest neighbours in $\z(\gamma)$. Then the intersection of the unit circles $\partial B(z_0)\cap \partial B(z_{\theta})$ contains a point $v\in A_{R,\eps}$ -- a boundary cusp -- such that $z_0, z_\theta\in \partial B(v)$ (all 3 points $z_0, z_\theta$, $v$ in red). Here, we illustrate the case when $v$ lies above $\ell_\theta$. Shift $\partial B(v)$ to $\partial B(\bar v)$ along the direction $\ell_0$, so that $z_0$ moves into $\bar z_0$ on $\partial B_{R-1+\eps}$. Then we get the point $\bar z_\theta = \partial B(\bar v) \cap \ell_\theta\cap A_{R-1,\eps}$ (all 3 points $\bar z_0, \bar z_\theta$, $\bar v$ in green), as well as $\bar z' = \partial B(\bar v)\cap \ell'\cap A_{R-1,\eps}$, which is indistinctly close to $\bar z_\theta$ such that $\bar r_0-r_0 = \abs{\bar z'-z_\theta}$. Finally, $\tilde z_\theta = \partial B(\tilde v)\cap \ell_\theta\cap A_{r-1,\eps}$ with $\tilde v\in \partial B(\bar z_0) \cap \partial B_{R-\eps}(0))$ (both in blue), which yields the upper bound on $r_0-r_\theta$ in the form $\bar r_0-\tilde r_\theta\le  K\theta\sqrt{\eps}$.}
\label{fig:mon}
\end{figure}

The third estimate is slightly more involved. Omit the index $i$, similarly as in the proof of Lemma~\ref{L:boundaryorder}(c), and consider points $z_0, z_\theta \in A_{R-1,\eps}$ with polar coordinates $(r_0,0)$ and $(r_\theta,\theta)$, respectively. Assuming that these correspond to neighbouring boundary points $z_i$ and $z_{i+1}$ from $\z(\gamma)$, we have that the intersection of the circles $\partial B(z_0)\cap \partial B(z_\theta)$ necessarily contains a boundary cusp $v\in \partial h(\gamma) \subset A_{R,\eps}$. Thus, $z_0, z_\theta\in \partial B(v)$. Without loss of generality, we may assume that the point $v$ belongs to the half-space with boundary $\ell_0$ and containing the ray $\ell_\theta$. Our aim is to maximise $\abs{r_0-r_\theta}$ with a fixed $\theta$.

First, consider the case when the point $v$ falls into the region of $A_{R,\eps}$ between the rays $\ell_0$ and $\ell_\theta$ with the difference $\abs{r_0-r_\theta}$ maximised when $v\in \ell_\theta$ or, symmetrically, $v\in\ell_0$. We begin by inspecting the particular instance of $v\in\ell_\theta$ with polar coordinates $(R-\eps,\theta)$. The circle $\partial B(v)$ intersects the rays $\ell_\theta$ and $\ell_0$ in points $z_\theta, z_0 \in A_{R-1,\eps}$ with radial coordinates $r_\theta=R-1-\eps$ and  $r_0$, respectively. To compute $r_0$, we rely on the condition $z_0\in\partial B(v)$ implying that $r_0$ is a root of quadratic equation $(r_0-(R-\eps)\cos\theta)^2 + ((R-\eps)\sin\theta)^2=1$ with the solution  
\be
\label{E:r_0}
r_0=(R-\eps-1)\bigl((1+\tfrac12(R-\eps)\theta^2+O(\theta^4)\bigr)= (1+\alpha) (R-\eps-1),
\ee
where we introduce the short-hand notation $\alpha= \tfrac12(R-\eps)\theta^2+O(\theta^4)$. By rescaling $r_0, r_\theta  \to r_0' = \frac{R-1+\eps}{r_0} r_0= R-1+\eps, r_\theta'=  \frac{R-1+\eps}{r_0} r_\theta$, we get that the difference $r_0-r_\theta=\alpha (R-1-\eps)$ increases to $r_0'-r_\theta'= \frac{R-1+\eps}{r_0} \alpha (R-1-\eps)=\frac{\alpha}{1+\alpha} (R-1+\eps)$, which in view of the second bound in \eqref{apriori} yields the necessary upper bound
\be
r_0-r_\theta\le \tfrac12(R-\eps)(R-1+\eps)\theta^2\le \tfrac12 R(R-1) \overline{K} \theta \sqrt{\eps} = K \theta \sqrt{\eps}.
\ee

Next, let us move to the case when the boundary cusp $v$ determined by $z_0$ and $z_\theta$ falls ``above'' $\ell_\theta$,  as illustrated in the Fig.~\ref{fig:mon}. Note that, necessarily, $r_0=\abs{z_0}\ge r_\theta=\abs{z_\theta}$. Let us shift the circle $\partial B(v)$ along the ray $\ell_0$ so that the point $z_0$ is shifted onto the point $\bar z_0\in \partial B_{R-1+\eps}(0)\cap \ell_0$. Use $\bar v$ for the centre of the shifted circle with $\bar v\in \partial B(\bar z_0)$. Further, consider the half-line $\ell'$ starting at $z_\theta$ parallel with $\ell_0$ and the points $\bar z_\theta=\ell_\theta\cap\partial B(\bar v)\cap A_{R-1,\eps}$ and $\bar z'=\ell'\cap\partial B(\bar v)\cap A_{R-1,\eps}$. Note that these two points are very close (diverging by the angle $\theta$ over the distance less than $2\eps$), $\abs{\bar z'-\bar z_\theta}\le \theta O(\eps)$. Also observe that the quadrangle formed by the points $z_0$, $\bar z_0$, $\bar z'$ and $ z_\theta$ is a parallelogram with $\bar r_0-r_0= \abs{\bar z'-z_\theta}=\abs{\bar v-v}$. Thus, $\bar r_0 -r_0=\bar r_\theta-  r_\theta+\theta O(\eps)$ implying $r_0-r_\theta = \bar r_0-\bar r_\theta+\theta O(\epsilon)$.

Finally, moving the point $\bar v$ on the circle $\partial B(\bar z_0)$ towards the point $\tilde v$ on the inner boundary of $A_{R,\eps}$, $\tilde v\in \partial B(\bar z_0)\cap \partial B_{R-\eps}(0))$, the point $\bar z_\theta \in \ell_{\theta}\cap \partial B(\bar v)$ moves on the ray $\ell_\theta$ in the direction of the origin to the point $\tilde z_\theta= \partial B(\tilde v)\cap \ell_\theta\cap A_{r-1,\eps}$, leading to the upper bound $r_0-r_\theta \le \bar r_0-\abs{\tilde z_\theta}+\theta O(\epsilon)$.

Recall that $\bar r_0=R-1+\eps$. It remains to compute $\tilde r_\theta=\abs{\tilde z_\theta}$. To compute $\abs{\tilde z_\theta}$, let us first get the point $\tilde v=(x,y)\in \partial B_{R-\eps}(0) \cap \partial B(\bar z_0)$ satisfying the equations
\be
\label{E:x^2+y^2}
x^2+y^2=(R-\eps)^2\  \text{ and }\  (x-(R-1+\eps))^2+y^2=1,
\ee
which yields 
\be
\label{E:x,y}
x=\frac{R(R-1)-\eps+\eps^2}{R-1+\eps}, 
\qquad y=\sqrt{(R-\eps)^2
- \Bigl(\frac{R(R-1)-\eps+\eps^2}{R-1+\eps} \Bigr)^2}
=2\frac{\sqrt{R(R-1)(1-\eps)\eps}}{R-1+\eps}.
\ee
Given that $\tilde z_\theta\in \partial B(\tilde v)$, we get that $\tilde v\in  \partial B(z_\theta)$ and thus 
\be
(x-\tilde r_\theta\cos\theta)^2+(y-\tilde r_\theta\sin\theta)^2=1.
\ee
In view of the first equation in \eqref{E:x^2+y^2}, this is equivalent to
\be
{\tilde r_\theta}^2-2(x\cos\theta+y\sin\theta)\tilde r_\theta+(R-\eps)^2-1=0.
\ee
Substituting for $x$ and $y$ from \eqref{E:x,y} and solving this equation and expanding into powers in $\eps$ and $\theta$, we get 
\be
\label{E:r-t}
\tilde r_\theta=\tilde r_\theta(R,\eps,\theta)= (R-1+\epsilon )-2 \sqrt{(R-1) R \epsilon} \theta + O(\theta\eps^{3/2})+ O(\theta^2),
\ee 
which yields $\bar r_0-\tilde r_\theta \le 2 \sqrt{(R-1)R\epsilon}\, \theta = K \theta \sqrt{\eps}$.
\end{proof}

In the sequel we will need six sums involving $\theta_i,\rho_i$, which we collect here.

\begin{definition}
\label{def:yzsums}
Fix $R\in(1,\frac{L}{\pi}+\tfrac12)$. Recall \eqref{polar} and \eqref{rhoidef}--\eqref{baridef}. Define
\be
\label{yzdef}
\begin{array}{llll}
&y_1(\z) = \sum\limits_{i=1}^n \theta_i^3, 
&y_2(\z) = \sum\limits_{i=1}^n (\rho_{i+1}-\rho_i)^2/\theta_i, 
&y_3(\z) = \sum\limits_{i=1}^n \overline{\rho_i}^2 \theta_i,\\[0.3cm]
&y_4(\z) = \sum\limits_{i=1}^n \overline{\rho_i}\,\theta_i,
&y_5(\z) = \sum\limits_{i=1}^n \theta_i\,\overline{\rho_i \cos t_i}, 
&y_6(\z) = \sum\limits_{i=1}^n \theta_i\,\overline{\rho_i \sin t_i}.
\end{array}
\ee
\end{definition}

\noindent
These expressions will appear in the expansions in Proposition~\ref{prop:expansion} and in Definition~\ref{def:bd} below. The following estimates will play a crucial role in the sequel. Note that $y_1(\z),y_2(\z),y_3(\z)$ are non-negative, while $y_4(\z),y_5(\z),y_6(\z)$ are not necessarily so.
 
\begin{corollary}[A priori estimates for sums in approximations]
\label{cor:apriori}
$\mbox{}$\\
Fix $R\in(1,\frac{L}{\pi}+\tfrac12)$. If $\z \in \mathcal{O}$ and $d_\text{{\rm H}}(\partial S(\z), \partial B_{R}(0)) \leq \eps$, then, as $\eps \downarrow 0$,
\be
\label{yzest}
0 \leq y_1(\z) \leq  2\pi{\overline K}^2 \eps, \ 
0 \leq y_2(\z) \leq  2\pi K^2 \eps,\ 
0 \leq y_3(\z) \leq  2\pi \eps^2,\ 
\abs{y_4(\z)}, \abs{y_5(\z)}, \abs{y_6(\z)} \leq 2\pi\eps.
\ee 
\end{corollary}

\begin{proof}
Using the bounds in \eqref{apriori}, we estimate
\be
\label{yzestalt}
\begin{array}{lll}
&0 \leq y_1(\z) \leq \bigl(\max\limits_{1 \leq i \leq n} \theta_i\bigr)^2 \sum\limits_{i=1}^n \theta_i \le 2\pi{\overline K}^2 \eps, 
&0 \leq y_2(\z) \leq \bigl(\max\limits_{1 \leq i \leq n} \frac{|\rho_{i+1}-\rho_i|}{\theta_i}\bigr)^2 \sum\limits_{i=1}^n \theta_i 
\le 2\pi K^2 \eps,\\ [0.3cm]
&0 \leq y_3(\z) \leq \bigl(\max\limits_{1 \leq i \leq n} |\rho_i|^2 \bigr) \sum\limits_{i=1}^n \theta_i \le 2\pi \eps^2,
&\abs{y_4(\z)}, \abs{y_5(\z)}, \abs{y_6(\z)}
\leq \bigl(\max\limits_{1 \leq i \leq n} |\rho_i|\bigr) \sum\limits_{i=1}^n \theta_i \le 2\pi\eps.
\end{array}
\ee 
\end{proof}

The following estimate will be needed later on.

\begin{lemma}
\label{ordermagn1}
$\sum_{i=1}^n  \theta_i\,\overline{\cos t_i} = O\left(\sum_{i=1}^n\theta_i^3\right)$.
\end{lemma}

\begin{proof}
We begin by writing
\be
\begin{aligned}
\sum_{i=1}^n  \theta_i\, \overline{\cos t_i}- \int_0^{2\pi} \d t\, \cos t 
= \sum_{i=1}^n \int_{t_i}^{ t_{i+1}} \d t\,\bigl(\overline{\cos t_i} - \cos t\bigr),
\end{aligned}
\ee
where in the left-hand side we have simply subtracted 0. Substituting $\tau =t-\bar t_i$, we get
\be
\begin{aligned}
&\int_{t_i}^{ t_{i+1}} \d t\,\bigl(\overline{\cos t_i} - \cos t\bigr)
= \int_{-{\frac{\theta_i}{2}}}^{\frac{\theta_i}{2}} \d\tau\,\Bigl(\tfrac12[\cos(\bar t_i+\tfrac{\theta_i}{2}) 
+ \cos(\bar t_i-\tfrac{\theta_i}{2})] - \cos(\bar t_i +\tau)\Bigr)\\
&= \int_{-\frac{\theta_i}{2}}^{\frac{\theta_i}{2}} \d\tau\,\Big(\cos\bar t_i\cos \tfrac{\theta_i}{2} 
- \cos(\bar t_i+\tau)\Big)
= \int_{-\frac{\theta_i}{2}}^{\frac{\theta_i}{2}} \d\tau\,\Big(\cos\bar t_i(\cos \tfrac{\theta_i}{2} 
- \cos \tau) +\sin\bar t_i\sin \tau\Big)\\
&= \cos\bar t_i \int_{-\frac{\theta_i}{2}}^{\frac{\theta_i}{2}} \d\tau\,(\cos \tfrac{\theta_i}{2} - \cos \tau)
= \cos\bar t_i \int_{-\frac{\theta_i}{2}}^{\frac{\theta_i}{2}} \left[\tfrac12(\tau^2-\theta_i^2) + O(\theta_i^4)\right]\, \d\tau
= O(\theta_i^3).
\end{aligned}
\ee
Here, when passing to the third line, we use that $\int_{-\theta_i/2}^{\theta_i/2} \sin\tau\, d\tau = 0$. 
\end{proof}


\subsection{Locality for boundary determination}
\label{bounddet}

In this section we present a crucial property of the boundary points, namely, their location is constrained only by the location of the two neighbouring boundary points (Proposition \ref{prop:locality} below). This property will be used in the proof of the lower bound in Theorem \ref{thm:zoom1} carried out in Section \ref{sec:wdmp}. 

\begin{definition}
\label{def:bd}
$\mbox{}$
\begin{itemize}
\item[{\rm (a)}] 
Let  $\z = (z_i)_{1 \leq i \leq n} = ((r_i\cos t_i,r_i\sin t_i))_{1 \leq i \leq n}$ be a sequence of points in $A_{R-1,\eps}$ in polar coordinates ordered by angle, i.e., $t_1<\cdots<t_n$ (and with $z_{n+1} = z_n$), such that for each pair $z_{i}, z_{i+1}$ there exists a $v_i\in \partial B(z_i) \cap \partial B(z_{i+1})\cap A_{R,\eps}$. Note that this condition implies that $S(\z)$ is well defined ($S(\z)$ is obtained by filling the inner part) and that $\partial S(\z)\subset A_{R,\eps}$. but does not necessarily imply that $\z\in\cO$, because possibly only a subset of $\z$ contributes to $\partial S(\z)$.
\item[{\rm (b)}] 
For any $1 \leq i,j \leq n$, let  $v_{i,j}\in \partial B(z_i)\cap  \partial B(z_j)$ be such that $\abs{v_{i,j}}=\max\{\abs{v}\colon\, v\in \partial B(z_i)\cap  \partial B(z_j)\}$ (if there is a tie, then take the intersection with minimal angle in polar coordinates). In polar coordinates, 
$v_{i,j}=(r_{i,j}\cos t_{i,j},r_{i,j}\sin t_{i,j})$, $r_{i,j}\in [R-\eps, R+\eps]$.
\begin{itemize}
\item[{\rm (i)}] 
A point $z_i$ is \emph{extremal} in $\z$ if $\partial B(z_i)  \cap \partial S(\z)\neq\emptyset$.
\item[{\rm (ii)}]  
A sequence $\z$ is a \emph{set of boundary points}, i.e., $\z\in \cO$, if each $z_i$, $1 \leq i \leq n$, are extremal.
\item[{\rm (iii)}] 
A triplet $(z_i,z_j,z_k)$ with $t_i<t_j<t_k$ is called an \emph{extremal triplet} if $( B(z_j)\setminus B_{R-\eps}(0))\setminus (B(z_i)\cup  B(z_k))\neq\emptyset$.
\end{itemize}
\end{itemize}
\end{definition}

We need the following lemma.

\begin{lemma}
\label{L:local}
Let $R\in(1+\frac{\eps}{1-2\eps},\frac{L}{\pi}+\tfrac12)$. Consider two points $x,x'\in A_{R-1,\eps}$. Set $\{v,v'\} = \partial B(x) \cap \partial B(x')$ and suppose that $\{v,v'\} \cap A_{R,\eps} \neq\emptyset$. Then the following hold:
\begin{itemize}
\item[{\rm (i)}] 
Exactly one of the vectors $\{v,v'\}= \partial B(x) \cap \partial B(x')$, say $v$, is in $A_{R,\eps}$. The other is in the interior of the ball $B_{R-\eps}(0)$.
\item[{\rm (ii)}]  
Let $x=(r\cos t, r\sin t)$, $x' = (r'\cos t', r'\sin t')$ and assume that $t<t'$. Let $H$ be the halfplane with the boundary consisting of the line $v v'$ containing the point $x'$. Then $B(x) \cap H \subset B(x')\cap H$.
\item[(iii)] 
A triplet $(z_i,z_j,z_k)$  with $v_{j,k} \in A_{R,\eps}$ is extremal if and only if $t_{i,j} <t_{j,k} $.
\end{itemize}
\end{lemma}

\begin{proof} 
$\mbox{}$

\begin{figure}[htbp]
\centering
\begin{overpic}[width=9cm,scale=.25,percent]
{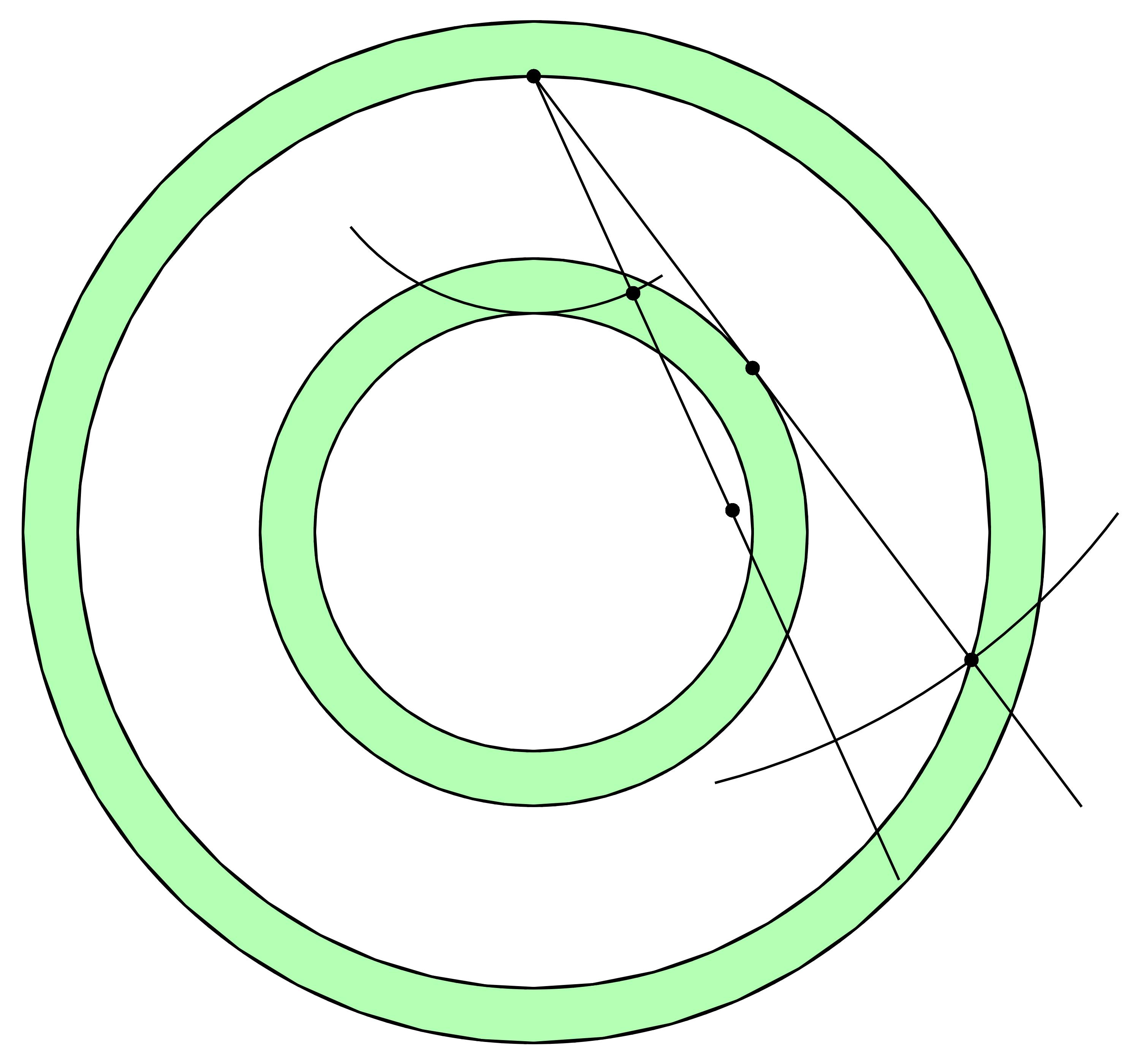}
\put(45,88) {$v$} 
\put(80,80) {$A_{R,\eps}$}
\put(14,60) {$A_{R-1,\eps}$} 
\put(25,75) {$\partial B(v)$} 
\put(60,72) {$\ell$} 
\put(68,61) {$\tau$} 
\put(80.5,35) {$\overline v$} 
\put(51,67.5) {$s$} 
\put(59,47.5) {$v'$} 
\put(60,20) {$\partial B_{2\ell}(v)$} 
\end{overpic}
\caption{\small Illustration of the fact that $v'\in B_{R-\eps}(0)$.}
\label{fig:volume}
\end{figure}

\medskip\noindent
(i) Choosing a position of $v$ in $A_{R,\eps}$, we see that the points $x$ and $x'$ necessarily lie on the arc $\partial B(v)\cap A_{R-1,\eps}$. Thus, the barycentre $s=\frac12(x+x')$, which is the centre of symmetry of the union $B(x)\cup B(x')$, is contained in  $\partial  B(v)\cap A_{R-1,\eps}$. The point $v'$ is symmetric to $v$ with respect to the barycentre $s$: $v'= s - (v-s)$. Suppose, without loss of generality, that $v=(0,y)$ with $y\in [R-\eps, R+\eps]$. To show that $v'\in B_{R-\eps}(0)$, consider the most extremal case $y=R-\eps$ when $\abs{v}=R-\eps$. Indeed, for $y>R-\eps$ we could shift the points $x$ and $x'$, and thus also $v$ and $v'$, by the vector $u=(0,R-\eps-y)$. The shifted $x+u, x'+u$ lead to the shifted $v+u$ and $v'+u$. Note that necessarily $x+u, x'+u\in A_{R-1,\eps}$ since $x+u, x'+u\in \partial B(v)\cap A_{R-1,\eps}+u\subset \partial B(v+u)\cap A_{R-1,\eps}$ in view of the fact that the point  $v+u-1=(0,R-1-\eps)\in  A_{R-1,\eps}$. Now, if $v'+u\in B_{R-\eps}(0)$, then necessarily also $v'\in B_{R-\eps}(0)$. Consider, in addition, the extreme case when $s$ is asymptotically approaching the extremal point $B(v)\cap \partial B_{R-1-\eps}(0)$. Consider the tangent line from $v$ to the disk $B_{R-1+\eps}(0)$ touching in the point $\tau$. This tangent line intersects the circle $B_{R-\eps}(0)$ in $v$ and a point $\tilde v$ symmetric with respect $\tau$. Clearly, if the distance $\ell$ from $v$ to $\tau$ is larger than $1$, then the point $v'$ on the line $v s$ falls short of $\partial B_{R-\eps}(0)$ and we get the claim. To show that $\ell > 1$, we compute it from the rectangular triangle $v \tau 0$:
\be
\ell^2=(R-\eps)^2-(R-1+\eps)^2=(2R-1)(1-2\eps).
\ee
This yields $\ell >1$ if and only if $R>1+\frac{\eps}{1-2\eps}$. 

\medskip\noindent
(ii) The claim immediately follows by inspecting the union $B(x)\cup B(x')$ with the intersection points $\{v,v'\}=\partial B(x) \cap \partial B(x')$ and noting the symmetry with respect to the barycentre $s$.

\medskip\noindent 
(iii) Just observe that the condition  $t_{i,j}=t_{j,k}$ means that the circle $\partial B(x_j)$ is touching the set $\partial B(x_i)\cup \partial B(x_k)$ in the point $v_{i,k}$. 
\end{proof}
	
\begin{proposition}[Local characterisation of sets of boundary points]
\label{prop:locality}
$\mbox{}$\\
Let $R\in(1+\frac{\eps}{1-2\eps},\frac{L}{\pi}+\tfrac12)$ and $\z = ((r_i\cos t_i,r_i\sin t_i))_{1 \leq i \leq n}$ be a sequence of points in $A_{R-1,\eps}$, ordered by angle. Then the following two conditions are equivalent:\\
(i) The set $\z$ is a set of boundary points, $\z\in\cO$.\\
(ii) Every triplet $(z_{j-1},z_{j},z_{j+1})$, $1 \leq j \leq n$, is extremal.
\end{proposition}  

\begin{proof}
$\mbox{}$

\medskip\noindent
(i) $\implies$ (ii):\\
If (ii) does not hold, then there exist a $j$ such that the triplet $(z_{j-1},z_{j},z_{j+1})$ is not extremal. According to Lemma~\ref{L:local}(iii), this implies that  $z_j$ is not extremal in $\z$ and that condition (i) is not satisfied. 

\medskip\noindent
(ii) $\implies$ (i):\\
If (i) does not hold, then there exist either $k>j$ or $i<j$ such that either $v_{j-1,j} \in  B(z_{k}) \cap A_{R,\eps}$ or $v_{j-1,j} \in B(z_{i}) \cap A_{R,\eps}$. Consider the former case. We will show that, necessarily, one of the triplets  
\be
(z_{j-1}, z_j z_{j+1}),\, (z_j, z_{j+1}, z_{j+2}), \dots, (z_{k-2}, z_{k-1}, z_{k})
\ee 
is not extremal, which breaks (ii). Indeed, if all these triplets were extremal, then we would have $t_{j-1,j}<t_{j,j+1} <t_{j+1,j+2}<\dots < t_{k-2,k-1}<t_{k-1,k}$. Just observe that $t_{j-1,j}<t_{j,j+1}$ because the triplet $(z_{j-1}, z_j, z_{j+1})$ is extremal, $t_{j,j+1} <t_{j+1,j+2}$ because the triplet $(z_j, z_{j+1}, z_{j+2})$ is extremal, etc. Now, given that $t_{k-1}<t_k$ and the fact that the arcs $\partial B(z_k)\cap A_{R,\eps}$ and $\partial B(z_{k-1})\cap A_{R,\eps}$ intersect only once at $v_{k-1,k}$,  all points  $x=(t,\varphi)\in B(z_k)\cap A_{R,\eps}$ with $t< t_{k-1,k}$ belong to $B(z_{k-1})$. On the other hand, the point $v_{j-1,j} = \partial B(z_{j-1}) \cap \partial B(z_j)\cap A_{R,\eps}$ does not belong to $B(z_{j+1}), B(z_{j+2}), \dots, B(z_{k-1}), B(z_{k})$. This is in contradiction with the condition that $v_{j-1,j}\in B(z_{k})$. 
\end{proof}

Proposition~\ref{prop:locality} shows that $\1_{\cO \cap \cD_\eps(0)}(\z)$ is a product of indicators involving triples of successive boundary points. This means that the constraint given by $\cO\cap\cD_\eps(0)$ is \emph{nearest-neighbour} only. 


\subsection{Volume and surface approximation}  
\label{sec:taylor}

In this section we derive approximations of two key quantities for $\z\in\mathcal O$ with $d_{\text{\rm{H}}}(\partial S(\z),\partial B_{\Rc}(0))\leq\eps$: 
\begin{itemize}
\item[$\bullet$]
$\mathcal H(\z)= (|S(\z)|-\ka|S^-(\z)|)- (\pi \Rc^2-\kappa\pi (\Rc-1)^2)$, the rate function for \emph{surface fluctuations} (see Fig.~\ref{fig:cleansausage}).
\item[$\bullet$]
$|S(\z)| - \pi \Rc^2$, the volume of the halo minus the volume of the \emph{critical disk}.
\end{itemize} 
For fixed $\kappa \in (1,\infty)$, we write $\Rc=\Rc(\kappa)$, assume that $\Rc\in(1+\frac{\eps}{1-2\eps},\frac{L}{\pi}+\tfrac12)$, and introduce explicit constants 
\be 
\label{C123def}
C_1 = \frac1{24}\Rc^2(\Rc-1), \quad  C_2 = \frac{1}{2(\Rc-1)}, \quad C_3 = \frac12, \quad C_4=\Rc.
\ee 

\begin{proposition}[Volume and surface approximation] 
\label{prop:expansion}
If $\z\in\mathcal O$ and $d_{\text{\rm{H}}}(\partial S(\z),\partial B_{\Rc}(0))\leq\eps$, then as $\eps\downarrow 0$, 
\be 
\begin{aligned}
\label{expansions}
\mathcal H(\z) 
&= \big(|S(\z)|-\ka|S^-(\z)|\big)-\big(\pi \Rc^2 - \kappa \pi (\Rc-1)^2\big)
= C_1^\eps y_1(\z) + C_2^\eps y_2(\z)-C_2 y_3(\z),\\
|S(\z)| - \pi \Rc^2
&= -C_1^\eps y_1(\z) + C_2^\eps y_2(\z) + C_3 y_3(\z) + C_4 y_4(\z),
\end{aligned}
\ee 
where $C_k^\eps = C_k\, [1+O(\eps)]$, $k=1,2$.
\end{proposition}

\begin{proof}
The proof comes in 4 steps. 

\medskip\noindent
{\bf 1.} 
First, we have $|S^-(\z)|=\sum_{i=1}^n V^-_i$, 
with  
\be
\label{eq:V-new}
V^-_i= A_i-B_i = \tfrac{1}{2} r_i r_{i+1}\sin\theta_i - \tfrac{1}{2}(\varphi_i-\sin\varphi_i), \qquad 1 \leq i<n,
\ee
where $A_i$ is the area of the triangle $z_i\, z_{i+1}\, 0$ and  $B_i$ is the area of the circular segment bounded by the line segment $z_i\, z_{i+1}$ and the arc of the circle $\partial B(v_i)$ subtending the angle $\varphi_i$ between the line segments $ v_i\, z_i$ and $ v_i\, z_{i+1}$ touching at the boundary cusp $v_i$. We use $\theta_i$ to denote the angle between the rays $0z_i$ and $0z_{i+1}$.

According to  Lemma~\ref{lem:boundary}, 
\be\label{eq:Delta}
\Delta(\z)=\mathcal H^1(\partial S^-(\z)) +\pi= \sum_{i=1}^n (\varphi_i + \tfrac12\theta_i), 
\ee
where $ \sum_{i=1}^n \varphi_i $ represents the length of the boundary $\partial S^-(\z)$ consisting of the union of arcs of circles $\partial B(v_i)$. Hence
\be
\abs{S(\z)}= \sum_{i=1}^n \tfrac{1}{2} (r_i r_{i+1}\sin\theta_i +\sin\varphi_i +\varphi_i+\theta_i).
\ee
In addition to the shorthand notation $V_i^-= \tfrac{1}{2} r_i r_{i+1}\sin\theta_i - \tfrac{1}{2}(\varphi_i-\sin\varphi_i)$, we introduce
\be
\label{E:Vi}
V_i = \tfrac{1}{2} (r_i r_{i+1}\sin\theta_i +\sin\varphi_i +\varphi_i+\theta_i) = V_i^-+ \varphi_i+ \tfrac{1}{2} \theta_i.
\ee
Our aim is to expand all the relevant terms in powers of $\theta_i$ and $\rho_{i+1} - \rho_{i}$.

\vspace{0.3cm}
\begin{figure}[htbp]
\centering
\begin{overpic}[width=10cm,scale=.25,percent]
{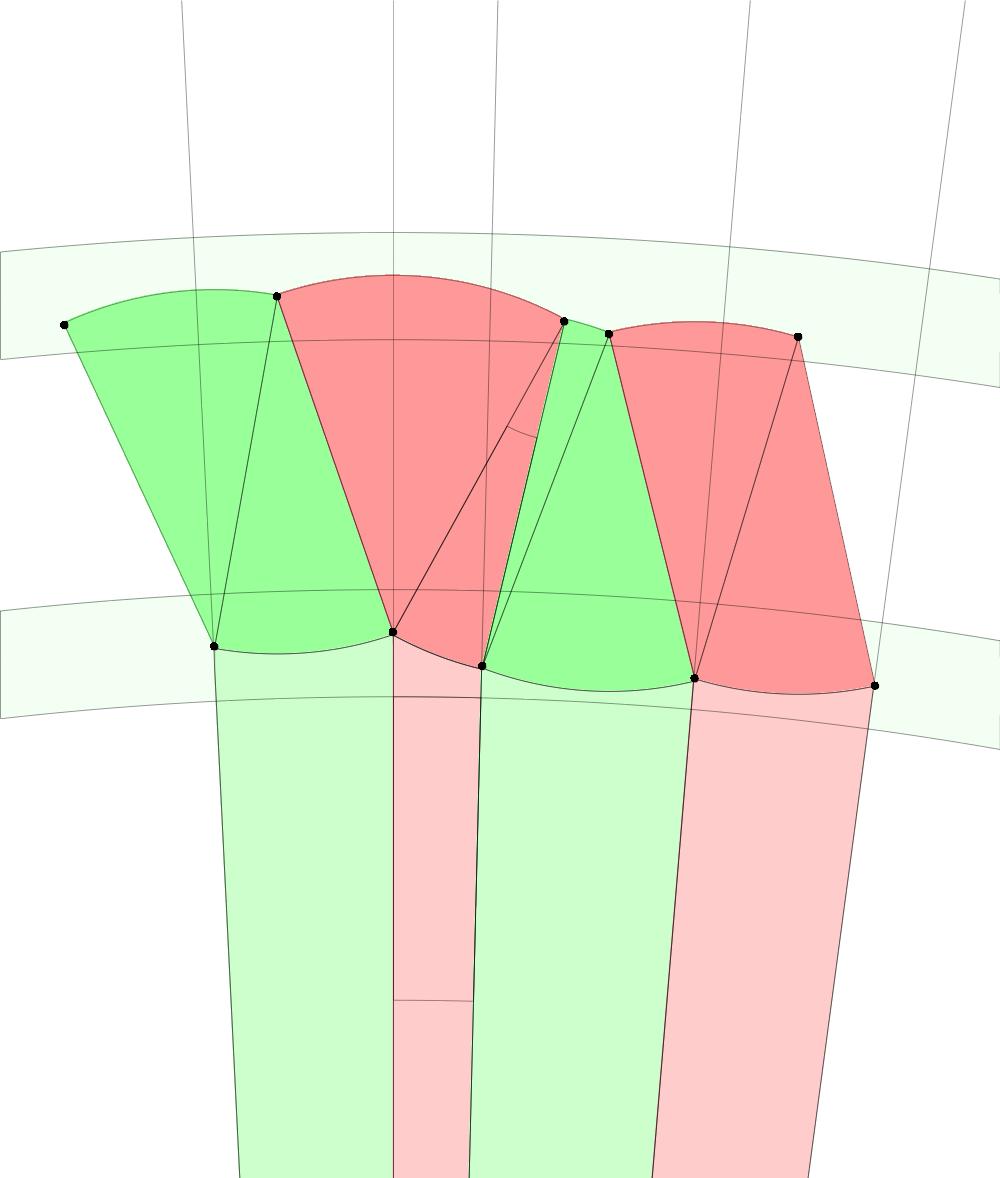}
\put(1.5,74.5) {$v_{i-1}$}
\put(20,76,5) {$v_{i}$}
\put(45,75) {$v_{i+1}$}
\put(51,73.5) {$v_{i+2}$}
\put(67,72.5) {$v_{i+3}$}
\put(41,58) {$\varphi_{i}$}
\put(35,25) {$V_{i}^{-}$}
\put(36,10) {$\theta_{i}$}
\put(11.5,44) {$z_{i-1}$}
\put(30.5,43) {$z_{i}$}
\put(35,41.7) {$z_{i+1}$}
\put(52.7,39.4) {$z_{i+2}$}
\put(71,39) {$z_{i+3}$}
\end{overpic}
\vspace{0.5cm}
\caption{\small The area $S^-(\z)$ splits into the union of fragments  with areas $V_{i}^{-} = \tfrac{1}{2} r_i r_{i+1}\sin\theta_i - \tfrac{1}{2}(\varphi_i-\sin\varphi_i)$ (in lighter colours), while the  boundary region  (darker) splits into the fragments of the areas $\varphi_i+\tfrac12\theta_i$, yielding $\abs{S(\z)} = \abs{S^-(\z)}+\Delta(\z) = \sum_{i=1}^n \tfrac{1}{2} (r_i r_{i+1}\sin\theta_i + \sin\varphi_i +\varphi_i + \theta_i)$.}
\label{fig:volumealt}
\end{figure}

\medskip\noindent
{\bf 2.}
Write
\be
\label{twodifferences}
\begin{aligned}
V_i - \tfrac{1}{2} \Rc^2\theta_i= I_i ^{(1)}+ I_i^{(2)},\\
V^-_i - \tfrac{1}{2} (\Rc-1)^2\theta_i=I_i^{(3)}+ I_i^{(4)},
\end{aligned}
\ee
where 
\be
\label{eq:vvmin}
\begin{array}{lll}
&I_i^{(1)} = V_i - \frac 1 2 (\overline{r_i}+1)^2\theta_i,
\qquad
&I_i^{(2)} = \frac12 (\overline{r_i}+1)^2\theta_i-\frac12\Rc^2\theta_i,\\[0.2cm]
&I_i^{(3)} = V^-_i-\frac{1}{2}\overline{r_i}^2\theta_i,
\qquad
&I_i^{(4)} = \frac{1}{2}\overline{r_i}^2\theta_i-\frac12 (\Rc-1)^2 \,  \theta_i .
\end{array}
\ee
To evaluate the two quantities in \eqref{expansions}, we need to compute the relevant approximations of the expressions
\be
\label{E:HS}
\mathcal H(\z) = \sum_{i=1}^n (I_i^{(1)}-\kappa I_i^{(3)}) + \sum_{i=1}^n (I_i^{(2)}-\kappa I_i^{(4)}),
\qquad \abs{S(\z)} -\pi \Rc^2 = \sum_{i=1}^n I_i^{(1)} + \sum_{i=1}^n I_i^{(2)}.
\ee 
For the terms $I_i^{(2)}$ and $I_i^{(4)}$, we use the definition $\rho_i=r_i-(\Rc-1)$ and $\overline{\rho_i}=\frac{1}{2}(\rho_i+\rho_{i+1})$ to get the identities 
\be 
\label{eq:i2i4}
I_i^{(2)}  = \tfrac12\overline{\rho_i} \theta_i \bigl(\overline{\rho_i}+2\Rc\bigr), \qquad 
I_i^{(4)}  = \tfrac12 \overline{\rho_i}\,\theta_i\bigl( \overline{\rho_i}+2(\Rc-1)\bigr), 
\ee 
yielding, in view of the relation $\kappa=\frac{\Rc}{\Rc-1}$, the equality
\be 
\label{E:i2i4comb}
I_i^{(2)} -\kappa I_i^{(4)} =- \frac{\kappa-1}2 \overline{\rho_i}^2 \theta_i.
\ee 

\medskip\noindent
{\bf 3.}
The terms $I_i^{(1)}$ and $I_i^{(3)}$ are a bit more elaborate and require an expansion of some terms.  Beginning with $I_i^{(3)}$, we use \eqref{eq:V-new} to rewrite the definition in \eqref{eq:vvmin} as
\be
\label{I3exp}
I_i^{(3)} = \tfrac{1}{2} (r_i r_{i+1}\sin\theta_i - \varphi_i+\sin\varphi_i 
-\overline{r_i}^2\theta_i).
\ee
Then, using 
\be
\label{E:rr}
r_i r_{i+1}= (\overline{r_i}-\tfrac12 (\rho_{i+1}-\rho_i) )(\overline{r_i}+\tfrac12(\rho_{i+1}-\rho_i))
=\overline{r_i}^2-(\tfrac{\rho_{i+1}-\rho_i}2)^2,
\ee
we write the difference between the first and the last term in \eqref{I3exp} as
\begin{align}
\label{eq:secpar}
\tfrac{1}{2} r_i r_{i+1}\sin\theta_i-\tfrac{1}{2}\overline{r_i}^2\theta_i
&=\tfrac12\overline{r_i}^2(\sin\theta_i-\theta_i) 
- \tfrac{1}{2} \bigl(\tfrac{\rho_{i+1}-\rho_i}{2}\bigr)^2\sin\theta_i\\\notag
&=-\tfrac{1}{12}\overline{r_i}^2(\theta_i^3+O(\theta_i^5))
- \tfrac{1}{2}\bigl(\tfrac{\rho_{i+1}-\rho_i}{2}\bigr)^2\sin\theta_i.
\end{align}
With the shorthand notation 
\begin{equation}
\label{short}
u_i=\abs{z_{i+1}-z_i}
\end{equation}
for the length of the line segment $z_i\,z_{i+1}$, we express this in terms of the angle $\varphi_i$ in the isosceles triangle $z_i\, z_{i+1}\, v_{i+1}$,
\be
\label{E:uisin}
u_i= 2\sin(\tfrac{\varphi_i}{2}).
\ee
Inverting \eqref{E:uisin}, we get
\be
\label{varphii}
\varphi_i=u_i+\tfrac{1}{24}u_i^3+ O(u_i^5)
\ee
and thus
\be
\tfrac12 (\varphi_i-\sin\varphi_i)=\tfrac{1}{12}u_i^3+O(u_i^5),
\ee
which together with \eqref{I3exp} and \eqref{eq:secpar} implies
\be
\label{eq:expvmin}
I_i^{(3)} = -\tfrac{1}{12}\overline{r_i}^2\theta_i^3
-\tfrac{1}{12}u_i^3-\tfrac{1}{2}\bigl(\tfrac{\rho_{i+1}-\rho_i}{2}\bigr)^2\sin\theta_i
+ O(u_i^5) +\overline{r_i}^2 O(\theta_i^5).
\ee
Expressing now $u_i$ from the triangle $0z_iz_{i+1}$ and using \eqref{E:rr}, we get
\begin{multline}
\label{trifirst}
u_i=(r_i^2+ r_{i+1}^2 - 2 r_i r_{i+1} \cos\theta_i)^{1/2}=\bigl((\rho_i-\rho_{i+1})^2 +
2 (\overline{r_i}^2-(\tfrac{\rho_{i+1}-\rho_i}2)^2) (1-\cos\theta_i)\bigr)^{1/2}=\\
= \bigl(4 \overline{r_i}^2 \sin^2(\tfrac{\theta_i}{2}) +( \rho_{i+1}-\rho_i )^2\cos^2(\tfrac{\theta_i}{2})\bigr)^{1/2}.
\end{multline}
 Approximating $u_i$ by $\widetilde u_i= 2  \overline{r_i} \sin(\tfrac{\theta_i}{2})$ ($\widetilde u_i$ representing  the base of isosceles triangle with legs $\overline{r_i}$ and the vertex angle $\theta_i$), we express the difference as
\be
\label{E:Eidef}
u_i -\widetilde u_i=2  \overline{r_i} \sin(\tfrac{\theta_i}{2})
\Biggl(
\sqrt{1+\frac{( \rho_{i+1}-\rho_i )^2\cos^2(\tfrac{\theta_i}{2})}
{4\overline{r_i}^2\sin^2(\tfrac{\theta_i}{2})}}-1\Biggr). 
\ee
Whenever $\z\in\mathcal O$ and $d_{\text{\rm{H}}}(\partial S(\z),\partial B_{\Rc}(0))\leq\eps$, we can use Proposition~\ref{prop:apriori}, asserting that $\overline{r_i}\ge \Rc-1-\eps$ and $\abs{\rho_{i+1}-\rho_i}/\theta_i \le K(\Rc)\sqrt{\eps}$. Assuming further that $\eps$ is  small enough so that $\Rc -1\ge \eps^{1/4}$, we get the bound
\be
\frac{( \rho_{i+1}-\rho_i )^2\cos^2(\frac{\theta_i}{2})}
{4\overline{r_i}^2\sin^2(\frac{\theta_i}{2})}\le
 \frac{4\Rc(\Rc-1)}{4 (\eps^{1/4}-\eps)^2} 4\eps\le O(\eps^{1/2}).
\ee
Hence, using that $\frac{\cos^2(\theta_i/2)}{\sin(\theta_i/2)} = \frac{2}{\theta_i}\bigl(1-\frac5{24}\theta_i^2+O(\eps^2)\bigr)$, we get
\be
\label{E:Eiexp}
u_i -\widetilde u_i = 2\,\overline{r_i}\sin(\tfrac{\theta_i}{2}) \tfrac12 \frac{( \rho_{i+1}-\rho_i )^2\cos^2(\tfrac{\theta_i}{2})}
{4\overline{r_i}^2\sin^2(\tfrac{\theta_i}{2})}\,[1+O(\eps)]
= \tfrac12\frac{(\rho_{i+1}-\rho_i)^2}{\overline{r_i}\theta_i}\,[1+O(\eps)].
\ee
For $u_i$ determined by \eqref{E:Eidef} we thus get, up to the order $O(\eps^{3/2})$, 
\be
\label{uiexp}
u_i= \overline{r_i}\left(\theta_i-\tfrac{\theta_i^3}{24} [1+O(\eps)]
+ \tfrac12 \frac{(\rho_{i+1}-\rho_i)^2} 
{\overline{r_i}^2\theta_i} [1+O(\eps)]\right)
\ee
and thus also
\be
\label{trisec}
 u_i^3 = (\bar  r_i \theta_i)^3[1 +  O(\eps)].
\ee 
Furthermore, considering the definitions of $I_i^{(1)}$ and $I_i^{(3)}$ in \eqref{eq:vvmin} and substituting in $I_i^{(1)}$ the term $V_i$ by $V_i^-$ from the equality \eqref{E:Vi}, we get
\be
\label{eq:I1-I3}
I_i^{(1)}=I_i^{(3)}+\varphi_i - \overline{r_i}\theta_i.
\ee
Using \eqref{uiexp} and  \eqref{trisec} jointly with \eqref{varphii}, we get for the term $\varphi_i - \overline{r_i}\theta_i$ above that
\be
\varphi_i - \overline{r_i} \theta_i = \Bigl( \tfrac12\tfrac{(\rho_{i+1}-\rho_i)^2}
{\overline{r_i}\theta_i}+\tfrac{\overline{r_i}}{24}(\overline{r_i}^2-1)\theta_i^3 \Bigr)[1+O(\eps)].
\ee
Inserting \eqref{trisec} into \eqref{eq:expvmin}, we get
\be
I_i^{(3)} = -\frac{\overline{r_i}^2(\overline{r_i}+1)}{12}\theta_i^3
-\frac{1}{2}\bigl(\frac{\rho_{i+1}-\rho_i}{2}\bigr)^2\sin\theta_i
+ O(\eps^{5/2}).
\ee
Combining \eqref{eq:I1-I3} with the last two equations, absorbing the term $\bigl(\rho_{i+1}-\rho_i\bigr)^2\theta_i$ into $\frac{(\rho_{i+1}-\rho_i)^2}{\theta_i}O(\eps)$, and using that $-\frac{1}{12}\overline{r_i}^2(\overline{r_i}+1)+\tfrac{1}{24}\overline{r_i}(\overline{r_i}^2-1) = -\tfrac{1}{24}\overline{r_i}(\overline{r_i}+1)^2$, we obtain
\be
\label{Ii1}
I_i^{(1)}=\Bigl(\frac12\frac{(\rho_{i+1}-\rho_i)^2}
{\overline{r_i}\theta_i}-\frac{\overline{r_i}(\overline{r_i}+1)^2}{24} \theta_i^3
\Bigr)[1+O(\eps)]
\ee
and
\be
\label{Ii13}
I_i^{(1)}-\kappa I_i^{(3)}=\Bigl(\frac12\frac{(\rho_{i+1}-\rho_i)^2}{\overline{r_i}\theta_i}
+\bigl(-\overline{r_i}(\overline{r_i}+1)^2+2 \kappa\overline{r_i}^2(\overline{r_i}+1)\bigr) \frac{\theta_i^3}{24}
\Bigr)[1+O(\eps)].
\ee

\smallskip\noindent
{\bf 4.}
Using that $\overline{r_i} = \Rc-1 + O(\eps)$ and $\kappa(\Rc-1)=\Rc$, and applying Proposition~\eqref{prop:apriori} once more while referring to  the definitions in \eqref{C123def}, we arrive at
\be
\label{eq:I1fin}
I_i^{(1)}  = \Bigl(-C_1 \theta_i^3+C_2 \frac{(\rho_{i+1}-\rho_i)^{2}}{\theta_i}\Bigr)[1+O(\eps)]
\ee
and
\be
\label{eq:I1I3fin}
I_i^{(1)}-\ka I_i^{(3)} = \Bigl(C_1\theta_i^3 +C_2 \frac{(\rho_{i+1}-\rho_i)^{2}}{\theta_i} \Bigr) [1+O(\eps)].
\ee
Combining \eqref{E:HS} with \eqref{eq:I1fin}, \eqref{eq:I1I3fin}, \eqref{eq:i2i4}, \eqref{E:i2i4comb}, \eqref{C123def}, and summing over $i$, we arrive at the claims in \eqref{expansions}.
\end{proof} 


\subsection{Geometric centre of a droplet}
\label{app:cenmass}

In the following, it will be useful to classify connected outer contours $\z$ by their uniquely and deterministically defined \emph{centre} 
$\cC(\z) = (\Sigma_1(\z), \Sigma_2(\z))\in \mathbb T$.

\begin{definition}[Centre of droplet]
\label{def:bdalt}
$\mbox{}$
For any $\z = \{z_1,\ldots,z_n\}\in\mathcal O$, consider the point $v=(v_1,v_2)\in \mathbb T$ defined as the centre of the smallest rectangle $R(\z)$, with horizontal and vertical sides, containing the set points from $\z$. Let $\z-v$ be the shift of $\z$ by $-v$, $\z-v= \{z_1-v,\ldots,z_n-v\}$, such that the rectangle $R(\z-v)$ has its centre in the origin. The \emph{geometric centre} of the halo shape $S(\z)$ is \emph{defined} as  
\be
\mathcal C(\z) = (\Sigma_1(\z),\Sigma_2(\z))
\ee
with (recall Definition \ref{def:yzsums})
\be
\label{CMexp}
\Sigma_1 (\z)=v_1+\tfrac{1}{2\pi} \,y_5(\z-v),  
\qquad
\Sigma_2 (\z)=v_2+\tfrac{1}{2\pi} \,y_6(\z-v).
\ee
\end{definition}

It can be shown that the centre $\mathcal C(\z)$ is an approximation of the barycentre of the boundary $\partial P(\z)$ of the polygon $P(\z)$ obtained by connecting the boundary points $z_i$ with the line segment of length $u_i$ connecting $z_i$ and $z_{i+1}$ (see Fig.~\ref{fig:volume}). However, this fact is not really needed for our purpose. The above definition is adopted for mathematical convenience -- its single role is to be used as a technical tool to achieve a discretisation that is needed to separate the set of all halo shapes $S(\z)$ into subsets with the geometric centre in a well defined area. 

In addition, the principal feature of crucial use is that the centre is given by the discretised Fourier coefficients $\frac1{2\pi} y_5(\z-v)$ and $\frac1{2\pi} y_6(\z-v)$. 

\begin{lemma}[A priori estimate of the centre location]
\label{lem:aprvolsurcen}
$\mbox{}$\\
Let $\z\in\mathcal O$ and $d_\text{{\rm H}}(\partial S(\z),\partial B_{R_c}(x))\leq \eps$ with $x\in \mathbb T$. Then, as $\eps\downarrow 0$,
\be
\label{E:cCeps}
\norm*{\cC(\z)-x}_{\infty} = \max\{\Sigma_1(\z)-x_1, \Sigma_2(\z)-x_2\} \le 2 \eps.
\ee
For $x=0$, this equals
\be
\label{E:cCeps0}
\norm*{\cC(\z)}_{\infty} = \max\{\Sigma_1(\z), \Sigma_2(\z)\} \le \eps.
\ee
\end{lemma}

\begin{proof}
If $d_\text{{\rm H}}(\partial S(\z),\partial B_{\Rc}(v))\leq \eps$, then the lengths of both sides of the rectangle $R(\z)$ belong to the interval $[2(\Rc -1)-2\eps, 2 (\Rc-1)+2\eps]$ and thus the position of its centre $v$ from $x$ satisfies $\norm*{v-x}_{\infty} \le \eps$. Given that, by \eqref{yzestalt}, $\abs{y_5(\z-v) }\le 2\pi \eps$ and $\abs{y_6(\z-v) }\le 2\pi \eps$, the claim \eqref{E:cCeps} follows. For $x=0$, the claim  \eqref{E:cCeps0} follows directly from \eqref{yzestalt}.
\end{proof}


\section{Stochastic geometry II: representation of probabilities as surface integrals} 
\label{sec:surface}

In Section~\ref{diskdrop} we prove that the contribution to the free energy coming from halos that are \emph{not} close to a critical disk either in volume or in Hausdorff distance is negligible (Lemma~\ref{lem:ng}), and that the centre of the critical disk can be placed at the origin (Lemma~\ref{lem:centering}). In Section~\ref{sec:inthaloshape} we show how to rewrite the integral over halo shapes in \eqref{Diri6} representing the free energy of the critical droplet, by tracking the \emph{boundary points} (Lemma~\ref{lem:coarsed} and Corollary~\ref{cor:recapI}). In Section~\ref{sec:auxiliary} we introduce \emph{auxiliary random variables}, list some of their properties (Lemma~\ref{lem:rhoangleint}), and rewrite the integral over halo shapes as an expectation of a certain exponential functional over these auxiliary random variables (Proposition~\ref{prop:keyrep}). The latter will serve as the starting point for the analysis in Sections~\ref{sec:prep}--\ref{proofmoddev}.  


\subsection{Only disk-shaped droplets matter} 
\label{diskdrop}

For $\delta, \eps>0$, define the events 
\be
\label{eventepsalt}
\begin{aligned}
\cV_\delta
&= \big\{\gamma \in \Gamma\colon\,|V(\gamma)-\pi \Rc^2| \leq \delta\big\},\\
\cD_\eps
&= \big\{\gamma \in \Gamma\colon\, d_{\text{\rm{H}}}(\partial h(\gamma),
\partial B_{\Rc}(x)) < \eps \mbox{ for some } x\in\T\big\},
\end{aligned}
\ee
i.e., the events where the the halo is $\delta$-close to a critical disk in volume and $\eps$-close to a critical disk in Hausdorff distance. Similarly as in Section~\ref{sec:heur-surfred}, where we introduced the notation $\cD_\eps(0)$ both for the configuration $\gamma \in \cD_\eps(0)$ and its boundary points $\z \in \cD_\eps(0)$, to lighten the notation we also write $\z\in\cV_\delta$ and $\z\in \cD_\eps$.

From Lemma~\ref{lem:admissible}(\ref{noholes}) we know that,  on the event $\cD_\eps$ with $\eps$ small enough, $h(\gamma)^-$ is connected and simply connected. First we check that we may discard the event $\cV_\delta \cap \cD_\eps^c$ with $\cD_\eps^c = \Gamma \setminus \cD_\eps$. 

\begin{lemma}
\label{lem:ng}
For every $C \in (0,\infty)$ and all $\eps>0$ small enough,
\be 
\label{eq:ng}
\mu_\beta\Bigl( |V(\gamma)- \pi \Rc^2|\leq C \beta^{-2/3}, \, \gamma \notin \cD_\eps \Bigr) 
\leq \e^{-\beta[1+o(1)][I^*(\pi \Rc^2) + \eps^2\frac{\ka-1}{4\pi}]}, \qquad \beta \to \infty.
\ee
\end{lemma}

\begin{proof}
Fix $C \in (0,\infty)$ and $\eps>0$. Note that $\cV_{C \beta^{-2/3}} \subset \cV_\eps$ for $\beta$ large enough. Since $\cV_\eps \cap \cD_\eps^c$ is closed, we can use the large deviation principle for the halo shape and the halo volume, derived in Theorems \ref{thm:ldp-halo} and \ref{thm:ldp-volume}, to bound 
\be
\label{ldpAB}
\limsup_{\beta\to \infty}\frac{1}{\beta} \log \mu_\beta\Bigl( |V(\gamma)- \pi \Rc^2| \leq C \beta^{-2/3}, \, 
\gamma \notin \cD_\eps \Bigr)  \leq -\inf_{S \in \cV_\eps \,\cap\, \cD_\eps^c} I(S).
\ee
In view of~\eqref{jhalo} and  \eqref{ihalo}, we are left with the minimisation problem
\be
\inf\Big\{|S| - \kappa |S^-|\colon\, S \in \cV_\eps \cap \cD_\eps^c\Big\}.
\ee
If $|S|=\pi R^2$, then, on the event $\cV_\eps$, we have $|R^2- \Rc^2| \leq \frac{\eps}{\pi}$ and, because $R + \Rc\ge 1$, also $\abs{R- \Rc} \leq \frac{\eps}{\pi}$. According to the negation of the second claim in Theorem~\ref{thm:isope}, $\z\in \cD_\eps^c$ implies that $|S(\z)|-\ka |S(\z)^-| \geq (\pi R^2-\ka \pi(R-1)^2) + \pi \ka \frac{\eps^2}{5R}$. Furthermore, employing the equality $\Rc = \kappa/(\kappa-1)$, we get $(\pi R^2-\ka \pi(R-1)^2) - (\pi \Rc^2 -\ka \pi(\Rc-1)^2) = -\pi(\ka-1)(R-\Rc)^2 \ge -\pi(\ka-1)\tfrac{\eps^2}{\pi^2}$, where, in the first equality we use that $(\ka-1)(\Rc^2-R^2)-2\ka(\Rc-R) = (\Rc-R)((\ka-1)(\Rc+R)-2\ka) = (\Rc-R)(\ka+(\ka-1)R-2\ka) = -(\ka-1)(\Rc-R)(\Rc-R)$. Hence
\be
|S|-\ka|S^-| \geq  (\pi \Rc^2 -\ka \pi(\Rc-1)^2) + \eps^2\bigl(\tfrac{\pi\ka}{5R}-\tfrac{\ka-1}{\pi}\bigr),
\ee
where $\tfrac{\pi\ka}{5R}-\tfrac{\ka-1}{\pi} \ge \frac{\ka-1}{4\pi}$ once $\eps$ is sufficiently small. To verify the latter, evaluate 
\be
\tfrac{\pi\ka}{5R} - \tfrac{\ka-1}{\pi} = \bigl(\tfrac{\pi^2}{5}\tfrac{\Rc}{R}-1\bigr)\tfrac{\ka-1}{\pi}
\ge \bigl(\tfrac{3}{2}\tfrac{\Rc}{R}-1\bigr)\tfrac{\ka-1}{\pi}\ge \tfrac{\ka-1}{4\pi}.
\ee
The last inequality uses $\frac{\Rc}{R}\ge 1-\frac{\eps}{\pi} \frac{\ka-1}{\ka}$, implied by $\abs{R- \Rc} \leq \frac{\eps}{\pi}$, and yields $\abs{\frac{R}{\Rc}-1}\le \frac{\eps}{\pi}\frac{\ka-1}{\ka}$. Consequently, \eqref{ldpAB} yields
\be
\limsup_{\beta\to \infty}\frac{1}{\beta} \log \mu_\beta\bigl( |V(\gamma)- \pi \Rc^2|\leq C \beta^{-2/3}, \, 
\gamma \notin \cD_\eps \bigr)  \leq  - I^*(\pi \Rc^2) - \eps^2\tfrac{\ka-1}{4\pi}
\ee
for $\eps$ small enough.
\end{proof} 

Next, we check that on the event $\z \in \cD_\eps = \bigcup_{x\in \mathbb T} \cD_\eps(x)$ we need only to consider droplets that are close to $B_{\Rc}(0)$. Here we exploit translation invariance. Introduce the event 
\be
\label{eventepsdelta}
\cC_\delta(x) = \bigl\{\z\in\mathcal O\colon\, \norm*{\cC(\z)-x}_{\infty}\leq \delta\bigr\},
\ee 
i.e., the centre $\cC(\z) = (\Sigma_1(\z), \Sigma_2(\z))$ of the halo shape $S(\z)$ is $\delta$-close to $x$.

\begin{lemma} 
\label{lem:centering}
For every $C \in (0,\infty)$, every $\eps>0$, and all $\beta$ sufficiently large, 
\be
\mu_\beta\bigl( \cV_{C\beta^{-2/3}}\cap \cD_{\eps}(0)\bigr)
\le \mu_\beta\bigl( \cV_{C\beta^{-2/3}}\cap \cD_{\eps} \bigr) 
\leq  L^2 \beta^{4/3}  \mu_\beta\bigl( \cV_{C\beta^{-2/3}} \cap \cC_{\delta}(0) \cap \cD_{\eps+\sqrt2 \beta^{-2/3}}(0) \bigr). 
\ee
\end{lemma} 

\begin{proof}
If $S(\z)$ is $\eps$-close in Hausdorff distance to the boundary of a disk of radius $\Rc$, then by Lemma~\ref{lem:aprvolsurcen} the centre $\cC(\z)$ of $ h(\gamma)$ is $2\eps$-close to the centre of that disk. Let $G_\delta=(\delta\Z)^2\cap \T$ be the grid of linear spacing $\delta= \beta^{-2/3}$. Observing that $\cD_{\eps}(y) \subset \cD_{\eps+\sqrt2 \delta}(x)$ if $\norm{x-y}_\infty \le \delta$, we get
\be
\label{eventsand}
\1_{\cD_{\eps}(0)} \leq \1_{\cD_\eps}
\leq \sum_{x \in G_\delta} \1_{\cC_{\delta}(x) \,\cap\, \cD_{\eps+\sqrt2 \delta}(x)}.
\ee
Given that the torus is periodic, every indicator contributes the same. Clearly, $\abs{G_{\delta}} = L^2 \beta^{4/3}$, and hence the proof is complete. 
\end{proof} 

Lemmas~\ref{lem:ng} and~\ref{lem:centering} leave us with the task of bounding $\mu_\beta(\mathcal  V_{C\beta^{-2/3}}\cap \cD_{\eps+\sqrt2 \delta}(0))$ from above. For the lower bound, it will be enough to bound $\mu_\beta(\cV_{C\beta^{-2/3}}\cap \cD_{\eps}(0) )$ from below. 


\subsection{Integration over halo shapes}
\label{sec:inthaloshape} 

Recall from Section~\ref{sec:heur-surfred} the definition of \emph{boundary points} of a halo $S=h(\gamma)$ with a simply connected $1$-interior $S^-$ and of the corresponding notion of a \emph{connected outer contour} $\z = (z_1,\ldots,z_n) \in \mathcal O$  (see Fig.~\ref{fig:cleansausage}).

\begin{lemma} 
\label{lem:coarsed}
Let $\Pi_\alpha$  be a Poisson point process of intensity $\alpha$ in $\mathbb T$. For all $\eps>0$ small enough and every bounded test function $f\colon\,\mathcal S\to \R$, as $\alpha \to\infty$,
\begin{equation}
\label{PPrepr}
\begin{aligned}
&\E\bigl[f \bigl(h(\Pi_\alpha)\bigr) \,\1_{\{\Pi_\alpha \in \cD_{\eps}(0) \}}\bigr]\\
&\in \bigl(1-\pi(\Rc-2+\eps)^2\alpha^{2/3}\e^{-\alpha}, 1\bigr)\,\e^{-\alpha |{\T}|}\, 
\sum_{n\in\N_0} \frac{\alpha^n}{n!} \int_{{\T}^n} \d \z\, f(S(\z))\, 
\e^{\alpha|S(\z)^-|}\1_{\cD_\eps(0)}(\z).
\end{aligned}
\end{equation}
\end{lemma} 

\begin{proof} 
For $\z \in \cO$, define 
\be 
\label{rhokappabeta}
r_{\alpha}(S(\z)) 
= \sum_{m\in\N_0}\frac{\alpha^m}{m!} \int_{{\T}^m} \d \y\,\,
\1_{\{h(\y \cup \z) = S(\z)\}}.
\ee
Clearly,
\be
\begin{aligned}
r_{\alpha}(S(\z)) &= \sum_{m\in\N_0}\frac{\alpha^m}{m!} \int_{{\T}^m} \d \y\,\,
\1_{\{h(\y \cup \z) \subset S(\z)\}}
- \sum_{m\in\N_0}\frac{\alpha^m}{m!} \int_{{\T}^m} \d \y\,\,
\1_{\{h(\y \cup \z) \subsetneq S(\z)\}}\\
&= \e^{\alpha |S(\z)^-|}\,[1-p_{\alpha}(\z)], 
\end{aligned}
\ee
where
\be 
\label{eq:pkappa}
p_{\alpha}(\z) =  \P\bigl( h(\Pi_{\alpha} \cup \z) 
\subsetneq S(\z) ~\big|~ \Pi_{\alpha} \subset S(\z)^-\Bigr),
\ee
is the probability that the halo has a \emph{hole}. To prove \eqref{PPrepr}, we argue as follows. Any configuration $\gamma$ with halo $h(\gamma) = S(\z)$ must be of the form 
\be
\gamma = \z\, \dot \cup \,\y \quad \text{ for some } \quad \z = \{z_1,\ldots,z_n\}, 
\quad \y = \{y_1,\ldots,y_{{\ell}-n}\}, \quad \ell \geq n,
\ee 
where $\y$ represents the set of \emph{interior points} of the configuration, i.e., $B(y_i) \cap \partial S(\z) = \emptyset$ for $1 \leq i \leq {\ell}-n$. Therefore, for every bounded test function $f$, 
\be 
\label{eq:shapeint}
\begin{aligned}
&\E\bigl[f \bigl(h(\Pi_\alpha)\bigr) \,\1_{\{\Pi_\alpha \in \cD_{\eps}(0) \}}\bigr] \\
&\quad = \e^{-\alpha |{\T}|} \sum_{\ell\in\N_0} \frac{\alpha^{\ell}}{{\ell}!} 
\int_{{\T}^{\ell}} \d\gamma\,\, f(h(\gamma)) 
\,\1_{\cD_{\eps}(0)}(\gamma)\\
&\quad = \e^{-\alpha |{\T}|} \sum_{\ell\in\N_0} \frac{\alpha^{\ell}}{\ell!} 
\sum_{n=0}^{\ell} \binom{\ell}{n} \int_{{\T}^n} \d \z\,\, 
f(S(\z))\,\1_{\cD_{\eps}(0)}(\z)
\int_{{\T}^{\ell-n}} \d \y\, \,\1_{\{h(\y \cup \z) = S(\z)\}} \\
&\quad = \e^{-\alpha |{\T}|} \sum_{n\in\N_0} \frac{\alpha^n}{n!}
\int_{{\T}^n} \d \z\,\, f(S(\z))\,\1_{\cD_{\eps}(0)}(\z)
\sum_{{\ell}=n}^\infty \frac{\alpha^{\ell-n}}{(\ell-n)!} 
\int_{{\T}^{\ell-n}} \d \y\,\,\1_{\{h(\y \cup \z) = S(\z)\}}\\ 
&\quad = \e^{-\alpha |{\T}|} \sum_{n\in\N_0} \frac{\alpha^n}{n!} 
\int_{{\T}^n} \d \z\,\, f(S(\z))\,\1_{\cD_{\eps}(0)}(\z)\,r_{\alpha}(S(\z))\\
&\quad = \e^{-\alpha |{\T}|} \sum_{n\in\N_0} \frac{\alpha^n}{n!} 
\int_{{\T}^n} \d \z\,\, f(S(\z))\,\1_{\cD_\eps(0)}(\z)\, 
\e^{\alpha|S(\z)^-|}\,[1-p_{\alpha}(\z)].
\end{aligned}
\ee
We get the claim in \eqref{PPrepr} once we show that, for $\eps$ small enough and uniformly in $\z \in \cD_\eps(0)$, 
\be 
\label{eq:rhokappaest}
p_{\alpha}(\z) \le \pi(\Rc-2+\eps)^2 \alpha^{2/3}\e^{-\alpha}, \qquad \alpha \to \infty.
\ee

The proof of the upper bound goes as follows. Let $\Pi^{S(\z)^-}_{\alpha}$ be the Poisson point process on $S(\z)^-$ with intensity $\alpha$. Then
\be
\label{pzrew}
\begin{aligned}
p_{\alpha}(\z) 
= \P\bigl(h\bigl(\Pi^{S(\z)^-}_{\alpha} \!\!\cup \z\bigr) \subsetneq S(\z)\bigr)
= \P\bigl(\exists\,y \in S(\z)\setminus h(\z)\colon\Pi_\alpha^{S(\z)^{-}}\!\! \cap B(y) 
= \emptyset\bigr).
\end{aligned}
\ee
To bound the last expression in \eqref{pzrew} we discretise. As in the proof of Lemma~\ref{lem:centering}, we consider the grid $G_{\delta}$ of linear spacing $\delta$. For every $y\in S(\z)\setminus h(\z)$ there exists an $x \in G_{\delta} \cap B_{\Rc-2}(0)$ such that $B(y) \supset B_{1-\delta\sqrt2} (x)$. Therefore
\be
\P\bigl(\exists\,y \in S(\z)\setminus h(\z)\colon\Pi_\alpha^{S(\z)^{-}}\!\! \cap B(y) = \emptyset\bigr)
\leq \sum_{x\in \mathcal G_{\delta}}  \P\bigl(\Pi^{S(\z)^{-}}_{\alpha}\!\! \cap B_{1-\delta\sqrt{2}}(x)= \emptyset\bigr).
\ee

The fact  that $\z$ is contained in $A_{\Rc-1,\eps}$ implies that the outer boundary of $h(\z)$, which is identical to the boundary of $S(\z)$, is contained in $A_{\Rc,\eps}$, while its inner boundary, which is identical with the boundary of $S(\z)\setminus h(\z)$, is contained in $A_{\Rc-2,\eps}$ (see the illustration on the left in Fig.~\ref{fig:cleansausage} as well as Fig.~\ref{fig:volumealt}). Using this observation, we note that $S(\z)\setminus h(\z) \supset  B_{\Rc -2-\eps}(0)$, and as a result 
\be
\abs{B(y) \cap (S(\z)\setminus h(\z))} \geq \abs{B(y)\cap B_{\Rc -2-\eps}(0)} \ge  1
\ee
for any $y\in S(\z)\setminus h(\z)$ and $\eps$ sufficiently small. The final inequality follows from the easily verified fact that, for the point $y_0(\eps)$ with orthogonal coordinates $x=0$ and $y=\Rc -2+\eps$,
\be
\lim_{\eps\to 0} \abs{B(y_0(\eps)) \cap B_{\Rc -2-\eps}(0)} = \abs{B(y_0(0)) \cap B_{\Rc -2}(0)} \ge \pi/3.
\ee
Analogously, for all $x \in G_{\delta} \cap (S(\z)\setminus h(\z))$ we have $\abs{B_{1-\sqrt2\delta}(x) \cap (S(\z)\setminus h(\z))} \geq 1$ for $\delta$ sufficiently small. Hence we get
\be
p_{\alpha}(\z) \leq  \abs{G_{\delta} \cap B_{\Rc-2+\eps}(0)} \e^{-\alpha} \leq \pi(\Rc-2+\eps)^2  \delta^{-2} \,\e^{-\alpha}.
\ee
Choosing $\delta=\alpha^{-2/3}$, we get the upper bound in \eqref{eq:rhokappaest}.
\end{proof} 

Recall from Section~\ref{LDiso} that $\mathcal{F}$ is the family of non-empty closed (and hence compact) subsets of the torus $\T$. In view of \eqref{Gibbsid}, we next observe that, for any $\mathcal{A} \subset \mathcal F$ measurable,
\be
\label{muidplus}
\mu_\beta \bigl( h(\gamma)\in\mathcal{A}, \, \gamma \in \cD_\eps(0)\bigr)
= \frac{\e^{-|{\T}|}}{\Xi_{\beta}}\,\e^{\kappa\beta |\mathbb{T}|}\E\bigl[\e^{-\beta V(\Pi_{\kappa\beta})} 
\,\1_{\{h(\Pi_{\kappa\beta}) \in \mathcal{A}\}}
\,\1_{\{\Pi_{\kappa\beta} \in \cD_{\eps}(0) \}}\bigr].
\ee
Applying Lemma~\ref{lem:coarsed} with $\alpha = \kappa \beta$ and with test functions of the form 
\be
f(S) = \e^{-\beta |S|}\1_\mathcal{A}(S), \quad \mathcal{A} \subset \mathcal F\text{ measurable},
\ee
we get the upper bound
\be
\mu_\beta \bigl( h(\gamma)\in\mathcal{A}, \, \gamma \in \cD_\eps(0)\bigr) \leq 
\frac{\e^{-|{\T}|}}{\Xi_{\beta}}\sum_{n\in\N_0} \frac{(\kappa\beta)^n}{n!} \int_{{\T}^n} 
\d \z\, \e^{-\beta(|S(\z)|-\kappa|S(\z)^-|)}\, 
\1_{\cD_\eps(0)}(\z)\,\1_{\{S(\z)\in\mathcal A\}}.
\ee
Abbreviating
\be 
\label{Ezdefsalt}
\cH(\z) = \big(|S(\z)|-\ka|S^-(\z)|\big) - \big(\pi \Rc^2 - \kappa\pi (\Rc-1)^2\big),
\ee
and recalling that
\be
I^*(\pi R^2)=I(B_R)=\pi R^2-\kappa\pi (R-1)^2 + (\kappa-1)\T,
\ee
we obtain
\be\label{eq:maybelast}
\mu_\beta \bigl( h(\gamma)\in\mathcal{A}, \, \gamma \in \cD_\eps(0)\bigr) 
\leq \e^{-\beta I^*(\pi R_c^2) }[1-o(1)] \sum_{n\in\N_0} \frac{(\kappa\beta)^n}{n!} 
\int_{{\T}^n} \d \z\, \e^{-\beta\mathcal H(\z)}\, 
\1_{\cD_\eps(0)}(\z)\1_{\{S(\z)\in\mathcal A\}},
\ee
where we use that $\Xi_\beta = \e^{[(\kappa-1)\beta-1]|\mathbb T|}[1+o(1)]$ (recall Lemma~\ref{lem:partfun}). An analogous lower bound holds (we skip the details).

Specialising the choice of $\mathcal{A}$ in \eqref{eq:maybelast}, we obtain the following corollary. 

\begin{corollary}[Representation as surface integral]
\label{cor:recapI}
For every $C \in (0,\infty)$ and all $\eps>0$ small enough, as $\beta\to\infty$,
\be 
\begin{aligned}
\mu_\beta\bigl( \cV_{C\beta^{-2/3}}\cap \cD_\eps(0) \cap \cC_{C\beta^{-2/3}}(0) \bigr)	
&= \big[1-o(1)\bigr]
\,\e^{-\beta\, I^*(\pi \Rc^2)}\,\mathcal{I}^\UB(\kappa,\beta;C,\eps),\\
\mu_\beta\bigl( \cV_{C\beta^{-2/3}}\cap \cD_\eps (0)\bigr)	
&= \big[1-o(1)\bigr]
\,\e^{-\beta\, I^*(\pi \Rc^2)}\,\mathcal{I}^\LB(\kappa,\beta;C,\eps)
\end{aligned}
\ee
with
\be
\label{eq:sumsurface} 
\begin{aligned}
\mathcal{I}^\UB(\kappa,\beta;C,\eps) 
&= \sum_{n \in \N_0} \frac{(\kappa \beta)^n}{n!} \int_{{\T}^{n}} \d \z\,\, 
\e^{-\beta \cH(\z)} \,\1_{\cV_{C\beta^{-2/3}} \,\cap\, \cC_{C\beta^{-2/3}}(0)
\,\cap\, \cD_\eps(0)}(\z),\\
\mathcal{I}^\LB(\kappa,\beta;C,\eps) 
&= \sum_{n \in \N_0} \frac{(\kappa \beta)^n}{n!} \int_{{\T}^{n}} \d \z\,\, 
\e^{-\beta \cH(\z)} \,\1_{\cV_{C\beta^{-2/3}} \,\cap\, \cD_\eps(0)}(\z).
\end{aligned}
\ee
\end{corollary}


\subsection{From surface integral to auxiliary random variables}
\label{sec:auxiliary}

To describe the boundary points, it will be expedient to introduce \emph{auxiliary random variables}. These variables will not only allow us to write surface integrals as expectations of exponential functionals involving a \emph{single} mean-centred Brownian bridge, they also allow us to \emph{embed} into the latter the surface integrals for $n$ particles in a \emph{consistent way}, so that we can pass to the limit $n\to\infty$ (see Fig.~\ref{fig:embed}). 

Let $(W_t)_{t\geq 0}$ be standard Brownian motion starting in $0$, and let
\begin{equation}
\label{Bhatdef}
(\widetilde W_t)_{t \in [0,2\pi]}, \qquad \widetilde W_t  = W_t - \frac{t}{2\pi} W_{2\pi},
\end{equation} 
be standard Brownian bridge on $[0,2\pi]$. Consider the process
\begin{equation} 
\label{Btildedef}
(B_t)_{t \in [0,2\pi]}, \qquad  B_t  = \widetilde W_t - \frac{1}{2\pi}\int_0^{2\pi} \widetilde W_s \d s,
\end{equation}
called the \emph{mean-centred Brownian bridge} (Deheuvels~\cite{D}). Set 
\be
\label{lambdadef}
\lambda(\beta) = G_\kappa \beta^{1/3}, \quad G_\kappa = \frac{\kappa^{2/3}}{\kappa-1}.
\ee
Let 
\be
\label{TNdef}
\begin{aligned}
\mathcal{T} &= \text{Poisson point process on $[0,2\pi)$ with intensity $\lambda(\beta)$},\\ 
N &= |\mathcal{T}| = \text{cardinality of $\mathcal{T}$}.
\end{aligned}
\ee 
Thus, $N$ is a Poisson random variable with parameter $2 \pi \lambda(\beta) = 2\pi G_\kappa\beta^{1/3}$. We assume that $(B_t)_{t\in [0,2\pi]}$ and $\mathcal T$ are defined on a common probability space $(\Omega,\mathscr F,\P)$ and they are independent. Conditional on the event $\{N=n\}$, we may write $\mathcal{T} = \{T_i\}_{i=1}^n$ with $0\leq T_1<\cdots<T_n<2\pi$, and define 
\be
\label{Thetadef}
\Theta_i = T_{i+1}- T_i,\,\,1 \leq i \leq n, \qquad \Theta_n = (T_1 + 2\pi) - T_n.
\ee 
Note that $\Theta_i\geq 0$, $1 \leq i \leq n$, and $\sum_{i=1}^n \Theta_i = 2 \pi$. For $m\in \R$, set 
\be\label{Zmdef}
\ZZ^{(m)} = \{Z_i^{(m)}\}_{i=1}^N
\ee 
with
\be\label{Zimdef}
Z_i^{(m)} = \bigl(r_i^{(m)} \cos T_i,  r_i ^{(m)} \sin T_i\bigr), \quad 
r_i^{(m)} = (\Rc-1) + \frac{m+  B_{T_i}}{\sqrt{(\kappa -1)\beta}},
\qquad 1 \leq i \leq N.
\ee
Later we will see that the natural reference measure for the angles $t_i$ is not the Poisson process $\mathcal T$ but a periodic version of a renewal process, which is conveniently constructed by tilting the distribution of the Poisson process $\mathcal T$. 

\begin{figure}[htbp]
\begin{center}
\includegraphics[width=5cm]{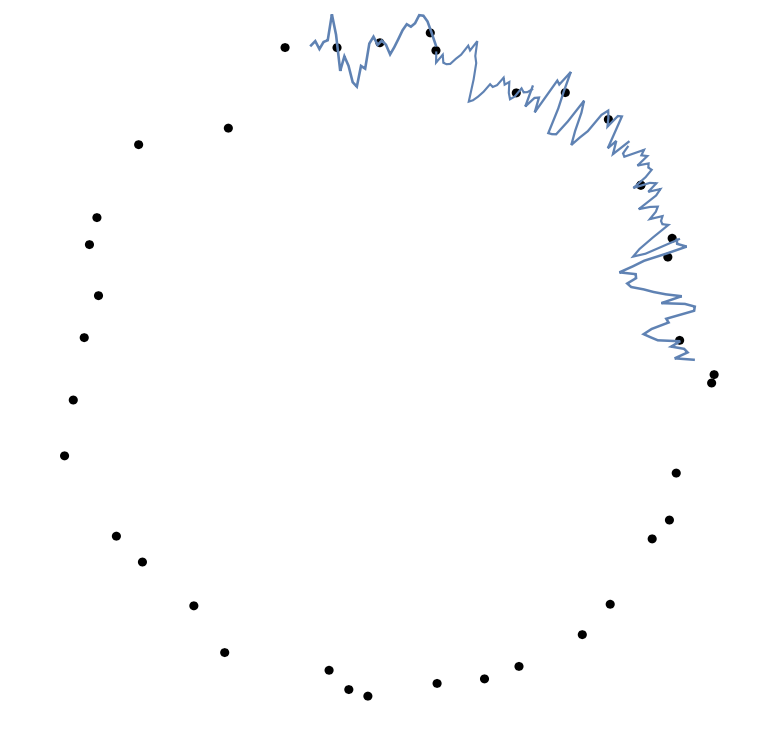}
\end{center}
\vspace{-0.3cm}
\caption{\small The black dots are the centers of the unit disks at the boundary, the blue curve draws a Brownian bridge through these centers. See also Fig.~\ref{fig:cleansausage}.}
\label{fig:embed}
\end{figure}

The following notation allows us to avoid indicators that order variables. For $(\t,\rr)\in [0,2\pi)^n \times\R^n$ with $\t=(t_1,\ldots,t_n)$ pairwise distinct, let $\sigma\in \mathfrak{S}_n$ (the set of permutations of $\{1,\ldots,n\}$) be such that $0\leq t_{\sigma(1)}<\cdots < t_{\sigma(n)} <2\pi$. We abbreviate $(t_{(i)},  r_{(i)}) = (t_{\sigma(i)},  r_{\sigma(i)})$ and $(t_{(n+1)},  r_{(n+1)}) = (t_{(1)} + 2\pi , r_{(1)})$. The reordering depends on the vector $\t$, but for simplicity we  suppress the $\t$-dependence from $r_{(i)}$ and $t_{(i)}$.

The following lemma can be viewed as a relation between measures on the space of marked point processes.

\begin{lemma}
\label{lem:rhoangleint}
For every non-negative test function $f$ on $([0,\infty) \times [0,2\pi])^n$ with $f(\emptyset)=0$, 
\begin{multline} 
\label{eq:rhoaint}
\sum_{n\in\N_0} \frac{1}{n!} (\kappa \beta (\Rc-1))^n  
\int_{\R^n} \d  \r \int_{[0,2\pi)^n} \d \t \,\, \exp\left(-\beta \frac{\kappa-1}{2} \sum_{i=1}^n 
\frac{(r_{(i+1)}- r_{(i)}) ^2}{t_{(i+1)} - t_{(i)}}\, f\bigl( \{( r_{i}, t_{i})\}_{i=1}^n\right) \\
= \frac{1}{\sqrt{2\pi}}\, \e^{2\pi G_\kappa \beta^{1/3} } \int_\R \d m\,\, 
\E\left[ f\left( \left\{ \left(\frac{m+ B_{T_i}}{\sqrt{(\kappa-1)\beta}}, T_i\right)\right\}_{i=1}^N\right) 
\prod_{i=1}^N \sqrt{2\pi G_\kappa \beta^{1/3} \Theta_i}\,\right]. 	
\end{multline} 
\end{lemma} 

\begin{proof} 
We revisit the computations from Section~\ref{sec:heur-aux}. Put 
\be
\label{xrhoscal}
x_i = \sqrt{(\kappa-1)\beta}\, \rho_i, \qquad \rho_i= r_i-(\Rc-1),
\ee
and set $\theta_i = t_{(i+1)}- t_{(i)}$, where $t_{(n+1)} = t_{(1)} +2\pi$, and $\bar f_{t}(\x) = f(\{( r_i,t_i)\}_{i=1}^n)$. Let $P_\theta(x-y)$ be the heat kernel, see~\eqref{heatkernel}. Proceeding as in Section~\ref{sec:heur-aux}, we see that the left-hand side of~\eqref{eq:rhoaint} is equal to 
\be 
\label{rhoaint2}
\sum_{n\in\N_0} \frac{1}{n!} (\beta^{1/3} G_\kappa )^n  
\int_{[0,2\pi)^n} \d \t \,  \prod_{i=1}^n \sqrt{(2\pi) \beta^{1/3} G_\kappa \theta_i} \, 
\int_{\R^n} \d  \x \,   \bar f_{t}(\x)\,  \prod_{i=1}^n  P_{\theta_i}(x_{(i+1)} - x_{(i)}).
\ee
We first evaluate the inner integral for fixed $n\in \N$ and $\t\in [0,2\pi)^n$ with $t_1<\cdots <t_n$, so that $(t_{(i)},  r_{(i)}) = (t_i, r_i)$ and $x_{i+1} = x_i$. Change variables as $x\to (x_1,x'_2,\ldots, x'_n)$ with $x'_i = x_i - x_1$. The inner integral becomes 
\be
\int_{\R^n} \d  \x \,   \bar f_{t}(\x)\,  \prod_{i=1}^n  P_{\theta_i}(x_{(i+1)} - x_{(i)}) 
= \int_\R \d x_1\, \int_{\R^{n-1}} \d \x'\,  \bar f_{t}(x_1,x_1+\x')\,  \prod_{i=1}^n  P_{\theta_i}(x'_{i+1} - x'_{i}), 
\ee
where $x'_1 = x'_{n+1}=0$. From the semi-group property of the heat kernel, we get
\be
\label{eq:gaussiannorm} 
\int_{\R^{n-1}} \d \x'\, \prod_{i=1}^n P_{\theta_i}(x'_{i+1} - x'_i) 
= P_{\sum_{i=1}^n \theta_i}(x'_{n+1} - x'_{1}) = P_{2\pi}(0) = \frac{1}{\sqrt{2\pi}}. 
\ee
Next, we note 
\be
\int_\R \d x_1\, \int_{\R^{n-1}} \d \x'\,  \bar f_{t}(x_1,x_1+\x')\,  \prod_{i=1}^n  P_{\theta_i}(x'_{i+1} - x'_{i})  
= (2\pi)^{-1/2} \int_\R \d x_1\, \E\Bigl[\bar f_t\big(\{x_1 + \widetilde W_{t_i}\}_{i=1}^n\big) \Bigr]
\ee
(think of $x'_i = \widetilde W_{t_i}$). Furthermore, for every non-negative test function $g$ on path space, 
\begin{equation}
\begin{aligned}
\E\Big[\int_\R \d x_1\, g( x_1 + \widetilde W)\Bigr]	
& = \E\Bigg[\E\Big[\int_\R \d x_1\, g( x_1 + \widetilde W) 
\,\Bigg|\, (\widetilde W_t)_{t\in [0,2\pi]}\Bigr]\Biggr]\\
& = \E\Bigg[\E\Big[\int_\R \d m\, g( m + B)\,\Bigg|\, (\widetilde W_t)_{t\in [0,2\pi]}\Bigr]\Biggr]
= \E\Big[\int_\R \d m\, g( m +  B) \Bigr],
\end{aligned}
\end{equation}
where we have changed variables $m= x_1 + M$ with $M= \frac{1}{2\pi} \int_0^{2\pi} \d t\,\widetilde W_t$. It follows that 
\begin{equation} 
\label{eq:innerintegral}
\int_{\R^n} \d  \x \,   \bar f_{t}(\x)\,  \prod_{i=1}^n  P_{\theta_i}(x_{(i+1)} - x_{(i)}) \\
=\int_\R \d m\, (2\pi)^{-1/2}\, \E \Bigl[ f\Bigl( \Bigl\{ \Bigl(\frac{m+ B_{t_{(i)}}}{\sqrt{(\kappa-1)\beta}}, 
t_{(i)} \Bigr)\Bigr\}_{i=1}^n\Bigr) \Bigr]. 	
\end{equation}
This holds as well when the $t_i$'s are pairwise distinct but not necessarily labeled in increasing order. The case when $t_i = t_j$ for some $i\neq j$ has Lebesgue measure zero and need not be considered. Denote the right-hand side in \eqref{eq:innerintegral} by $g(t)$. Then~\eqref{rhoaint2} reads 
\be
\frac{1}{\sqrt{2\pi}}\,\sum_{n\in\N_0} \frac{1}{n!} (G_\kappa \beta^{1/3})^n  
\int_{[0,2\pi)^n} \d \t \,  \prod_{i=1}^n \sqrt{(2\pi) \beta^{1/3} G_\kappa \theta_i}\,\, g(\t).
\ee
With the help of \eqref{TNdef}--\eqref{Thetadef}, this expression in turn is equal to 
\be 
\frac{1}{\sqrt{2\pi}}\e^{2\pi G_\kappa \beta^{1/3}} 
\E\Bigl[ g(\{T_i\}_{i=1}^N)\prod_{i=1}^N \sqrt{2\pi G_\kappa \beta^{1/3}}\,\Bigr]
\ee 
and the proof is readily concluded.
\end{proof} 

The integral over $m$ corresponds to a freedom of choice in the \emph{average height of the surface of the critical droplet} with respect to critical radius $\Rc$, and constitutes a fine tuning of the volume. We will later see that the integral is dominated by values of $m$ that are at most of order $\beta^{1/6}$.       

Remember the process $Z^{(m)}= (Z_i^{(m)})_{i=1}^N$ from ~\eqref{Zmdef} and~\eqref{Zimdef}, define the following  random variables
\be 
\label{Y0Y1}
\widehat{Y}_0 = \frac12 \sum_{i=1}^N \log \bigl(2\pi \beta^{1/3} G_\kappa \Theta_i\bigr), 
\quad \widehat{Y}_1= \frac{1}{24} \sum_{i=1}^N \bigl(\beta^{1/3}G_\kappa\Theta_i\bigr)^3,
\ee
and consider the tilted probability measure $\widehat \P$ on $(\Omega,\mathscr F,\P)$ defined by
\be 
\label{tiltedproba}
\widehat \P(\mathcal{A}) = \frac{\E[\exp( \widehat{Y}_0 - \widehat{Y}_1)\1_\mathcal{A}]}
{\E[\exp(\widehat{Y}_0 - \widehat{Y}_1)]}, \qquad \mathcal{A}\subset \Omega \text{ measurable}.
\ee
In the sequel we use the notation $\overline{u_i} = \tfrac12(u_i+u_{i+1})$, $1 \leq i \leq N$. 
 Further define (recall the definition of $C_1$ in \eqref{C123def})
\be
\label{Y0Y1def}
Y_0=\widehat{Y}_0, \qquad Y_1 = \sum_{i=1}^N \Theta_i^3= \frac{1}{C_1\beta} \widehat{Y}_1,
\ee
\be
\label{Ydefs}
Y_2 = \sum_{i=1}^N \frac{(B_{T_{i+1}}- B_{T_i})^2}{\Theta_i}, \quad
Y_3^{(m)} = \sum_{i=1}^N \big(m+ \overline{B_{T_i}}\,\big)^2 \Theta_i.
\ee
Put
\be
\delta(\beta) = \beta^{-2/3}.
\ee

\begin{proposition}[Representation of key surface integrals] 
\label{prop:keyrep}
The integrals in \eqref{eq:sumsurface} equal
\begin{equation}
\label{kikihat}
\begin{aligned}
\mathcal{I}^\UB(\kappa, \beta;C,\eps) 
&= \e^{2\pi G_\kappa \beta^{1/3} + o(\beta^{1/3})}\, \E[\e^{Y_0 - C_1\beta Y_1}]\\
&\qquad \times \int_{\mathbb R} \mathrm \d m\,\,\widehat \E\Bigl[ \e^{O(\eps) (\beta Y_1 + Y_2  + N) 
+ \frac12 Y_3^{(m)}} \1_{\{Z^{(m)}\in \cV_{C\delta(\beta)} \,\cap\, 
\cC_{C\delta(\beta)}(0) \,\cap\, \cD_\eps(0)}\}\Bigr],\\
\mathcal{I}^\LB(\kappa, \beta;C,\eps) 
&= \e^{2 \pi G_\kappa \beta^{1/3} + o(\beta^{1/3})}\, \E[\e^{Y_0 - C_1\beta Y_1}]\\
&\qquad \times \int_{\mathbb R} \mathrm \d m\,\,\widehat \E\Bigl[ \e^{O(\eps) (\beta Y_1 + Y_2  + N) 
+ \frac12 Y_3^{(m)}} \1_{\{Z^{(m)}\in \cV_{C\delta(\beta)} \,\cap\, \cD_\eps(0)}\}\Bigr].
\end{aligned}
\end{equation} 
where $C_1=\frac{\kappa^2}{24(\kappa-1)^3}$.
\end{proposition} 

\begin{proof} 
Return to \eqref{eq:sumsurface} and recall \eqref{C123def}, Propositions~\ref{prop:expansion} and \eqref{Ezdefs}. First we rewrite the expression in \eqref{eq:sumsurface} in terms of polar coordinates $z_i = ( r_i \cos t_i, r_i \sin t_i)$, $1 \leq i \leq N$, i.e., $\int_{\T^N} \d \z$ becomes  $\int _{\R^N} \d \r \int_{[0,2\pi)^N} \d \t \, \prod_{i=1}^N  r_i$. The latter product becomes 
\be
\label{radlim}
\prod_{i=1}^{N}  r_i = \prod_{i=1}^{N} \big[\Rc-1+\rho_i \big]
= [\Rc-1+O(\eps)]^{N} = (\Rc-1)^{N}[1+O(\eps)]^{N},
\ee
where the last equality uses the constraint $\cC_{\delta(\beta)}(0) \cap \cD_\eps(0)$ together with the a priori estimate from Proposition~\ref{prop:apriori}. In this way we obtain the factor $(\Rc-1)^N$ needed for Lemma~\ref{lem:rhoangleint}, together with an error term $\exp[O(\eps) N]$, which shows up as the term $O(\eps) N$ in \eqref{kikihat}. The factor $\e^{2\pi G_\kappa \beta^{1/3}}$ is needed to compensate for the exponent in the Poisson distribution of $N$ (recall \eqref{lambdadef}).

The Gaussian density in Lemma~\ref{lem:rhoangleint} is obtained from $\exp[-\beta\cH(\z)]$, more precisely, from the second  term in the expansion of $\cH(\z)$ in the first line of \eqref{expansions}, up to an error term $O(\eps)$ in the constant $C_3^\eps=C_3\,[1+O(\eps)]$, which shows up as the term $O(\eps)Y_2$ in \eqref{kikihat}. Hence Lemma~\ref{lem:rhoangleint} is applicable. The first and the third term in the expansion of $\cH(\z)$ in the first line of \eqref{expansions} give rise to the term $- \beta C_1^\eps Y_1 + \tfrac12 Y_3^{(m)}$ (recall \eqref{xrhoscal}). The factor $\prod_{i=1}^N \sqrt{2\pi G_\kappa \beta^{1/3}\Theta_i}$ in \eqref{eq:rhoaint} gives rise to $Y_0$. The indicator is inherited from the original expression for the integral. 
\end{proof}


\section{Asymptotics of surface integrals I: preparations} 
\label{sec:prep}

Our primary task for proving Theorem~\ref{thm:zoom1} in Section~\ref{proofmoddev} is the evaluation of the key surface integrals $\mathcal I^\mathrm{UB}$ and $\mathcal I^\mathrm{LB}$ in \eqref{kikihat}. In this section we collect some properties of the auxiliary random variables appearing in Section~\ref{sec:auxiliary} that will help us to estimate these integrals. This requires various approximation arguments, including control of \emph{exponential moments} and \emph{discretisation errors}. 

In Section~\ref{angularldp} we look at moderate deviations for the angular process and compute the leading order contribution to the key surface integrals (Proposition~\ref{prop:theta-asympt} and Lemma~\ref{lem:ehatperturbed}). In Section~\ref{mcbb} we analyse the radial process, which is controlled by the mean-centred Brownian bridge introduced in \eqref{Bhatdef}--\eqref{Btildedef}(Lemma~\ref{lem:gaussproc}), and estimate two exponential moments involving the latter (Lemma~\ref{lem:varadhanprep}--\ref{lem:y3moments}). In Section~\ref{diskrerrors} we focus on discretisation errors that arise because the mean-centred Brownian motion is only observed along the angular process (Lemmas~\ref{estimateslamb}--\ref{lem:stochdiscr2}).


\subsection{Moderate deviations for the angular point process}
\label{angularldp} 

The key technical result of this section is the following. Remember $\widehat{Y}_0,\widehat{Y}_1$ from~\eqref{Y0Y1} and $Y_0, Y_1$ from \eqref{Y0Y1def}, $G_\kappa$ and $\lambda(\beta) = G_\kappa \beta^{1/3}$ from~\eqref{lambdadef}. Let $\tau^* \in \R$ be the unique solution to the equation 
\be 
\label{tstar}
\int_0^\infty \sqrt{2\pi u} \exp\Bigl(-\tau^* u - \frac{u^3}{24}\Bigr)\, \dd u =1. 
\ee
The change of variables $s = u^3/24$ together with $\int_0^\infty s^{-1/2}\e^{-s} \d s = \Gamma(\frac12) 
= \sqrt \pi$ yields
\be
\int_0^\infty \sqrt{2\pi u}\exp\left(-\frac{u^3}{24}\right)\,\d u = \frac{4 \pi}{\sqrt 3} > 1,
\ee 
and so $\tau^*>0$. A numerical approximation of the integral gives $\tau^*\in (1.60,1.61)$.

\begin{proposition}[Leading order prefactor] 
\label{prop:theta-asympt}
\be 
\label{asympbare} 
\lim_{\beta \to \infty}\frac{1}{\beta^{1/3}}\log \E\bigl[\e^{Y_0 - \beta C_1 Y_1}\bigr] =- 2\pi G_\kappa (1- \tau^*).
\ee
\end{proposition} 

\begin{proof} 
The proof builds on an underlying renewal structure. Define the probability density 
\be
\label{eq:qdensity} 
q_*(u) = \sqrt{2\pi u}\, \exp\Bigl(-\tau^* u - \tfrac{1}{24} u^3 \Bigr), \qquad u \in (0,\infty), 
\ee
where $\tau^*$ is given in~\eqref{tstar}. A close look at the relevant expressions in polar coordinates reveals that
\be 
\begin{aligned} 
\label{renewal1}
\E\bigl[\e^{Y_0 - \beta C_1 Y_1}\bigr]
&= \e^{- 2\pi G_\kappa \beta^{1/3}} \Biggl(1+
\sum_{n\in\N} (G_\kappa \beta^{1/3})^n \frac{2\pi}{n}  \int_{[0,2\pi]^{n-1}} \d \theta\,  
\mathbf1_{\{\sum_{i=1}^{n-1}\theta_i\leq 2\pi\}}\\
&\qquad\qquad\qquad\qquad\qquad \times  
\prod_{i=1}^n \Bigl(\sqrt{2\pi G_\kappa \beta^{1/3} \theta_i}\,\, \e^{- \frac{1}{24} 
G_\kappa \theta_i^3}\Bigr)\Biggr),
\end{aligned}
\ee 
where $\theta_n = 2\pi - \sum_{i=1}^{n-1} \theta_i$. (The factor $\frac{2\pi}{n}$ represents the number of ways to rotate a configuration in such a way that the origin falls within in an interval of average length $\frac{2\pi}{n}$, and is similar to a factor appearing in the definition of stationary renewal processes.) Rewrite~\eqref{renewal1} with $n$-fold convolutions of $q_{*}$ as 
\be 
\label{eq:renewal2}
\E\bigl[\e^{Y_0 - \beta C_1 Y_1}\bigr]
= \e^{- 2\pi G_\kappa \beta^{1/3}} \Bigl(1+ \e^{2\pi G_\kappa \beta^{1/3} \tau^*} G_\kappa \beta^{1/3}
\sum_{n\in\N} \frac{2\pi}{n} (q_{*})^{*n}(2\pi G_\kappa \beta^{1/3})\Bigr).
\ee
If the factor $\frac{2\pi}{n}$ were absent, then the sum over $n$ would correspond to the probability for a renewal process with interarrival distribution $q_*$ to have a renewal point at time $2\pi G_\kappa \beta^{1/3}$ given that it has  a renewal point at time $0$. For large time intervals standard renewal theory tells us that the renewal probability converges to the inverse of the expected interarrival time, which is finite. Therefore we may expect the inner sum to converge and the overall expression to behave like a constant times $G_\kappa \beta^{1/3}\, \exp(-2\pi G_\kappa \beta^{1/3}(1-\tau^*))$. Remember that $\tau^*>0$, so the contribution from $n=0$ is negligible. 

For the proof of the upper bound in~\eqref{asympbare}, we bound the sum over $n$ in~\eqref{eq:renewal2} by 
\be
2\pi\, \mathcal R\bigl( 2\pi G_\kappa \beta^{1/3}\bigr) 
\quad \text{ with } \quad \mathcal R(\ell)= \sum_{n\in \N} (q_{*})^{*n}(\ell).
\ee
The quantity $\mathcal R(\ell)$ solves the renewal equation 
\be
\mathcal R(\ell) = q_*(\ell) + \int_0^\infty \d y\,q_*(y)\,\mathcal R(\ell - y). 
\ee
It follows from~\cite[Theorem 2, Chapter XI.3]{feller-vol2} and the smoothness of $\ell\mapsto\mathcal R(\ell)$ that 
\be
\lim_{\ell\to \infty} \mathcal R(\ell) = \frac{1}{\int_0^\infty \d u\,u\, q_*(u)}
\ee
and hence $\lim_{\ell\to \infty}\frac{1}{\ell} \log \mathcal R(\ell) = 0$. Combining this with~\eqref{eq:renewal2} and recalling that $\tau^*>0$, we get 
\be
\limsup_{\beta\to \infty} \frac{1}{\beta^{1/3}} \log \E\bigl[\e^{Y_0 - \beta C_1 Y_1}\bigr] 
\leq 2\pi G_\kappa (\tau^*-1). 
\ee

For the proof of the lower bound in~\eqref{asympbare}, we drop all except one term from the sum in~\eqref{eq:renewal2}, i.e., 
\be
\label{lbrenew}
\E[\e^{Y_0 - C_1\beta Y_1}] \geq G_\kappa \beta^{1/3} \e^{- 2\pi G_\kappa \beta^{1/3}(1 - \tau^*)} 
\frac{2\pi}{n} (q_{*})^{*n}(2\pi G_\kappa \beta^{1/3}).
\ee
This inequality holds for every $n\in \N$, and a proper choice will be made later.  Let $(X_i)_{i\in\N}$ be i.i.d.\ random variables with probability density function $q_*$. Then $\E[X_1] = \mu_*$ with $\mu_* = \int_0^\infty \d u\,u\, q_*(u)$, and $(\sum_{i=1}^n X_i - n \mu_*)/\sqrt{n}$ has probability density function 
\be
p_n(y) = \sqrt{n}\, (q_*)^{*n}\bigl( \sqrt{n}\,[n \mu_* +y]\bigr). 
\ee
Put differently, 
\be
(q_*)^{*n}(x) = \tfrac{1}{\sqrt{n}}\,p_n\bigl(\tfrac{x - n \mu_*}{\sqrt n}\bigr). 
\ee
By the local central limit theorem for i.i.d.\ random variables with densities (see~\cite[Chapter 4.5]{IL}), we have 
\be
\lim_{n\to \infty} \sup_{y\in \R}\,  \Bigl|\, p_n(y) - \frac{1}{\sqrt{2\pi  \sigma^2}}
\exp\Bigl( - \tfrac{y^2}{2\sigma^2}\Bigr)\Bigr| = 0
\ee
with $\sigma^2$ the variance of $X_1$. We now choose $n = n(\beta) = \lfloor 2\pi G_\kappa \beta^{1/3} /\mu_* \rfloor$.  Then $2\pi G_\kappa \beta^{1/3} = n(\beta) \mu_* + o(1)$ and 
\be
(q_{*})^{*{n(\beta)}}(2\pi G_\kappa \beta^{1/3}) = \frac{1}{\sqrt{n(\beta)}}\, p_{n(\beta)}\bigl( o(1)\bigr) 
= [1+ o(1)]\, \frac{1}{\sqrt{2\pi G_\kappa \beta^{1/3}/\mu_*}} \frac{1}{\sqrt{2\pi  \sigma^2 }}. 
\ee
Consequently, 
\be
\lim_{\beta \to \infty} \frac{1}{\beta^{1/3}} \log\,\Big[(q_{*})^{*{n(\beta)}}(2\pi G_\kappa \beta^{1/3})\Big] = 0
\ee
and so \eqref{lbrenew} gives
\be
\liminf_{\beta \to \infty}\frac{1}{\beta^{1/3}}\log \E[\e^{Y_0 - C_1\beta Y_1}] \geq- 2\pi G_\kappa (1-\tau^*). 
\ee
\end{proof}

Proposition~\ref{prop:theta-asympt} is complemented by the following lemma, which will help us take care of small perturbations. 

\begin{lemma}
\label{lem:ehatperturbed}
As $\delta \downarrow 0$, 
\begin{equation}
\limsup_{\beta\to \infty}\left| \frac{1}{\beta^{1/3}} \log
\widehat \E\bigl[\e^{O(\delta)(\beta Y_1 + N)}\bigr] \right| = O(\delta). 
\end{equation}
\end{lemma} 

\begin{proof}
Let $c\geq 0$ be an arbitrary constant that does not depend on $\beta$. We write $- c \min (1, C_1) \delta \leq O(\delta) \leq c \min (1, C_1) \delta$. Then 
\be 
\label{ehat1}
\log\widehat \E\bigl[\e^{c \delta (\beta C_1 Y_1 + N)}\bigr] 
= \log \E \bigl[\e^{c \delta N + Y_0 - (1- c\delta) \beta C_1 Y_1}\bigr] 
- \log \E\bigl[\e^{Y_0 - \beta C_1 Y_1}\bigr]. 
\ee
The asymptotic behavior of the second term in the difference is given by Proposition~\ref{prop:theta-asympt}. For the first term, let $\tau^*(c\delta)$ be the solution of 
\be
\e^{c\delta} \int_0^\infty \sqrt{2\pi u}\, \exp\Bigl( - \tau^*(c\delta) - \tfrac{1}{24}(1-c\delta) u^3\Bigr) \, \d u =1. 
\ee
Thus, $\tau^*(0) = \tau^*$. For sufficiently small $\delta$, the solution exists, is unique, and satisfies 
\be
\tau^*(c\delta) - \tau^* = O(\delta) \quad (\delta \downarrow 0). 
\ee
Arguments analogous to the proof of Proposition~\ref{prop:theta-asympt} show that
\be
\lim_{\beta\to \infty}\frac{1}{\beta^{1/3}} \log  \E \bigl[\e^{c \delta N + Y_0 - (1- c\delta) 
\beta C_1 Y_1}\bigr] = - 2\pi \bigl(1 - \tau^*(c\delta)\bigr).  
\ee
Hence \eqref{ehat1} yields 
\be
\lim_{\beta \to \infty} \frac{1}{\beta^{1/3}} \log \widehat \E\bigl[\e^{c \delta (\beta C_1 Y_1 + N)}\bigr]  
= - 2\pi \bigl( \tau^*(c\delta) - \tau^*\bigr) = O(\delta). 
\ee
A similar argument shows that
\be
\lim_{\beta \to \infty} \frac{1}{\beta^{1/3}} \log \widehat \E\bigl[\e^{- c \delta (\beta C_1 Y_1 + N)}\bigr]  
= O(\delta). 
\ee
\end{proof} 

We close this section with the following observation, which will not be needed in the sequel but is nonetheless instructive. Let $\mathcal P(0,\infty)$ denote the space of probability measures on $(0,\infty)$, equipped with the weak topology, and put
\begin{equation}
\label{Varadhan2}
\mathcal M = \left\{(x,\mu) \in [0,\infty) \times \cP(0,\infty) \colon\,	
\int_0^\infty \d\mu(\alpha)\,\alpha =1\right\}.
\end{equation}
Define
\be
L = \frac{1}{N} \sum_{i=1}^{N} \delta_{N\Theta_i/2\pi}.
\ee 
For $(x,\mu) \in \mathcal M$ with $x>0$, define 
\be 
\label{ratefunct}
I_\mathcal{T}(x,\mu) = (x \log x - x +1) + x H \big(\mu \mid \mathrm{EXP}(1)\big),
\ee
where $\mathrm{EXP}(1)$ is the exponential distribution with parameter $1$, and $H(\mu \mid \mathrm{EXP}(1))$ is the relative entropy of $\mu$ with respect to $\mathrm{EXP}(1)$. For $(0,\mu)\in \mathcal M$, define
\be 
\label{ratefunct0}
I_\mathcal{T}(0,\mu) = \liminf_{\substack{(x,\nu)\to (0,\mu):\\x>0,\,\nu \in \cP(0,\infty)}} 
I_\mathcal T(x,\nu),
\ee
Then the family
\be
\label{NLfam}
\left(\mathbb{P}\bigg(\Big(\frac{N}{2\pi G_\kappa \beta^{1/3}},L\Big) \in \cdot\, 
\mid N\geq 1\bigg)\right)_{\beta \geq 1}
\ee
satisfies the weak LDP on $\mathcal M$ with rate $2\pi G_\kappa \beta^{1/3}$ and lower semi-continuous rate function $I_\mathcal T$. (Not all level sets of $I_\mathcal T$ are compact.)


\subsection{Properties of the mean-centred Brownian bridge} 
\label{mcbb}

\medskip \noindent 
\textbf{Covariance of mean-centred Brownian bridge.}
The following lemma clarifies the nature of the process $( B_t)_{t\in [0,2\pi]}$ and will be used repeatedly later on. 
 
\begin{lemma} 
\label{lem:gaussproc}\hfill 
\begin{enumerate}[(a)] 
\item[{\rm (a)}]
$( B_t)_{t\in [0,2\pi]}$ is a Gaussian process with mean $\E[ B_t]=0$ 
and covariance $\E[B_t  B_s] = k(t-s)$, where 
\be\label{covdef}
k(t) = \frac{1}{4\pi}\big[(\pi - |t|)^2 - \pi^2\big] +\frac{\pi}{6}.
\ee 
\item[{\rm (b)}] 
For every continuous function $f\colon\,[0,2\pi]\to \R$, 
\be\label{eq:gaussiancf}
\E\bigl[ \e^{\mathrm{i} \int_0^{2\pi} \d t\,f(t)  B_t}\bigr] 
= \e^{- \frac{1}{2}\langle f,K f\rangle},
\ee
with $\langle \cdot,\cdot\rangle$ the scalar product in $L^2([0,2\pi])$, and $g(t)=(Kf)(t)= \int_0^{2\pi} \d s\,k(t-s) f(s)$ the solution of $- g''(t) =f (t)- \frac{1}{2\pi} \int_0^{2\pi} \d s\,f(s)$, $g(2\pi) = g(0)$ and $\int_0^{2\pi} \d s\,g(s) = 0$.  
\end{enumerate} 
\end{lemma} 

\begin{proof}  
(a) $( B_t)_{t \in [0,2\pi]}$ is a linear transformation of the Gaussian process $(W_t)_{t\in[0,2\pi]}$ and therefore is itself Gaussian. The mean-zero property of $B_t$ is inherited from $W_t$. The elementary computation of the covariance is similar to Deheuvels~\cite[Lemma 2.1]{D}. We provide the details to identify constants. Set $M= \frac{1}{2\pi}\int_0^{2\pi} \d s\,B_s$. Since $\E[W_t W_s] = \min (s,t)$ we have, for $s\leq t$,
\be
\begin{aligned} 
\E[\widetilde W_s \widetilde W_t] & =\min(s,t) - \frac{st}{2\pi} = \frac{1}{2\pi} s(2\pi - t),\\
\E[\widetilde W_t M] & = \frac{1}{2\pi} \int_0^{2\pi} \d u\,\Bigl( \min(u,t) - \frac{ut}{2\pi}\Bigr)
=  \frac{1}{4\pi}t(2\pi - t),\\
\E[M^2] &=  \frac{1}{8\pi^2} \int_0^{2\pi} \d t\, (2 \pi t - t^2)
= \frac{\pi}{6},
\end{aligned} 
\ee
and hence 
\be
\begin{aligned} 
\E[  B_t  B_s] 
&= \min(s,t) - \frac{st}{2\pi} -\frac{1}{2\pi} \Bigl( - \frac{t^2}{2}+t\pi\Bigr)
- \frac{1}{2\pi} \Bigl( - \frac{s^2}{2}+s\pi\Bigr) + \E[M^2] \\
& =  \frac{1}{4\pi}\bigl[\bigl(\pi - |t-s|\bigr)^2 - \pi^2 \bigr] +\frac{\pi}{6} = k(t-s).
\end{aligned}
\ee 
By symmetry, the identity also holds for $t\leq s$. 

\medskip\noindent	
(b) \eqref{eq:gaussiancf} follows from standard arguments for Gaussian processes. The kernel $k\colon\,[-2\pi,2\pi]\to \R$ satisfies
\be 
\label{kernelprop}
k(t) = k(t+2\pi) \quad \forall\,t\in [-2\pi,0], \qquad 
\int_0^{2\pi} \d t \, k(t) = 0,
\ee 
and is twice differentiable with second derivative $-1/2\pi$, except at $t=0$ where the first derivative jumps from $+\frac{1}{2}$ to $-\frac{1}{2}$. Let $g= K f$. Then $g$ has mean zero and satisfies $g(2\pi) = g(0)$. Furthermore, 
\be
\begin{aligned}
g'(t) & = \int_0^{t} \d s\,k'(t-s) f(s) + \int_t^{2\pi} \d s\,k'(t-s) f(s),\\
g''(t) & = k'(0+) f(t) - k'(0-) f(t) + \tfrac{1}{2\pi} \int_0^{2\pi} \d s\,f(s)
= - f(t) + \tfrac{1}{2\pi} \int_0^{2\pi} \d s\,f(s).
\end{aligned}
\ee
\end{proof} 

\noindent 
For later purpose we record the variance of the increments, namely,
\be
\label{eq:cov1}
\E\bigl[( B_{t+h} -  B_t)^2\bigr] = 2 k(0) - 2 k(h) =  |h| -  \frac{h^2}{2\pi}.
\ee
Thus, for small time increments we recover the variance of standard Brownian increments. We also record the covariance of two distinct increments, namely, for $h,u\geq 0$ and $t+h\leq s$, 
\be
\label{eq:cov2}
\begin{aligned}
\E\bigl[( B_{t+h} -  B_t) ( B_{s+u} -  B_s)\bigr]
&=  \frac{1}{4\pi}\Bigl((s+u-t-h)^2 - (s+u-t)^2 - (s-t-h) ^2 + (s-t)^2\Bigr) \\
& = -\frac{1}{2\pi} hu. 
\end{aligned}
\ee 
Thus, two distinct increments are not independent, however, for $h,u \downarrow 0$ the covariance is negligible compared to the variance of the individual increments (since $hu = o(h)+ o(u)$). This will be needed in Lemma~\ref{lem:y3moments} below.

\medskip \noindent 
\textbf{Exponential moments for the mean-centred Brownian bridge.} 
For $k\in \N$, define the random variables
\be
\label{foucoeff}
A_k = \frac{k}{\sqrt{\pi}} \int_0^{2\pi} \d t\, B_t \cos (kt), \quad 
 A_k^* = \frac{k}{\sqrt{\pi}} \int_0^{2\pi} \d t\, B_t \sin (kt). 
\ee
In view of Lemma~\ref{lem:gaussproc}, these random variables are i.i.d.\ standard normal. They represent the Fourier coefficients of $ B$, i.e., 
\be
\label{fourier} 
 B_t =\frac{1}{\sqrt{\pi}} \sum_{k\in\N} \Bigl[ \frac{A_k}{k} \cos (kt) 
+ \frac{ A_k^*}{k} \sin (kt) \Bigr],
\ee 
where the series converges in $L^2(0,2\pi)$ $\P$-a.s. The expansion in \eqref{fourier} is the Karhunen-Lo{\`e}ve expansion of the Gaussian process $ B$ (see Alexanderian~\cite{A}, Deheuvels~\cite{D}, and references therein).

\begin{lemma} 
\label{lem:varadhanprep}
For every $-\infty<s<1$, 
\be
\label{Fest1}
\E\left[\exp\left[\tfrac{1}{2}s \left(\int_0^{2\pi} \d t\,{B}_t^2\right)\right]\right] 
= \prod_{k\in\N} \Bigl(1 - \frac{s}{k^2}\Bigr)^{-1} \quad \P\text{-a.s.}
\ee
For every $-\infty<s<4$,
\be
\label{Fest2}
\E\left[\exp\left[\tfrac{1}{2}s \left(\int_0^{2\pi} \d t\,{B}_t^2 - A_1^2 - {A^*_1}^2\right)\right]\right] 
= \prod_{k\in\N\setminus\{1\}} \Bigl(1 - \frac{s}{k^2}\Bigr)^{-1} \quad \P\text{-a.s.}
\ee 
\end{lemma} 

\begin{proof} 
Note that \eqref{fourier} implies 
\be 
\label{bnorm}
\int_0^{2\pi} \d t\, B_t^2 = \sum_{k\in\N} \tfrac{1}{k^2}(A_k^2 +  A_k^{*2}).
\ee
Since $A_k$, $A_k^*$ are i.i.d.\ standard normal, the claim follows from the identity $\E[\exp(\tfrac12 u X^2)] = (1-u)^{-1/2}$ when $X$ is standard normal and $u<1$. Apply this identity with $u=s/k^2$, $k \in \N$ to get \eqref{Fest1}. The proof of \eqref{Fest2} is similar.
\end{proof} 

The next lemma allows us to subsume the error term $O(\eps)Y_2$ into the error term $O(\eps) N$ appearing in \eqref{kikihat}.

\begin{lemma} 
\label{lem:y3moments}
Let $0\leq t_1<\cdots<t_n<2\pi$, $t_{n+1} = t_1+2\pi$, and $\theta_i = t_{i+1} - t_i$. 
For every $s \in (0,1)$, 
\be
\E\left[\exp\Bigl( \tfrac{1}{2} s \sum_{i=1}^n
\frac{(B_{t_{i+1}} - B_{t_i})^2}{\theta_i}\Bigr) \right]  
= \frac{1}{(\sqrt{1-s}\,)^{n-1}}.
\ee
\end{lemma} 

\begin{proof} 
Define $L_i = (B_{t_{i+1}} - B_{t_i})/\sqrt{\theta_i}$, $1 \leq i \leq n$. It follows from \eqref{eq:cov1} and \eqref{eq:cov2} that $(L_i)_{1 \leq i \leq n}$ is a Gaussian vector with covariance matrix $C =(C_{ij})_{1\leq i,j\leq n}$ given by
\be 
\label{scaledcov}
C_{ij} = \delta_{ij} - \frac{\sqrt{\theta_i\theta_j}}{2\pi}.
\ee
Hence $C = \mathrm{id} - P$, where $P$ is the orthogonal projection onto the linear span of $(\sqrt{\theta_i/2\pi})_{i=1}^n$ in $\R^n$. Thus, $C$ is the orthogonal projection onto the $(n-1)$-dimensional hyperplane defined by $\{\ell=(\ell_i)_{i=1}^n\colon\, \sum_{i=1}^n \ell_i \sqrt{\theta_i} = 0\}$. Using orthonormal coordinates on that hyperplane, we find that 
\be
\E\left[ \exp\Bigl(\tfrac{1}{2}s \sum_{i=1}^n L_i^2\Bigr)\right] 
= \frac{1}{\sqrt{2\pi}^{\,n-1}} \int_{\R^{n-1}} \d x\,\exp\Bigl( \tfrac{1}{2} s x^2 - \tfrac{1}{2} x^2\Bigr) 
= \frac{1}{(\sqrt{1-s}\,)^{\,n-1}},
\ee
which settles the claim.
\end{proof} 


\subsection{Discretisation errors}
\label{diskrerrors}
 

\paragraph{Deterministic discretisation errors.}
Later we need to bound the error in the approximations 
\be
\sum_{i=1}^n \overline B_{t_i}\theta_i \approx \int_0^{2\pi} B_t\, \d t 	=0, \qquad 
\sum_{i=1}^n {\overline B_{t_i}}^2 \theta_i \approx \int_0^{2\pi} B_t^2\, \d t, 
\ee
and discretisation errors for Fourier coefficients. Lemmas~\ref{lem:stochdiscr3} and~\ref{lem:stochdiscr1} treat the errors as random variables and provide bounds for their exponential moments. The proofs build on bounds for deterministic discretisation errors, which we provide first. 

Let $\mathcal H$ be the space of absolutely continuous functions $f: [0,2\pi]\to \R$  with square-integrable derivative, satisfying $\tau(2\pi) = \tau(0)$ and $\int 0^{2\pi} \tau(s) \d s =0$. Let $\tau\in\cH$. Note that $|\tau(t) - \tau(s)| = |\int_s^t \dot {\tau}(s')\d s'| \leq ||\dot \tau||_2 \sqrt{2\pi}$, and hence
\be\label{tausup}
\begin{aligned}
|\tau(t)| &= \left| \frac{1}{2\pi}\int_{0}^{2\pi} \d s\, [\tau(t) - \tau(s)] \right| 
\leq \sqrt{2\pi}\, \|\dot \tau\|_2,\\
\|\tau\|_\infty &\leq \sqrt{2\pi}\, \|\dot \tau\|_2.
\end{aligned}
\ee
Set 
\be
\overline{\tau_i}=\tfrac{1}{2}\big[\tau(t_i) +\tau(t_{i+1})\big],
\ee
and consider the sums
\begin{equation}
\label{Lambdaidef}
\begin{aligned}
&\Lambda_1 = \sum_{i=1}^n \overline{\tau_i}^2\,\theta_i, \quad  
\Lambda_2 = \sum_{i=1}^n \tau(t_i)\,\theta_i,\\
&\Lambda_3 = \sum_{i=1}^n \tau(t_i)\,\theta_i \cos t_i, \quad 
\Lambda_4 = \sum_{i=1}^n \tau(t_i)\,\theta_i \sin t_i.
\end{aligned}
\end{equation}

\begin{lemma}
\label{estimateslamb}
Suppose that $\tau\in\cH$ and put $\eps_n =\sum_{i=1}^n \theta_i^3 = y_2(\z)$. Then
\be
\label{XYapprox}
\begin{array}{lll}
&\big|\Lambda_1-\|\tau\|_2^2\big| \leq  \sqrt{\eps_n}\,\|\tau\|_\infty \|\dot\tau\|_2,
&\big|\Lambda_2\big| \leq \sqrt{\eps_n/3}\,\|\dot\tau\|_2,\\[0.2cm]
&\big|\Lambda_3-\int_0^{2\pi} \d t\,\tau(t)\cos t \big| 
\leq 2\sqrt{\eps_n/3}\,\|\dot\tau\|_2,
&\big|\Lambda_4-\int_0^{2\pi} \d t\,\tau(t)\sin t\big|  
\leq 2 \sqrt{\eps_n/3}\,\|\dot\tau\|_2.
\end{array}
\ee
\end{lemma}

\begin{proof}
Note that  
\be
\bigl|\tau(t)^2 - \overline{\tau_i}^2\bigr| = \bigl|\bigl(\tau(t) - \overline{\tau_i}) 
(\tau(t)+ \overline{\tau_i})\bigr|
\leq 2 \|\tau\|_\infty\, |\tau(t) - \overline{\tau_i}| 
\ee
and 
\be
\begin{aligned}
&\int_{t_i}^{t_{i+1}} \d t\,|\tau(t) - \overline{\tau_i}| 
= \int_{t_i}^{t_{i+1}} \d t \left| \int_{t_i}^t \d s\,\dot{\tau}(s)\right|
\leq  \int_{t_i}^{t_{i+1}} \d t \int_{t_i}^t \d s\,|\dot{\tau}(s)|
= \int_{t_i}^{t_{i+1}} \d s \,(t_{i+1}-s)\,|\dot{\tau}(s)|\\
&= \int_{t_{i}}^{t_{i+1}} \d s \,\tfrac12 \bigl[ (t_{i+1} - s)+ (s- t_i)\bigr] \,|\dot{\tau}(s)|
= \tfrac{1}{2} \theta_i \int_{t_i}^{t_{i+1}}\d s\,|\dot{\tau}(s)|.
\end{aligned}
\ee
It therefore follows that 
\be\label{lamba4}
\begin{aligned}
&\Bigl| \Lambda_1 - \|\tau\|_2^2\Bigr| 
= \left| \sum_{i=1}^n \overline{\tau_i}^2\theta_i - \int_0^{2\pi} \d t\, \tau(t)^2\right|
= \left| \sum_{i=1}^n \int_{t_i}^{t_{i+1}} \d t\,\big(\overline{\tau_i}^2 - \tau(t)^2\big)\right|\\
&\leq 2\|\tau\|_\infty \sum_{i=1}^n \int_{t_i}^{t_{i+1}} \d t\,|\tau(t)-\overline{\tau_i}|
\leq \|\tau\|_\infty \sum_{i=1}^n  \theta_i \int_{t_i}^{t_{i+1}} \d t\,|\dot{\tau}(t)|
=  \|\tau\|_\infty \int_0^{2\pi} \d t\, |\dot{\tau}(t)| 
\sum_{i=1}^n  \theta_i\, \1_{[t_i,t_{i+1})}(t)\\  
&\leq \|\tau\|_\infty  \|\dot \tau\|_2 \left(\int_0^{2\pi} \d t\, 
\left[\sum_{i=1}^n \theta_i \1_{[t_i,t_{i+1})}(t)\right]^2\right)^{1/2}
= \|\tau\|_\infty  \|\dot \tau\|_2 \left(\sum_{i=1}^n \theta_i^3\right)^{1/2},
\end{aligned}
\ee
which is the inequality for $\Lambda_1$.

Let $f\colon\,\R\to \R$ be absolutely continuous and $2\pi$-periodic. Suppose for simplicity that $t_1=0$ (otherwise replace integrals over $[0,2\pi]$ by integrals over $[t_1,t_{n+1}]= [t_1,t_1+2\pi]$). By partial integration we have 
\begin{align} 
\label{boundlambda2}
\sum_{i=1}^n \theta_i f(t_i) - \int_0^{2\pi} \d s\,f(s)
& = \sum_{i=1}^n \int_{t_i}^{t_{i+1}}  \d s\,[f(t_i)-f(s) ]
= - \sum_{i=1}^n \int_{t_i}^{t_{i+1}} \d u\, (t_{i+1} - u) \dot f(u),
\end{align} 
from which we get, by Cauchy-Schwarz, 
\begin{align}
\label{boundlambda2a}
\left| \sum_{i=1}^n \theta_i f(t_i) - \int_0^{2\pi} \d s\,f(s) \right|
\leq \|\dot{f}\|_2 \left(\sum_{i=1}^{n} \tfrac{1}{3} \theta_i^3\right)^{1/2}
= \|\dot{f}\|_2 \sqrt{\eps_n/3}.
\end{align}
The bound on $\Lambda_2$ is obtained from \eqref{boundlambda2a} by picking $f(s)= \tau(s)$ and using that $\int_0^{2\pi} \d s\, \tau(s) = 0$. The bounds on $\Lambda_3$ and $\Lambda_4$ are obtained from \eqref{boundlambda2a} by picking $f(s)= \cos s$ and $f(s) = \sin s$, respectively, and using that $\int_0^{2\pi} \d s\, \cos s = \int_0^{2\pi} \d s\, \sin s = 0$.
\end{proof}


\paragraph{Random discretisation errors.}
The estimates of \emph{deterministic} discretisation errors in Lemma~\ref{estimateslamb} can be used to derive estimates of exponential moments of \emph{random} discretisation errors which is needed later. 

\begin{lemma}[Discretised mean] 
\label{lem:stochdiscr3}
Put $\eps_n = \sum_{i=1}^n \theta_i^3$. Then, for every $s\in\R$, 
\be 
\label{EEintest}
\E\left[\exp\left(s\beta^{1/2} \sum_{i=1}^n \overline{B_{t_i}} \theta_i \right) \right] 
\leq \exp\Bigl(\tfrac{1}{32\pi} s^2\beta\eps_n\Bigr).
\ee
\end{lemma}

\begin{proof}
Write
\be
\label{Grepr1}
\mathbb E\left[\exp\left(s \beta^{1/2} \sum_{i=1}^n \overline{B_{t_i}}\,\theta_i\right)\right]
= \mathbb E\left[\exp\left(2\pi s \beta^{1/2} \int_0^{2\pi} \Gamma(\d t) B_t\right)\right],
\ee
where
\be
\Gamma = \sum_{i=1}^{n}\tfrac 1 2 \bigl(\delta_{t_i} +\delta_{t_{i+1}}\bigr) \frac{\theta_i}{2\pi}
= \sum_{i=1}^n \delta_{t_i} \frac{t_{i+1}-t_{i-1}}{4\pi}
\ee
is a probability measure on $[0,2\pi]$. Apply \eqref{eq:gaussiancf} to get
\be
\label{Grepr2}
\mathbb E\left[\exp\left(2\pi s \beta^{1/2}\int_0^{2\pi} \Gamma(\d t) B_t\right)\right]
=\exp\left(\tfrac12(2\pi s)^2 \beta G \right) 
\ee
with
\be
\label{Gint1}
\begin{aligned}
G = \int_0^{2\pi} \Gamma(\d s) \int_0^{2\pi} \Gamma(\d t)\, k(s-t)
= \frac{1}{(4\pi)^2} \sum_{i=1}^n (t_{i+1}-t_{i-1}) \sum_{j=1}^n (t_{j+1}-t_{j-1})\,k(t_i-t_j).
\end{aligned} 
\ee
The proof proceeds in two approximation steps. First, use \eqref{kernelprop} to write
\begin{equation}
\begin{aligned}
\sum_{i=1}^n (t_{i+1}-t_{i-1})\,k(t_i-t_j) 
&= \sum_{i=1}^n \int_{t_i}^{t_{i+1}} \d s\,\big[k(t_i-t_j) + k(t_{i+1}-t_j) - 2 k(s-t_j)\big]\\
&= - \sum_{i=1}^n \int_{t_i}^{t_{i+1}} \d s\, \int_{t_i-t_j}^{s-t_j} \d u\,\dot{k}(u)
+ \sum_{i=1}^n \int_{t_i}^{t_{i+1}} \d s\, \int_{s-t_j}^{t_{i+1}-t_j} \d u\,\dot{k}(u)\\
&=  - \sum_{i=1}^n \int_{t_i-t_j}^{t_{i+1}-t_j} \d u\,\dot{k}(u) \int_{u+t_j}^{t_{i+1}} \d s
+ \sum_{i=1}^n \int_{t_i-t_j}^{t_{i+1}-t_j} \d u\,\dot{k}(u) \int_{t_i}^{t_j+u} \d s\\
&= \sum_{i=1}^n \int_{t_i}^{t_{i+1}} \d u'\,\dot{k}(u'-t_j)\,(2u'-t_i-t_{i+1}).
\end{aligned}
\end{equation}
Substitute this into \eqref{Gint1} to get
\begin{equation}
\label{Gint2}
G = \frac{1}{(4\pi)^2} \sum_{i=1} \int_{t_i}^{t_{i+1}} 
\d u'\,(2u'-t_i-t_{i+1}) \sum_{j=1}^n (t_{j+1}-t_{j-1})\,\dot{k}(u'-t_j).
\end{equation}
Next, use \eqref{kernelprop} to write 
\begin{equation}
\begin{aligned}
\sum_{j=1}^n (t_{j+1}-t_{j-1})\,\dot{k}(u'-t_j)
&= \sum_{j=1}^n \int_{t_j}^{t_{j+1}} \d s\,\big[\dot{k}(u'-t_j) + \dot{k}(u'-t_{j+1}) - 2 \dot{k}(u'-s)\big]\\
&=  \sum_{j=1}^n \int_{t_j}^{t_{j+1}} \d s\, \int_{u'-s}^{u'-t_j} \d u\,\ddot{k}(u)
- \sum_{j=1}^n \int_{t_j}^{t_{j+1}} \d s\, \int_{u'-t_{j+1}}^{u'-s} \d u\,\ddot{k}(u)\\
&=   \sum_{j=1}^n \int_{u'-t_{j+1}}^{u'-t_j} \d u\,\ddot{k}(u) \int_{u'-u}^{t_{j+1}} \d s
- \sum_{j=1}^n \int_{u'-t_{j+1}}^{u'-t_j} \d u\,\ddot{k}(u) \int_{t_j}^{u'-u} \d s\\
&= \sum_{j=1}^n \int_{u'-t_{j+1}}^{u'-t_j} \d u\,\ddot{k}(u)\,(2u-2u'-t_j+t_{j+1}).
\end{aligned}
\end{equation}
Substitute this into \eqref{Gint2} to get
\begin{equation}
G = \frac{1}{(4\pi)^2} \sum_{i=1}^n \int_{t_i}^{t_{i+1}} \d u'\,(2u'-t_i-t_{i+1}) 
\sum_{j=1}^n \int_{t_j}^{t_{j+1}} \d u''\,(t_{j+1}-2u''-t_j)\,\,\ddot{k}(u'-u'').
\end{equation}
Finally, note that $|2u'-t_i-t_{i+1}| \leq \theta_i$, $|t_{j+1}-2u''-t_j| \leq \theta_j$ and $\|\ddot{k}\|_\infty =\tfrac12$, to estimate
\begin{equation}
\begin{aligned}
|G| &\leq \frac{1}{8} \left[ \frac{1}{2\pi} \int_0^{2\pi} \d t 
\left(\sum_{i=1}^n \theta_i 1_{[t_i,t_{i+1})}(t)\right) \right]^2
\leq \frac{1}{8} \frac{1}{2\pi} \int_0^{2\pi} \d t\,
\left[ \sum_{i=1}^n \theta_i 1_{[t_i,t_{i+1})}(t)\right]^2\\
&= \frac{1}{8} \frac{1}{2\pi} \int_0^{2\pi} \d t\, \sum_{i=1}^n \theta_i^2 1_{[t_i,t_{i+1})}(t) 
= \frac{1}{8} \frac{1}{2\pi} \sum_{i=1}^n \theta_i^3.
\end{aligned}
\end{equation}
Combine this with \eqref{Grepr1} and \eqref{Grepr2}, and note that $G \geq 0$, to get the claim in \eqref{EEintest}. 
\end{proof}

\begin{lemma}
\label{lem:stochdiscr1}
Put $\eps_n =\sum_{i=1}^n \theta_i^3$. Then, for every $s\in\R$ with $|s|< 1/\sqrt{2\pi \eps_n}$, 
\be 
\E\left[\exp\left[ \tfrac12 s \left(\sum_{i=1}^n 
\overline{B_{t_i}}^2 \theta_i - \int_0^{2\pi} \d t\, B_{t}^2 \right) \right] \right] 
\leq \exp\Bigl(|s|\,O\bigl(\sqrt{\eps_n} \log(1/\eps_n)\bigr) \Bigr).
\ee
\end{lemma}

\begin{proof} 
The proof uses a determinant formula for exponent moments of quadratic functionals in abstract Wiener spaces (see \eqref{quadmom} below). 
		
We view $ B$ as a random variable taking values in the Banach space $E$ of mean-zero, $2\pi$-periodic, continuous functions $f\colon\,\R \to \R$, equipped with the supremum norm $\|\cdot\|_\infty$. The scalar product $\langle f,g\rangle_\cH = \int_0^{2\pi} \d t\, f'(t)g'(t)$ turns $\cH \subset E$ into a real Hilbert space. For $\mu$ a finite signed measure on $[0,2\pi]$ with total mass zero, define $(-\Delta)^{-1} \mu\in\cH$ by 
\be
(-\Delta)^{-1}\mu(t) = \int_0^{2\pi}  \d \mu(t')\,k(t-t'), \qquad t\in \R, 
\ee
with $k(t-t')$ the Green function from Lemma~\ref{lem:gaussproc}. Following the proof of Lemma~\ref{lem:gaussproc}(b), we can check that $g= (-\Delta)^{-1} \mu$ is in $\cH$ and satisfies $-g'' = \mu$ in a distributive sense: 
\be
\label{dualemb}
\mu(f) = \int_0^{2\pi} \d \mu \,f = \int_0^{2\pi} \d t\,f(t)(-g''(t)) 
=  \langle g,f\rangle_\cH = \langle (-\Delta)^{-1}\mu, f\rangle_\cH, \qquad f\in \cH.
\ee
When $f\in E\setminus \cH$, we define $\langle g, f\rangle_\cH$ as $\mu(f)$, so that \eqref{dualemb} remains true.  A slight adaptation of the proof of Lemma~\ref{lem:gaussproc}(b) then shows that
\be 
\label{wiener} 
\E\bigl[ \e^{\mathrm{i} \langle g,  B\rangle_\cH}\bigr] 
= \E\bigl[ \e^{\mathrm{i} \mu( B)}\bigr] 
= \e^{-\frac{1}{2} \int_0^{2\pi} \d\mu(t)\,g(t)} = \e^{-\frac{1}{2} \|g\|_\cH^2}.
\ee
Now consider the continuous quadratic form
\be 
\label{qdef}
Q(f,f) = \sum_{i=1}^n \overline{f(t_i)}^2 \theta_i 
- \int_0^{2\pi} \d t\,f(t)^2, \qquad f\in E. 
\ee
The restriction of $Q$ to $\cH$ is represented by a bounded symmetric operator $\tilde Q\colon\,\cH\to\cH$ that is uniquely defined by the requirement that
\be
Q(h,h) = \langle h,\tilde Q h\rangle_\cH, \qquad h\in\cH.
\ee
From \eqref{tausup} and the first line of \eqref{XYapprox} we know that 
\be
\label{qerror}
\bigl|Q(h,h) \bigr| \leq \sqrt{\eps_n}\, \|h\|_\infty \|\dot{h}\|_2
\leq \sqrt{2\pi \eps_n}\, \|\dot{h}\|_2^2 = \sqrt{2\pi \eps_n}\, \|h\|_\cH^2, 
\qquad h\in\cH.
\ee
Hence $\|\tilde Q\| \leq \sqrt{2\pi \eps_n}$ and so, for all $s\in \R$ with $|s| \leq 1/\sqrt{2\pi \eps_n}$, the operator $\mathrm{id} - s \tilde Q$ is positive definite and invertible in $\cH$, with bounded inverse. 
	
Next, we check that $\tilde Q$ is a trace-class operator (Simon~\cite{simon-tracebook}). From\eqref{qdef} we see that $\tilde Q$ is a difference of two terms. The first term corresponds to the finite sum in \eqref{qdef}, has finite rank, and is therefore trivially trace class. The second term is $(-\Delta)^{-1}$. Indeed, for $f\in\cH$, setting $F= (-\Delta)^{-1} f \in\cH$ and integrating by part, we have
\be
\int_0^{2\pi} \d t\,f(t)^2 = \int_0^{2\pi} \d t\,f(t) (- F''(t)) 
= \int_0^{2\pi} \d t\, f'(t) F'(t) = \langle f, (-\Delta)^{-1} f\rangle_\cH.
\ee
Consider the orthonormal basis of $\cH$ consisting of the vectors
\be
e_k(t) = \frac{\cos (kt) }{k\sqrt{\pi}}, \quad e^*_k(t) = \frac{\sin(k t)}{k\sqrt{\pi}}, 
\qquad k\in \N.
\ee	
Then $(-\Delta)^{-1} e_k  = \frac{1}{k^2} e_k$ and $(-\Delta)^{-1} e^*_k = \frac{1}{k^2} e^*_k$. Thus, $(-\Delta)^{-1}$ is trace class with trace $2 \sum_{k=1}^\infty 1/k^2 <\infty$, and $\tilde Q$, as a difference of two trace class operators, is itself trace class. It follows from general results on abstract Wiener spaces (Chiang, Chow and Lee~\cite{chiang-chow-lee94}) that, for every $s\in \R$ with $|s|< (2\pi \eps_n)^{-1}$, the exponential moments of $Q( B, B)$ are given by 
\be
\label{quadmom}
\E\Big[\exp\big[\tfrac{1}{2} s Q( B, B)\big]\Big] = \det (\mathrm{id} -s \tilde Q)^{-1/2},
\ee
where the determinant is a Fredholm determinant. Indeed, if $\lambda_1\geq\lambda_2\geq \cdots$ are the eigenvalues of $\tilde Q$, then $\max_{j\in\N} |\lambda_j| = \|\tilde Q\| = O(\sqrt{\eps_n})$, and the Fredholm determinant is given by 
\be
\label{fredholm} 
\det (\mathrm{id} -s \tilde Q)  = \prod_{j\in\N} (1- |s|\lambda_j) 
= \exp\Bigg[|s|\, O\bigg(\sum_{j\in\N} \lambda_j\bigg)\Bigg] 
= \exp\Bigl(O(|s|\, \mathrm{Tr}\, \tilde Q)\Bigr).
\ee
Thus, it remains to estimate the trace of $\tilde Q$. To that aim we apply the first line of \eqref{XYapprox} for $\tau(t) = e_k(t)$. Since $\|e_k\|_2=1/k$ and $\|\dot e_k\|_\infty = 1/\sqrt{\pi}$, this gives
\be
\bigl|\langle e_k, \tilde Q e_k\rangle_\cH \bigr| = \bigl| Q(e_k,e_k)\bigr| 
\leq \sqrt{\eps_n}\, \|e_k\|_\infty \|\dot e_k\|_2
= \frac{\sqrt{\eps_n/\pi}}{k}.
\ee
On the other hand, clearly $Q(e_k,e_k) \leq 4 \pi \|e_k\|_\infty^2 = 4/ k^2$. Similar bounds hold for $Q(e^*_k, e^*_k)$. Therefore 
\be
\bigl|\mathrm{Tr}\, \tilde Q\bigr| \leq 2 \sum_{k\in\N} \min\bigl((\sqrt{\eps_n/\pi})/k, 4/k^2\bigr) 
= O\bigl(\sqrt{\eps_n}\log(1/\eps_n)\bigr).
\ee
The  proof is concluded with~\eqref{quadmom} and~\eqref{fredholm}.
\end{proof} 


\paragraph{Fourier coefficients.}

Define the discretised Fourier coefficients 
\be
\label{difou}
D_1= \frac{1}{\sqrt{\pi}} \sum_{i=1}^n  \theta_i\,	\overline{B_{t_i}\cos t_i},\qquad 
D^*_1 = \frac{1}{\sqrt{\pi}} \sum_{i=1}^n  \theta_i\,\overline{B_{t_i}\sin t_i}.
\ee
By a slight abuse of notation we shall use the same letter $D_1$, $D_1^*$ for the random variable obtained by substituting $t_i\to T_i$, $\theta_i\to \Theta_i$. 

\begin{lemma}
\label{lem:stochdiscr2}
Put $\eps_n =\sum_{i=1}^n \theta_i^3$. Then, for every $|s|<1/8\sqrt{2\pi\eps_n/3}$,
\be 
\E\Bigl[\exp\Bigl( \tfrac{1}{2} s [A_1^2 - D_1^2] \Bigr) \Bigr] 
\leq \exp\Bigl( |s|\, O\bigl(\sqrt{\eps_n} \log(1/\eps_n)\bigr) \Bigr).
\ee
A similar bound holds for $A_1^{*2} - {D_1^*}^2$. 
\end{lemma}

\begin{proof} 
The proof is similar to that of Lemma~\ref{lem:stochdiscr1}. Let 
\be
Q(f,f) = \Bigl(\frac{1}{\sqrt{\pi}}\sum_{i=1}^n \theta_i\,\overline{f(t_i) \cos t_i}\,\Bigr)^2 
- \Bigl(\frac{1}{\sqrt{\pi}} \int_0^{2\pi} \d t\,f(t)\,\cos t\Bigr)^2 
\ee
and let $\tilde Q\colon\,\cH\to\cH$ be the associated symmetric operator. Then $\tilde Q$ has finite rank and is therefore trace-class. Using the identity $|x^2 - y^2| = |x+y|\,|x-y|$, $x,y\in\R$, in combination with \eqref{tausup} and the second line of \eqref{XYapprox}, we see that for all $h\in\cH$, 
\be
|Q(h,h)| \leq 4\|h\|_\infty \left|\sum_{i=1}^n \theta_i\,\overline{h(t_i) \cos t_i}
- \int_0^{2\pi} \d t \, h(t)\,\cos t\,\right| \leq 8 \sqrt{\eps_n/3}\,\|h\|_\infty\,\|\dot{h}\|_2
\leq 8 \sqrt{2\pi \eps_n/3}\, \|h\|_\cH^2.
\ee
Consequently, $\|\tilde Q\|\leq 8\sqrt{2\pi \eps_n/3}$, $\mathrm{id} - s \tilde Q$ is positive definite for all $|s|< 1/8\sqrt{2\pi \eps_n/3}$, and
\be
|Q(e_k,e_k)| \leq 8 \sqrt{2\pi\eps_n/3}\,\|e_k\|_\infty\,\|\dot{e}_k\|_\infty 
= 8 \sqrt{2\eps_n/3}/k = O(\sqrt{\eps_n}/k),
\ee		 
and similarly for $e^*_k$. The proof is concluded in the same way as that of Lemma~\ref{lem:stochdiscr1}. 
\end{proof}


\section{Asymptotics of surface integrals II: proof of moderate deviations}
\label{proofmoddev}

In this section we collect the preparations made in Sections~\ref{app:geometry}--\ref{sec:prep}. Section~\ref{mdub} provides an upper bound for the key surface integral $\cI^\mathrm{UB}(\kappa,\beta; C,\eps)$ (Proposition~\ref{prop:upperkey}). The proof consists of a sequence of steps involving decomposition, elimination and separation of terms. A particularly delicate point is how to control the integral over $m$ in \eqref{kikihat}: we show that only small values of $m$ contribute, namely, $|m|=O(\beta^{1/6})$. Section~\ref{sec:wdmp} provides a lower bound for the key surface integral $\cI^\mathrm{LB}(\kappa,\beta; C,\eps)$ (Proposition~\ref{prop:lowerkey}). In Section~\ref{conds} we formulate \emph{three technical conditions} (C1)--(C3) in terms of the auxiliary random variables introduced in Section~\ref{sec:auxiliary}, and use Propositions~\ref{prop:upperkey}--\ref{prop:lowerkey} to show that under those conditions Conjecture~\ref{thm:zoom2} is true. In Section~\ref{sec:moddevbds} we prove Theorem~\ref{thm:zoom1}. 


\subsection{Upper bound on the key integral}
\label{mdub}

\begin{proposition}[Upper bound key integral] 
\label{prop:upperkey}
For every $C \in (0,\infty)$ and all $\eps>0$ sufficiently small, there exists a $c= c(C,\eps)>0$ such that
\be
\begin{aligned} 
\label{ubkey} 
&\limsup_{\beta \to \infty}\frac{1}{\beta^{1/3}} \log \mathcal I^\mathrm{UB}(\kappa,\beta; C,\eps)\\
&\qquad \leq  2\pi G_\kappa \tau^{*} + \limsup_{\beta \to \infty}\sup_{|m|\leq c\beta^{1/6}}  \frac{1}{\beta^{1/3}}  
\log \widehat \E \Bigl[\e^{\tfrac12 [1+O(\eps)]\chi}\,\1_{\{Z^{(m)} \in \mathcal O\}}\Bigr] + O(\eps),
\end{aligned}
\ee 
where $\chi=\sum_{i=1}^N \Theta_i \overline{B_{T_i}}^2 - D_1^2 - {D_1^*}^2$ and $D_1,D_1^*$ are defined in \eqref{difou}.
\end{proposition} 

\begin{proof}
The representation in \eqref{kikihat} for the key integral and the asymptotics from Proposition~\ref{prop:theta-asympt} imply that it suffices to show that
\be
\label{enough1}
\begin{aligned}
&\limsup_{\beta \to \infty}\frac{1}{\beta^{1/3}} \log \Biggl( \int_{\mathbb R} \mathrm \d m\,\,\widehat 
\E\Bigl[ \e^{O(\eps) (\beta Y_1 + Y_2  + N) + \frac12 Y_3^{(m)}} \1_{\{Z^{(m)}\in 
\cV_{C\beta^{-2/3}} \cap \cD_\eps(0) \cap \cC_{C\beta^{-2/3}}(0)\}} \Biggr)\\
&\leq O(\eps) + \limsup_{\beta \to \infty}\sup_{|m|\leq c\beta^{1/6}}  \frac{1}{\beta^{1/3}}  
\log \widehat \E \Bigl[\e^{\tfrac12 [1+O(\eps)]\chi}\,\1_{\{Z^{(m)} \in \mathcal O\}}\Bigr].
\end{aligned}
\ee
The main idea is the following. The error terms $O(\eps)(Y_1+ Y_3+N)$ should be negligible, and therefore our primary concern is $\frac12 Y_3^{(m)}$ in the exponential. In order to deal with this term, we approximate
\be
Y_3^{(m)} = \sum_{i=1}^N (m+\overline{B}_{T_i})^2 \Theta_i \approx \int_0^{2\pi}(m+B_t)^2 \d t 
= 2\pi m^2 + \int_0^{2\pi} B_t^2 \d t
\ee
(recall that $\int_0^{2\pi} B_t \d t=0$). The resulting expression is  problematic because 
\be
\int_\R \d m\, \e^{\pi m^2 } =\infty, \qquad \E\Bigl[\e^{\tfrac12 \int_0^{2\pi } B_t^2 \d t}\Bigr] = \infty. 
\ee
To cure the divergence, we use the geometric constraints. Roughly speaking, we show that the volume constraint imposes that the only relevant contributions are from $|m| =O(\beta^{1/6})$. In addition, we show that the centring constraint imposes that the Fourier coefficients $D_1$ and $D_1^*$ are negligible, so that we may replace $\sum_{i=1}^N \overline{B_{T_i}}^2\Theta_i$ by $\chi=\sum_{i=1}^N \overline{B_{T_i}}^2\Theta_i - D_1^2 - {D_1^*}^2$.

The proof comes in 8 steps. 

\medskip\noindent
{\bf 1. Decomposition of $Y_3^{(m)}$.} 
As a  preliminary step we decompose $Y_3^{(m)}$ as 
\be
\label{eq:y3deco}
Y_3^{(m)} =  \mathcal E_1^{(m)} +\mathcal {E}_2   - \mathcal E_3 +  \chi
\ee 
with (recall that $\overline{B_{T_i}} = \frac12 (B_{T_i}+ B_{T_{i+1}})$
\begin{equation}
\mathcal E_1^{(m)} =  \frac{1}{2\pi} \left(\sum\limits_{i=1}^N (m+\overline{B_{T_i}}\,) \Theta_i\right)^2,\quad
\mathcal E_2 = D_1^2 + {D_1^*}^2,\quad
\mathcal E_3 = \frac{1}{2\pi} \left(\sum\limits_{i=1}^N \overline{B_{T_i}} \Theta_i\right)^2.
\end{equation} 
To check~\eqref{eq:y3deco}, we apply the variance formula 
\be 
\label{eq:steinervariance}
\sum_{i=1}^n \left( x_i - \sum_{j=1}^n x_j p_j\right)^2 p_i = \sum_{i=1}^n x_i^2 p_i 
-  \left(\sum_{j=1}^n x_j p_j\right)^2 
\ee
with $n= N$, $p_i = \frac{\Theta_i }{2\pi}$ and $x_i = m + \overline{B_{T_i}}$, respectively, $x_ i = \overline{B_{T_i}}$, to get 
\begin{equation}
\begin{aligned}
Y_3^{(m)} & = \sum_{i=1}^N (m+ \overline{B_{T_i}}\,)^2 \Theta_i 
= \sum_{i=1}^N\left( \overline{B_{T_i}} - \frac{1}{2\pi} \sum_{j=1}^N  \overline{B_{T_j}} \Theta_j  \right)^2  \Theta_i
+ \mathcal E_1^{(m)}\\
& = \sum_{i=1}^N\overline{B_{T_i}}^2 \Theta_i - \mathcal E_3 +  \mathcal E_1^{(m)}.
\end{aligned}
\end{equation}
The claim in \eqref{eq:y3deco} now follows from the relation
\begin{equation}
\label{eq:discretesquareintegral}
\sum_{i=1}^N\overline{B_{T_i}}^2 \Theta_i = \mathcal E_2 + \chi. 
\end{equation}
Note that $\mathcal E_1^{(m)}$, $\mathcal E_2$, $\mathcal E_3$ are non-negative, while $\chi$ is not necessarily so. The terms $\mathcal E_1^{(m)}$ and $\mathcal E_2$ will be taken care of in Steps 2 and 4 via the volume and centering constraints. The non-negative term $\mathcal E_3$ can simply be dropped because the target in \eqref{enough1} is an upper bound.

\medskip \noindent 
\textbf{2. Elimination of $\mathcal E_1^{(m)}$ with the help of the volume constraint.} 
Next, we exploit the volume constraint and the a priori estimates to get rid of $\mathcal E_1^{(m)}$. By Proposition~\ref{prop:expansion}, recalling the definitions in Proposition~\ref{prop:keyrep}, we have 
\begin{equation} 
\label{volYk}
\begin{aligned}
|S(Z^{(m)})| - \pi \Rc^2 & = - C_1^\eps Y_1 + \frac{1+O(\eps)}{2 \beta}Y_2 
+ \frac12 \frac{1}{(\kappa -1) \beta} Y_3^{(m)} + \Rc \sum_{i=1}^N \chi_i^{(m)}\Theta_i \\
& = - C_1^\eps Y_1 + \frac{1+O(\eps)}{2\beta}Y_2 
+ \frac12 \sum_{i=1}^N \bigl(\chi_i^{(m)} + \Rc\bigr)^2 \Theta_i - \pi \Rc^2, 
\end{aligned}
\end{equation}
where we abbreviate 
\begin{equation}
\chi_i^{(m)} = \frac{m+ \overline{B_{T_i}}}{\sqrt{(\kappa-1)\beta}}.
\end{equation}
By the triangle inequality, on the event that $||S(Z^{(m)})| - \pi \Rc^2|\leq C \beta^{-2/3}$ we have 
\begin{equation}
\label{con1}
\left|\frac12 \sum_{i=1}^N \bigl(\chi_i^{(m)} + \Rc\bigr)^2 \Theta_i 
- \pi \Rc^2\right| \leq C \beta^{-2/3} + C_1^\eps Y_1 +  \frac{1+O(\eps)}{2\beta}Y_2.
\end{equation}
An elementary computation based again on the variance formula in \eqref{eq:steinervariance} gives 
\begin{equation}
\begin{aligned}
\sum_{i=1}^N \bigl(\chi_i^{(m)} + \Rc\bigr)^2 \Theta_i 
&= 2\pi  \left(\Rc + \sum_{i=1}^N \chi_i^{(m)} \frac{\Theta_i}{2\pi}\right)^2 
+\frac{1}{(\kappa-1)\beta}  \sum_{i=1}^N \left(\overline{B_{T_i}} - \frac{1}{2\pi} 
\sum_{j=1}^N  \overline{B_{T_j}} \Theta_j  \right)^2  \Theta_i \\
& = 2\pi  \left( \Rc + \sum_{i=1}^N  \chi_i^{(m)} \frac{\Theta_i}{2\pi}\right)^2 
+\frac{1}{(\kappa-1) \beta}  \bigl(\mathcal E_2 +  \chi - \mathcal E_3\bigr). 
\end{aligned}
\end{equation}
Consequently, \eqref{con1} becomes
\be
\label{con2}
\left|\pi\left(\Rc + \sum_{i=1}^N  \chi_i^{(m)} \frac{\Theta_i}{2\pi}\right)^2 
- \pi \Rc^2 \right| \leq C \beta^{-2/3} + C_1^\eps Y_1 +  \frac{1+ O(\eps) }{2\beta}Y_2 
+ \frac{1}{2(\kappa-1)\beta}  \bigl(\mathcal E_2 +  \chi - \mathcal E_3\bigr).
\ee
From the a priori estimate in Proposition~\ref{prop:apriori} we have $\rho_i =O(\eps)$ and 
\be 
\label{eq:apriorimean}
\left| \sum_{i=1}^N  \chi_i^{(m)} \frac{\Theta_i}{2\pi}\right| = O(\eps). 
\ee
Therefore 
\be
\pi \left( \Rc + \sum_{i=1}^N \chi_i^{(m)} \frac{\Theta_i}{2\pi}\right)^ 2 - \pi \Rc^2 
= [2\pi \Rc + O(\eps)] \sum_{i=1}^N  \chi_i^{(m)} \frac{\Theta_i}{2\pi},
\ee
which together with \eqref{con2} gives
\be 
\label{eq:vocomean}
\beta^{-1/2} \sum_{i=1}^N (m + \overline{B_{T_i}})\, \Theta_i \leq  O(1) \Bigl( \beta^{-2/3} +  Y_1 + \beta^{-1} Y_2 
+ \beta^{-1}\bigl(\mathcal E_2 +  \chi - \mathcal E_3\bigr)\Bigr). 
\ee
Combine the estimates in \eqref{eq:apriorimean} and~\eqref{eq:vocomean} to obtain 
\begin{equation}
\begin{aligned}
\mathcal E_1^{(m)} &= 2\pi \beta \left( \beta^{-1/2}  \sum_{i=1}^N (m + \overline B_{T_i}) \Theta_i\right)^2 
\leq 2 \pi \beta  O(\eps) \left( \beta^{-1/2}  \sum_{i=1}^N (m + \overline B_{T_i}) \Theta_i\right) \\
&\leq O(\eps) \Bigl( \beta^{1/3} +  \beta Y_1 +Y_2 + \bigl(\mathcal E_2 +  \chi - \mathcal E_3\bigr)\Bigr). 
\end{aligned}
\end{equation}
Insert this estimate into~\eqref{eq:y3deco} and drop the term $\mathcal E_3$, to find 
\be 
\label{eq:y3deco2}
Y_3^{(m)}  \leq  O(\eps) \bigl( \beta^{1/3} +  \beta Y_1 +Y_2 \bigr)+  [1+O(\eps)] 
\bigl(\mathcal E_2 + \chi\bigr).
\ee

\medskip\noindent
{\bf 3. Estimation of $m^2$.} 
Next we estimate $m^2$, which will be needed later. Write 
\begin{equation}
m = \frac{1}{2\pi}\sum_{i=1}^N (m+\overline{B_{T_i}}\,\bigr)\, \Theta_i -\frac{1}{2\pi} 
\sum_{i=1}^N  \overline{B_{T_i}} \Theta_i 
\end{equation}
and use $(a-b)^2\leq 2(a^2+b^2)$, to estimate
\begin{equation}
m^2 \leq 2 \left( \frac{1}{2\pi}  \sum_{i=1}^N (m+\overline{B_{T_i}}\,\bigr) \Theta_i \right)^2
+ 4\pi  \left(  \sum_{i=1}^N  \overline{B_{T_i}}\,\frac{ \Theta_i }{2\pi}\right)^2.
\end{equation}
Up to a multiplicative constant, the first term is equal to $\mathcal E_1^{(m)}$, which has been estimated in Step 2. For the second term we use Cauchy-Schwarz and \eqref{eq:discretesquareintegral}. Hence
\begin{align}
m^2 & \leq\frac1\pi  \mathcal E_1^{(m)} + \frac1\pi  \sum_{i=1}^N \overline{B_{T_i}}^2 \Theta_i 
\leq \frac1\pi \bigl( \mathcal E_1^{(m)} + \mathcal E_2 + \chi\bigr) \notag \\
&\leq   O(\eps) \bigl( \beta^{1/3} +  \beta Y_1 +Y_2 \bigr) 
+ [1+O(\eps)] \bigl(  \mathcal E_2 + \chi\bigr). 
\label{mesti}
\end{align}

\medskip\noindent
{\bf 4. Elimination of $\mathcal E_2$ with the help of the centering constraint.} 
Next we exploit the centering constraint and the a priori estimates  to get rid of $\mathcal E_2= D_1^2+ {D_1^*}^2$. We estimate $D_1^2$ only, since ${D_1^*}^2$ can be treated analogously. We have 
\begin{equation}
D_1 = \frac{1}{\sqrt{\pi}}\left( \sum_{i=1}^N  \overline{(m + B_{T_i}) \cos T_i}\,\Theta_i 
-  m \sum_{i=1}^N \overline{\cos T_i}\, \Theta_i\right). 
\end{equation}
Hence
\be
D_1^2 \leq\frac{2}{\pi} \left(  \sum_{i=1}^N  \overline{(m + B_{T_i}) \cos T_i}\, \Theta_i \right)^2 
+ \frac2\pi \left( \sum_{i=1}^N \overline{\cos T_i}\, \Theta_i \right)^2 m^2.
\ee
In the first sum we use the a priori estimate $\rho_i = ([\kappa -1] \beta)^{-1/2} (m+ B_{T_i}) = O(\eps)$. In the second sum we use that $\sum_{i=1}^N \overline{\cos T_i}\,\Theta_i = O( Y_1)$ by Lemma~\ref{ordermagn1}, and we use the a priori estimate $\theta_i = O(\sqrt{\eps}\,)$ together with $\sum_{i=1}^N \theta_i = 2\pi$ to estimate $O( Y_1)= O(\eps)$. Hence 
\be
\label{eq:d1esti1}
D_1^2  \leq O(\eps)\, \beta^{1/2} \sum_{i=1}^N  \overline{(m + B_{T_i}) \cos T_i}\,\Theta_i + O(\eps^2)\, m^2. 
\ee
The term $m^2$ appearing in the right-hand side has been estimated in~\eqref{mesti}. For the first term, we exploit the centreing constraint. Define
\be
\label{y4mdef}
Y_4^{(m)} = \sum_{i=1}^N (m+ \overline{B_{T_i}}\,)\,\Theta_i,
\ee
which has already been estimated in \eqref{eq:vocomean}. By Lemma~\ref{lem:aprvolsurcen}, recalling Definition~\ref{def:yzsums} and the notation in Proposition~\ref{prop:keyrep}, we have $\cC(Z^{(m)}) = (\Sigma_1,\Sigma_2)$ with 
\be 
\label{eq:sigma1}
\Sigma_1 = \frac1\pi \frac{1}{\sqrt{(\kappa-1)\beta}}\, Y_5^{(m)},
\ee
where
\be\label{y5new}
Y_5^{(m)} = \sum_{i=1}^N \overline{(m+B_{T_i}) \cos T_i}\,\Theta_i.  
\ee
By~\eqref{eq:sigma1}, on the event that $|\Sigma_1|\leq C\beta^{-2/3}$ we have 
\be
\label{Y5control}
\beta^{-1/2} Y_5^{(m)} \leq O(\beta^{-2/3}).
\ee
Multiply both sides by $\beta$ and combine with \eqref{eq:d1esti1}, to get 
\be
\label{D12bd}
D_1^2  \leq O(\eps)\, \beta^{1/3} + O(\eps^2)\, m^2.
\ee
A similar estimate holds for ${D_1^ *}^2$. Combine \eqref{D12bd} with the bounds in \eqref{eq:apriorimean}, \eqref{eq:vocomean}, \eqref{eq:y3deco2} and~\eqref{mesti} for $Y_3^{(m)}$, $|Y_4^{(m)}|$ and $m^2$, to obtain 
\be 
\label{eq:E2esti}
\mathcal E_2 \leq O(\eps) \bigl( \beta^{1/3} + \beta Y_1 + Y_2 + \chi \bigr) 
+ O(\eps)\, \mathcal E_2. 
\ee
We subtract $O(\eps)\,\mathcal E_2$ on both sides and multiply by $[1- O(\eps)]^{-1} = 1+ O(\eps)$, to conclude that we may drop the term $O(\eps)\,\mathcal E_2$ from \eqref{eq:E2esti}. Finally, we insert the estimate thus obtained into the bound \eqref{eq:y3deco2} for $Y_3^{(m)}$, to find 
\be 
\label{eq:y3deco3}
Y_3^{(m)}  \leq  O(\eps) \bigl( \beta^{1/3} +  \beta Y_1 +Y_2 \bigr)
+  [1+O(\eps)]\,\chi.
\ee

\medskip\noindent 
{\bf 5. Only small $m$ contribute.} 
From \eqref{mesti} and \eqref{eq:E2esti} we get 
\be
\label{mesti2}
m^2 \leq M^\eps, \qquad   M^\eps =  O(\eps) \bigl( \beta^{1/3} +  \beta Y_1 +Y_2 \bigr) 
+ [1+O(\eps)]\,\chi.
\ee
We may think of $M^\eps$ as a random variable that is typically of order $\beta^{1/3}$, so that $m$ is typically of order $\beta^{1/6}$ (at most). However, we need to estimate the contribution of the event that $M^\eps$ is much larger than its typical order of magnitude. Fix $C'>0$ (to be chosen later). In what follows we abbreviate
\be 
\label{Upsdef}
\Upsilon^{\UB'}_{\beta,\eps} = \cV_{C\delta(\beta)} \cap \cC_{C\delta(\beta)}(0) \cap \cD_\eps(0),
\qquad
\Upsilon^{\LB'}_{\beta,\eps} = \cV_{C\delta(\beta)} \cap \cD_\eps(0).
\ee
By \eqref{mesti2}, we have 
\begin{align}
&\int_{\mathbb R} \mathrm d m\,  \1_{\{m^2> C'\beta^{1/3}/\eps\}}     
\widehat \E\Bigl[ \e^{O(\eps) (\beta Y_1 + Y_2  + N) 
+ \frac12 Y_3^{(m)}} \1_{\{Z^{(m)} \in \Upsilon^{\UB'}_{\beta,\eps}\}}\Bigr] \notag \\
& \qquad \leq  \widehat \E\Biggl[ \int_{\mathbb R} \mathrm d m\,  
\1_{\{C'\beta^{1/3}/\eps < m^2 \leq M^\eps \}} \e^{O(\eps) (\beta Y_1 + Y_2  + N) 
+ \frac12 Y_3^{(m)}} \1_{\{Z^{(m)} \in \Upsilon^{\UB'}_{\beta,\eps}\}}\Biggr] \notag \\
&\qquad \leq   \widehat \E\Biggl[ \int_{\mathbb R} \mathrm d m\,  
\1_{\{C'\beta^{1/3}/\eps < m^2 \leq M^\eps \}} 
\e^{O(\eps) (\beta^{1/3}+\beta Y_1 + Y_2  + N)  
+ \frac12 [1+O(\eps)]\,\chi} 
\1_{\{Z^{(m)} \in \Upsilon^{\UB'}_{\beta,\eps}\}}\Biggr]. 
\label{mesti3}
\end{align}
In the last line we first use the bound on $Y_3^{(m)}$ from \eqref{eq:y3deco3} and then drop the indicator on $m$. The exponential in the last line is independent of $m$. We bound $\mathbf 1_{\{C'\beta^{1/3}/\eps < m^2 \leq M^\eps\}} \leq \1_{\{M^\eps > C'\beta^{1/3}/\eps\}} \1_{\{m^2  \leq M^\eps\}}$ and perform the integral over $m$, to find that we can further bound \eqref{mesti3} by 
\be\label{mesti4}
\widehat \E\Bigl[\1_{\{M^\eps >  C'\beta^{1/3}/\eps\}} 
\sqrt{M^\eps}\,\e^{O(\eps) (\beta^{1/3}+\beta Y_1 + Y_2  + N)  
+ \frac12 [1+O(\eps)]\,\chi}
\1_{\{Z^{(m)} \in \Upsilon^{\UB'}_{\beta,\eps}\}}\Bigr].
\ee
We can get rid of $\sqrt{M_\eps}$ via the inequality $x \leq \e^{x-1}$, $x\in\R$, with $x = \eps M_\eps$, which yields $\sqrt{M_\eps} \leq \frac{1}{\sqrt{\eps \mathrm e}} \exp(\frac12 \eps\, M^\eps)$. Because of \eqref{eq:y3deco}, the term $\frac12 \eps M^\eps$ can be absorbed into the exponential. We can get rid of the indicator of the event $\{M^ \eps > C'\beta^{1/3}/\eps\}$ by estimating $\1_{\{M^\eps >C'\beta^{1/3}/\eps\}}\leq \exp(- C'\beta^{1/3} + \eps M^\eps)$ and again absorbing the term $\eps M^\eps$ into the exponential. Thus \eqref{mesti4} is bounded by
\be
\label{mesti5}
\frac{1}{\sqrt{\eps \mathrm e}}\,  \e^{-C' \eps \beta^{1/3}} 
\widehat \E\Bigl[\e^{O(\eps) (\beta^{1/3}+\beta Y_1 + Y_2  + N)  
+ \frac12 [1+O(\eps)]\,\chi}\Bigr].
\ee
Using H{\"o}lder's inequality and Lemmas~\ref{lem:ehatperturbed}, \ref{lem:varadhanprep} and \ref{lem:y3moments}, we get
\be \label{mesti5b}
\limsup_{\beta\to \infty} \frac{1}{\beta^{1/3}} \log 
\widehat \E\Bigl[ \e^{O(\eps) (\beta^{1/3}+\beta Y_1 + Y_2  + N)  
+ \tfrac12[1+O(\eps)]\,\chi}
\Bigr] 
\leq k(\eps)
\ee
for some $k(\eps)<\infty$. The details are similar to Steps 6--7 below and therefore are omitted. Hence, given $\eps>0$ we can make \eqref{mesti5} arbitrarily small by making $C'$ sufficiently large. Altogether we obtain the following statement of exponential tightness: For every $C''>0$ there exists a $C'=C'(\eps,C,C'')>0$ such that 
\be
\label{mesti6}
\limsup_{\beta \to \infty} \frac{1}{\beta^{1/3}} \log 
\Biggl(\int_{\mathbb R} \mathrm d m\,  \1_{\{m^2> C' \beta^{1/3}/\eps\}}     
\widehat \E\Bigl[ \e^{O(\eps) (\beta Y_1 + Y_2  + N) + \frac12 Y_3^{(m)}} 
\1_{\{Z^{(m)} \in \Upsilon^{\UB'}_{\beta,\eps}\}}\Bigr]\Biggr) \leq - C''.
\ee
Hence we need only estimate contributions coming from $|m| \leq \sqrt{C'/\eps}\,\beta^{1/6}$.

\medskip\noindent
\textbf{6. Separation of terms with the H{\"o}lder inequality.} 
By~\eqref{eq:y3deco3} and the H{\"o}lder inequality, we have for all $c>0$ and $p,q\geq 1$ with $p^{-1}+q^{-1}=1$, 
\begin{multline} 
\label{keysplit1}
 \log \widehat \E\left[ \int_{|m|\leq c\beta^{1/6}} \d m \, 
 \e^{O(\eps) (\beta Y_1 + Y_2  + N) + \frac12 Y_3^{(m)}} 
\1_{\{Z^{(m)} \in \Upsilon^{\UB'}_{\beta,\eps}\}}\right] \\
\leq  \frac1p \sup_{|m|\leq c \beta^{1/6}} \log \widehat \E\Bigl[ \e^{\frac12 p[1+O(\eps)]\chi} 
\1_{\{Z^{(m)} \in \Upsilon^{\UB'}_{\beta,\eps}\}}\Bigr] \\
+ \frac1q \log\widehat \E\Bigl[ \e^{q O(\eps) (\beta^{1/3} + \beta Y_1 + Y_2  + N)}\Bigr] 
+ O(\eps \beta^{1/3}).
\end{multline}
(We have dropped the indicator in the second term, because it will not be needed.) We will want to choose $p$ close to $1$, which makes $q$ large and potentially dangerous for the second term in~\eqref{keysplit1}. It will turn out that a good choice is 
\be 
\label{eq:pqchoice}
q= \frac{c}{\sqrt{\eps}}, \qquad c \in (0,\infty).
\ee
for which $p=1+O(\sqrt{\eps}\,)$.

\medskip\noindent 
\textbf{7. Estimation of the second term in \eqref{keysplit1}.}
Note that $Y_1$ depends on $N$ and $\Theta_i$, $1 \leq i \leq N$, alone. The tilt by $\e^{Y_0 - \beta C_1 Y_1}$ in the definition of $\widehat\P$ affects the angular point process only, so it is still true under $\widehat{\mathbb {P}}$ that the distribution of $(B_t)_{t\in [0,2\pi]}$ is a mean-centred Brownian bridge independent from $N$ and the $\Theta_i$, $1 \leq i \leq N$. Therefore, by \eqref{Ydefs} and Lemma~\ref{lem:y3moments}, for every $s$ such that $s\in (-1,1)$, 
\be
\label{compmom}
\widehat \E\bigl[\e^{ \frac{1}{2} s Y_2 } \mid N,(\Theta_i)_{i=1}^N \bigr]  
 = (1- s)^{-(N-1)/2}\quad \widehat{\mathbb P}\text{-a.s.}
\ee
Applying this identity with $s = 2q O(\eps) = O( \sqrt{\eps}\,)$ (which falls in $(-1,1)$ for $\eps$ sufficiently small), we get 
\be 
\label{eq:ey2esti}
\widehat \E\Bigl[ \e^{q O(\eps) Y_2} \mid N,(\Theta_i)_{i=1}^N\Bigr] 
= \e^{O(\eps N)} \quad \widehat{\mathbb P}\text{-a.s.}
\ee
Multiplying both sides by $\e^{qO(\eps)(\beta Y_1+N)}$ and taking expectations, we find 
\be
\frac1q \log\widehat \E\Bigl[ \e^{q O(\eps) (\beta Y_1 + Y_2  + N)}\Bigr] 
\leq \frac1q \log\widehat \E\Bigl[ \e^{q O(\eps) (\beta Y_1 + N)}\Bigr]. 
\ee
It now follows from  Lemma~\ref{lem:ehatperturbed} with $\delta=q\eps= c\sqrt{\eps}$ that 
\be\label{dom}
\limsup_{\beta\to \infty}\frac{1}{\beta^{1/3}} \frac1q 
\log\widehat \E\Bigl[ \e^{q O(\eps) (\beta Y_1 + Y_2  + N) }\Bigr] \leq O(\eps). 
\ee
Hence the second term in \eqref{keysplit1} is negligible.

\medskip\noindent 
\textbf{8. Conclusion.} 
Combine~\eqref{mesti6}--\eqref{keysplit1} and \eqref{dom} to get 
\begin{multline} 
\limsup_{\beta \to \infty}\frac{1}{\beta^{1/3}} \log \Biggl( \int_{\mathbb R} 
\mathrm \d m\,\,\widehat \E\Bigl[ \e^{O(\eps) (\beta Y_1 + Y_2  + N) 
+ \frac12 Y_3^{(m)}} \1_{\{Z^{(m)}\in \Upsilon^{\UB'}_{\beta,\eps}\}}\Bigr] \Biggr) \\
\leq  O(\eps) + \limsup_{\beta \to \infty}\frac{1}{\beta^{1/3}} \sup_{|m|\leq c \beta^{1/6}} 
\log \widehat \E\Bigl[ \e^{\frac12[1+O(\eps)]\,\chi} \1_{\{Z^{(m)}\in \mathcal O\}}\Bigr]
\end{multline} 
for suitable $c= c(\eps, C)>0$. This completes the proof of Proposition~\ref{prop:upperkey}.
\end{proof}


\subsection{Lower bound on the key integral} 
\label{sec:wdmp}

The proof of the lower bound in Theorem~\ref{thm:zoom1} builds on the following key proposition. 

\begin{proposition}[Lower bound key integral]  
\label{prop:lowerkey}
For all $C\in (0,\infty)$ sufficiently large and all $\eps>0$ sufficiently small,
\be 
\label{lbkey} 
\liminf_{\beta \to \infty}\frac{1}{\beta^{1/3}} \log \mathcal I^\mathrm{LB}(\kappa,\beta; C,\eps)
\geq 2\pi G_\kappa \tau^{*} + \liminf_{\beta \to \infty}\inf_{|m|\leq \beta^{-1/6}}  
\frac{1}{\beta^{1/3}}  \log \widehat \P \bigl(Z^{(m)} \in \mathcal O\bigr) +  O(\eps).
\ee
\end{proposition} 

\begin{proof}
The representation in \eqref{kikihat} for the key integral and the asymptotics from Proposition~\ref{prop:theta-asympt} imply that it suffices to show that
\be
\label{enough2}
\begin{aligned}
&\liminf_{\beta \to \infty}\frac{1}{\beta^{1/3}} \log \Biggl( \int_{\mathbb R} \mathrm \d m\,\,\widehat 
\E\Bigl[ \e^{O(\eps) (\beta Y_1 + Y_2  + N) + \frac12 Y_3^{(m)}} \1_{\{Z^{(m)}\in 
\cV_{C\beta^{-2/3}} \cap \cD_\eps(0) \cap \cC_{C\beta^{-2/3}}(0)\}} \Biggr)\\
&\geq O(\eps) + \liminf_{\beta \to \infty}\inf_{|m|\leq \beta^{-1/6}}  
\frac{1}{\beta^{1/3}}  \log \widehat \P \bigl(Z^{(m)} \in \mathcal O\bigr).
\end{aligned}
\ee
The proof comes in 4 steps.


\paragraph{1. Separation of terms with the reverse H\"older inequality.} 
For the lower bound we can simply drop the non-negative term $Y_3^{(m)}$ and restrict the integral over $m$ to $|m|\leq \beta^{-1/6}$. Therefore it suffices to prove that 
\be
\label{enough*}
\begin{aligned}
&\liminf_{\beta \to \infty}\frac{1}{\beta^{1/3}} \log \Biggl( \int_{_{|m| \leq \beta^{-1/6}}} \mathrm \d m\,\,\widehat 
\E\Bigl[ \e^{O(\eps) (\beta Y_1 + Y_2  + N)} \1_{\{Z^{(m)}\in 
\cV_{C\beta^{-2/3}} \cap \cD_\eps(0) \cap \cC_{C\beta^{-2/3}}(0)\}} \Biggr)\\
&\geq O(\eps) + \liminf_{\beta \to \infty}\inf_{|m|\leq \beta^{-1/6}}  
\frac{1}{\beta^{1/3}}  \log \widehat \P \bigl(Z^{(m)} \in \mathcal O\bigr).
\end{aligned}
\ee
We separate the exponential from the indicator  $\Upsilon^{\LB'}_{\beta,\eps}$ (defined in \eqref{Upsdef} with the help of the reverse H\"older inequality with $p\in(1,\infty)$,
\begin{multline} 
\label{keysplit}
\log  \widehat \E\Bigl[\int_{|m| \leq \beta^{-1/6}} \d m\,  \e^{O(\eps) (\beta Y_1 + Y_2  + N) } 
\1_{\{Z^{(m)}\in \Upsilon^{\LB'}_{\beta,\eps}  \}}\Bigr] \\
\geq p \inf_{|m|\leq \beta^{-1/6}} \log \widehat{ \mathbb P}\bigl(Z^{(m)}\in \Upsilon^{\LB'}_{\beta,\eps} \bigr)
- (p-1) \sup_{|m|\leq \beta^{-1/6}} \log \widehat \E \Bigl[ \e^{-\frac{1}{p-1}O(\eps) (\beta Y_1 + Y_2 + N)} \Bigr] 
+ \log \beta^{-1/6}.
\end{multline}
We choose 
\be
p = 1+ c \sqrt{\eps}, \qquad c \in (0,\infty). 
\ee


\paragraph{2. Estimation of the second term in~\eqref{keysplit}.} 
Proceeding as in the proof of \eqref{dom}, we can again use \eqref{compmom} with $s=-2(p-1)^{-1}O(\eps) = -O(\sqrt{\eps}\,)$ to estimate, as in \eqref{eq:ey2esti}, 
\be
-(p-1)\log \widehat \E\Bigl[\e^{-\tfrac{1}{p-1}O(\eps)\,Y_2 } \mid N,(\Theta_i)_{i=1}^N\Bigr]
= (p-1)\tfrac{(N-1)}{2}\log(1-s)= O(\sqrt{\eps}\,N).
\ee
Taking expectations, we obtain 
\be
-(p-1)\log \widehat \E\Bigl[\e^{-\tfrac{1}{p-1}O(\eps)\,(\beta Y_1+Y_2+N) }\Bigr] 
= -(p-1)\log \widehat \E\Bigl[ \e^{-\tfrac{1}{p-1}O(\eps)\,(\beta Y_1+N) }\Bigr].
\ee
Applying Lemma~\ref{lem:ehatperturbed} with $\delta=\frac{1}{p-1}O(\eps) = O(\sqrt{\eps}\,)$, we conclude that
\be
\label{splitlb1}
\liminf_{\beta\to\infty} - \frac{1}{\beta^{1/3}}(p-1) \log \widehat 
\E\Bigl[\e^{-\tfrac{1}{p-1}O(\eps)\,(\beta Y_1+Y_2+N)}\Bigr] \geq O(\eps).
\ee


\paragraph{3. Estimate of the first term in~\eqref{keysplit}.}
Estimate
\begin{equation}
\label{splitlb2}
\begin{aligned}
&\widehat{\mathbb P} \bigl(Z^{(m)} \in \cO \cap \cV_{C\beta^{-2/3}}  \cap \cD_\eps(0)\bigr)
= \widehat{\mathbb P} \bigl(Z^{(m)} \in \cO\bigr)\\ 
&\quad - \left[\widehat{\mathbb P} \bigl(Z^{(m)} \in \cO,\,Z^{(m)} \notin \cD_\eps(0) \bigr)
+ \widehat{\mathbb P} \bigl(Z^{(m)} \in \cO \cap \cD_\eps(0),\,Z^{(m)} \notin \cV_{C\beta^{-2/3}}\bigr) 
\right].
\end{aligned}
\end{equation} 
We want to show that the last two probabilities are negligible. It suffices to show the following.

\begin{lemma} 
\label{lem:extreme}
For some $\eps_0>0$ small enough and uniformly on $|m| \leq \beta^{-1/6}$:
\begin{itemize}
\item[{\rm (a)}]
$\{Z^{(m)} \in \cO,\,\Theta_i \le \frac{2}{\sqrt{\Rc(\Rc-1)}} \sqrt{\eps}\,\,\,\forall\,1 \leq i \leq N,\, 
|m+B_{T_i}| \leq \tfrac12 \eps \sqrt{(\kappa-1)\beta}\,\,\,\forall\,1 \leq i \leq N\} 
\subset \{Z^{(m)} \in \cO \cap \cD_\eps(0)\}$ for every $0 < \eps \leq \eps_0$.
\item[{\rm (b)}] 
$\lim_{\beta\to\infty} \frac{1}{\beta^{1/3}} \log 
\widehat{\mathbb P} \bigl(Z^{(m)} \notin \cD_\eps(0)\bigr) = -\infty$ for every $0 < \eps \leq \eps_0$.
\item[{\rm (c)}] 
$\lim_{C\to\infty} \sup_{0<\eps \leq \eps_0} \limsup_{\beta\to\infty} \frac{1}{\beta^{1/3}} \log
\widehat{\mathbb P} \bigl(Z^{(m)} \in \cO \cap \cD_\eps(0),\,Z^{(m)} \notin \cV_{C\beta^{-2/3}}\bigr) = -\infty$.
\end{itemize}
\end{lemma}

\begin{proof}
(a) In polar coordinates, $Z_i^{(m)}=(r_i^{(m)} \cos T_i, r_i^{(m)} \sin T_i)$ with $\abs{r_i^{(m)}-(\Rc-1)}\le \frac12 \eps$. To prove that $Z^{(m)} \in \cO \cap \cD_\eps(0)$, it clearly suffices to show that, for any $i$, the boundary cusp $v_i=(x,y) \in \partial B(Z_i^{(m)})\cap  \partial B(Z_{i+1}^{(m)})$ belongs to $A_{\Rc,\eps}$. The most extremal case -- when the point $v_i$ is as close to the origin as possible -- occurs when 
$$
\abs{r_i^{(m)}}=\abs{r_{i+1}^{(m)}}=r=\Rc-1-\tfrac12\eps, \qquad \Theta_i= \frac{2}{\sqrt{\Rc(\Rc-1)}} \sqrt\eps.
$$ 
Assuming, without loss of generality, that $Z_i^{(m)}=(0,r)$ and $Z_{i+1}^{(m)}=(r\sin \Theta_i,r\cos \Theta_i)$, we find $v_i$ as the intersection of the line $\{(x,y)\colon\,x=y\tan(\tfrac12\Theta_i)\}$ (the axis of symmetry between the rays $\ell_i$ and $\ell_{i+1}$ passing through  the points $Z_i^{(m)}$ and $Z_{i+1}^{(m)}$) with the circle $\partial B((0,r))\colon\, x^2+(y-r)^2=1$. Using the shorthand 
$$
\eta=\tan^2(\tfrac12\Theta_i) = \frac{\eps}{\Rc(\Rc-1)} +O(\eps^2),
$$ 
we get $y=\frac{r + \sqrt{1- \eta (r^2-1)}}{1+\eta}$, which yields
\be
\abs{v_i}^2=x^2+y^2=(1+\tan^2(\tfrac12\Theta_i))y^2=(1+\eta)y^2=\frac{r^2+2r\sqrt{1-\eta(r^2-1)}+1-\eta (r^2-1)}{1+\eta}.
\ee
Hence, up to the order $O(\eps^2)$,
\be
\abs{v_i}^2=[r^2+2r(1-\tfrac12 \eta(r^2-1)+1-\eta (r^2-1)](1-\eta)=(r+1)^2(1-\eta r)=(\Rc-\tfrac{\eps}2)^2(1-\eta r),
\ee
which implies the claim for $\eps$ sufficiently small. Indeed,
\be
\label{E:vige}
\abs{v_i} \ge (\Rc-\tfrac{\eps}2) (1-\tfrac12\eta r)\ge (\Rc-\eps)
\ee
because
\be
\tfrac12\eta (\Rc-\tfrac{\eps}2)(\Rc-1-\tfrac{\eps}2)\le \frac12\frac{(\Rc-\tfrac{\eps}2)(\Rc-1-\tfrac{\eps}2)}{\Rc(\Rc-1)} 
\eps\le \tfrac{\eps}2.
\ee

\begin{figure}[htbp]
\centering
\begin{overpic}[width=8cm,scale=.25,percent]
{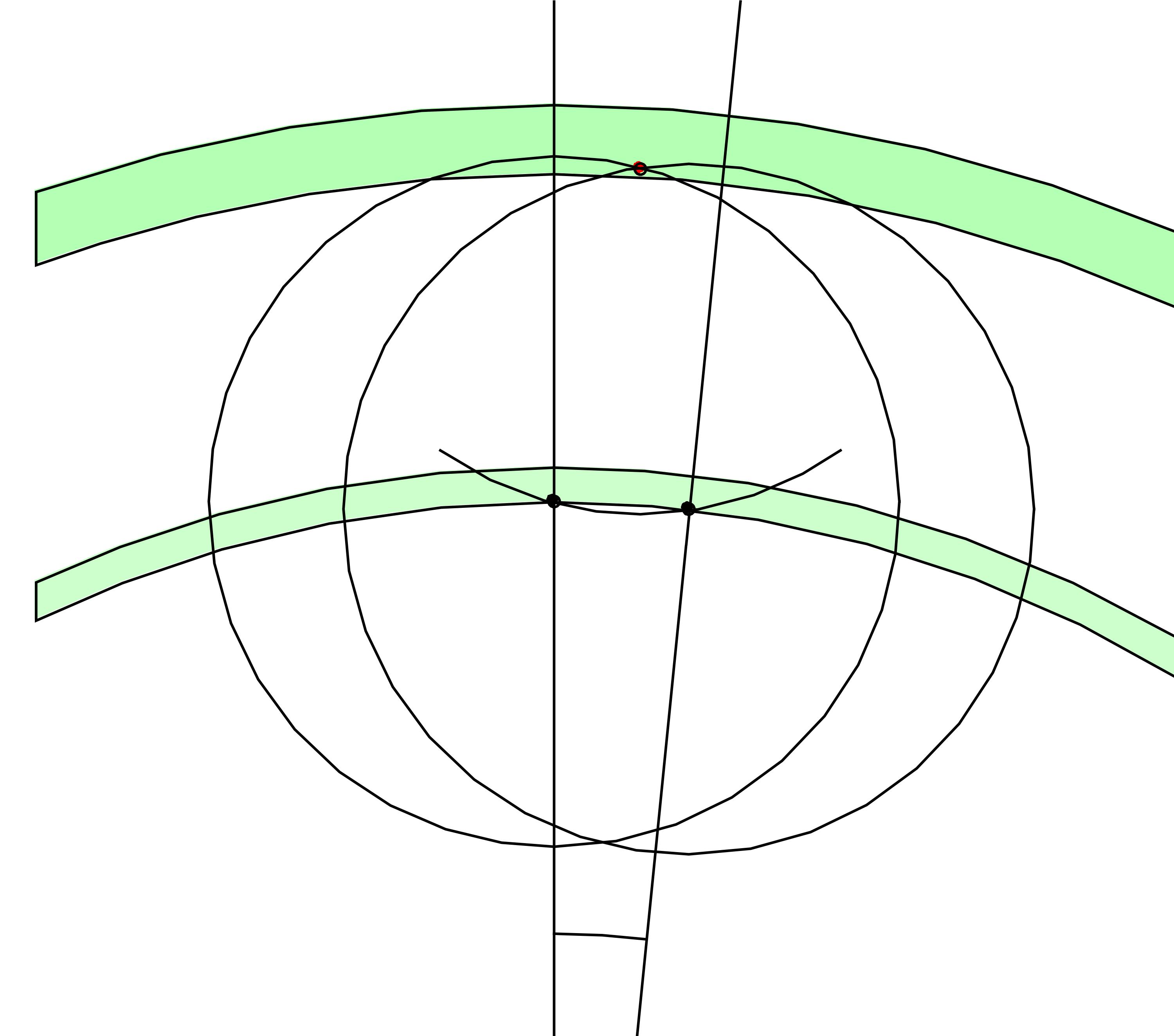}
\put(44,43) {$\scriptstyle z_{i}$}
\put(60.5,42) {$\scriptstyle z_{i+1}$}
\put(49.5,5){$\scriptstyle \Theta_i$}
\put(53.5,75) {$\scriptstyle v_{i}$}
\put(44,83) {$\scriptstyle \ell_{i}$}
\put(63.5,83) {$\scriptstyle \ell_{i+1}$}
\put(18,20) {$\scriptstyle \partial B(z_i)$}
\put(80,20) {$\scriptstyle \partial B(z_{i+1})$}
\put(34.5,51) {$\scriptstyle \partial B(v_i)$}
\end{overpic}
\caption{\small Illustration to the proof of the claim in Lemma~\ref{lem:extreme}(a).}
\label{fig:11}
\end{figure}

\medskip\noindent
(b) With the help of (a), we may estimate
\begin{equation}
\label{estsplit}
\begin{aligned}
\widehat{\mathbb P} \Bigl(Z^{(m)} \in \cO,\,Z^{(m)} \notin \cD_\eps(0)\Bigr) 
&\leq \widehat{\mathbb P} \Bigl(\exists\,1 \leq i \leq N\colon\,\Theta_i >\frac{2}{\sqrt{\Rc(\Rc-1)} \sqrt{\eps}} \Bigr)\\ 
&\qquad + \widehat{\mathbb P} \Bigl(\exists\,1 \leq i \leq N\colon\, |m+B_{T_i}| 
> \tfrac12 \eps \sqrt{(\kappa-1)\beta}\Bigr).
\end{aligned}
\end{equation}
Abbreviate $k = \frac{2}{\sqrt{\Rc(\Rc-1)}}$. Then the first term in the right-hand side of \eqref{estsplit} can be estimated as
\be
\begin{aligned}
\widehat{\mathbb P} \Bigl(\exists\,1 \leq i \leq N\colon\,\Theta_i > k \sqrt\eps \Bigr)
& \leq \frac{1}{\E[\exp( \widehat Y_0 - \widehat Y_1)]} \, \E\Biggl[ \sum_{i=1}^N 
\E\Bigl[ \e^{\widehat Y_0 - \widehat Y_1} 1_{\{\Theta_i >k \sqrt{\eps}\}}\,\Big|\, N\Bigr]\Biggr]. 
\end{aligned} 
\ee
In the conditional expectation we estimate 
\be
\widehat Y_1 = \beta C_1 \sum_{j=1}^n \Theta_j^3 \geq \tfrac12 \Theta_i^3 
+ \tfrac12 \widehat Y_1 \geq \tfrac12 k^3 \eps^{3/2} 
+ \tfrac12 \widehat Y_1 \quad \P\text{-a.s.}
\ee
With the help of the inequality $N\leq \frac{1}{\delta \e} \,\e^{N\delta}$ we deduce that, for every $\delta>0$,
\be
\begin{aligned}
\widehat{\mathbb P} \Bigl(\exists\,1 \leq i \leq N\colon\,\Theta_i > k \sqrt\eps \Bigr)
& \leq \e^{- k^{3} \eps^{3/2}\beta} \, \frac{\E[N\, \exp( \widehat Y_0 - \frac12 \widehat Y_1)]}
{\E[\exp( \widehat Y_0 - \widehat Y_1)]} \\
&\leq  \frac{1}{\delta \e}  \e^{- k^{3} \eps^{3/2}\beta} \, \frac{\E[\exp( \widehat Y_0 - \frac12 \widehat Y_1 + \delta N)]}
{\E[\exp( \widehat Y_0 - \widehat Y_1)]}.
\end{aligned} 
\ee
We already know from Proposition~\ref{prop:theta-asympt} that the denominator equals $\exp(-[1+o(1)]\,c\beta^{1/3})$ for some constant $c>0$ as $\beta\to \infty$. Arguments entirely analogous to those in the proof of Proposition~\ref{prop:theta-asympt} and Lemma~\ref{lem:ehatperturbed} show that the same holds for the numerator. It follows that
\be
\widehat{\mathbb P} \Bigl(\exists\,1 \leq i \leq N\colon\,\Theta_i > k \sqrt\eps \Bigr) 
= O \bigl( \e^{- c'\eps^{3/2}\beta}\bigr)
\ee
for some constant $c'>0$ as $\beta\to\infty$, and so we get the claim with a margin. 

As to the second term in the right-hand side of \eqref{estsplit}, since the tilting only affects the angular process and not the radial process, we have
\begin{equation}
\begin{aligned}
&\widehat{\mathbb P} \Bigl(\exists\,1 \leq i \leq N\colon\, |m+B_{T_i}| 
> \tfrac12\eps \sqrt{(\kappa-1)\beta}\Bigr)
= \sum_{n\in\N} \widehat{\mathbb E} \Bigl(1_{\{\exists\,1 \leq i \leq n\colon\, |m+B_{T_i}| 
> \tfrac12 \eps \sqrt{(\kappa-1)\beta},\,N=n\}}\Bigr)\\
&\qquad \leq \sum_{n\in\N} \sum_{i=1}^n \widehat{\mathbb E} 
\Bigl(1_{\{|m+B_{T_i}| > \tfrac12 \eps \sqrt{(\kappa-1)\beta},\,N=n\}}\Bigr)
= \sum_{n\in\N} \sum_{i=1}^n \widehat{\mathbb E} 
\Bigl(1_{\{|m+B_0| > \tfrac12 \eps \sqrt{(\kappa-1)\beta},\,N=n\}}\Bigr)\\[0.2cm]
&\qquad = \widehat{\mathbb E}[N]\, \widehat{\mathbb P} \bigl(|m+B_0| 
> \tfrac12 \eps \sqrt{(\kappa-1)\beta}\,\bigr)
= \widehat{\mathbb E}[N]\, \mathbb P \bigl(|m+B_0| > \tfrac12 \eps \sqrt{(\kappa-1)\beta}\,\bigr),
\end{aligned}
\end{equation}
where we use that the law of $(B_t)_{t \in [0,2\pi]}$ is invariant under shifts. Moreover, recalling \eqref{Bhatdef}--\eqref{Btildedef}, we have 
\begin{equation}
\begin{aligned}
&\mathbb P \bigl(|m+B_0| > \tfrac12 \eps \sqrt{(\kappa-1)\beta}\,\bigr) 
= \mathbb P \Biggl(\Big|\tfrac{1}{2\pi} \int_0^{2\pi} \d t\,\widetilde W_t\Big| 
> \tfrac12 \eps \sqrt{(\kappa-1)\beta}-|m|\Biggr)\\
&\leq \mathbb P \Biggl(\sup_{0\leq t\leq 2\pi} |W_t| 
> \tfrac23 \left[\tfrac12 \eps \sqrt{(\kappa-1)\beta}-|m|\right]\Biggr),
\end{aligned}
\end{equation}
where we use that $\frac{1}{2\pi}\int_0^{2\pi} \d t\,\widetilde W_t = \frac{1}{2\pi}\int_0^{2\pi} 
\d s \,W_s - \tfrac12 W_{2\pi}$ and
\be
\left\{\sup_{0\leq t\leq 2\pi} |W_t| \leq a\right\} \subset 
\left\{\left|\tfrac{1}{2\pi}\int_0^{2\pi} \d s \,W_s - \tfrac12 W_{2\pi} \right| \leq \tfrac32 a \right\},
\qquad a > 0.
\ee
But (see \cite[Lemma 5.2.1]{DZ})
\be
\mathbb P \Biggl(\sup_{0\leq t\leq 2\pi} |W_t| 
> \tfrac23 \left[\tfrac12 \eps \sqrt{(\kappa-1)\beta}-|m|\right]\Biggr)
\leq 4 \exp\left(-\tfrac12 \tfrac{1}{2\pi}\left[\tfrac23
\left[\tfrac12 \eps \sqrt{(\kappa-1)\beta}-|m|\right]\,\right]^2 \right),
\ee
and so we get the claim with a margin.

\medskip\noindent
(c) On the event $\{Z^{(m)} \in \cO \cap \cD_\eps(0)\}$ we can use the expansion of the volume in Proposition~\ref{prop:expansion} in the form given in~\eqref{volYk}, with $Y_4^{(m)}$ given in~\eqref{y4mdef}. We have 
\begin{equation} 
|S(Z^{(m)})| - \pi \Rc^2  = - C_1^\eps Y_1 + \frac{1+O(\eps)}{2 \beta}Y_2 
+ \frac12 \frac{1}{(\kappa -1) \beta} Y_3^{(m)} + \frac{\Rc}{\sqrt{(\kappa-1)\beta}} Y_4^{(m)}.
\end{equation}
Recall that $Y_1,Y_2,Y_3^{(m)}$ are non-negative, while $Y_4$ is not necessarily so. It follows that 
\begin{multline} 
\label{apron}
\widehat \P\bigl(Z^{(m)} \in \cO \cap \cD_\eps(0),\,Z^{(m)} \notin \cV_{C\beta^{-2/3}}\bigr)
\leq \widehat\P\Bigl( C_1^\eps Y_1 > \tfrac14 C\beta^{-2/3} \Bigr) 
+ \widehat\P\Bigl( \tfrac{1}{2\beta} [1+O(\eps)]\, Y_2 > \tfrac14 C\beta^{-2/3} \Bigr) \\
+ \widehat\P\Bigl( \tfrac12 \frac{1}{(\kappa -1) \beta}\, Y_3^{(m)} > \tfrac14 C\beta^{-2/3} \Bigr)
+ \widehat\P\Bigl( \frac{\Rc}{\sqrt{(\kappa-1)\beta}}\, |Y_4^{(m)}| > \tfrac14 C\beta^{-2/3} \Bigr).
\end{multline} 
The four probabilities on the right-hand side of~\eqref{apron} are estimated with the help of large deviations, Markov's inequality and the results from Section~\ref{sec:prep}. The first probability is bounded by
\be
\exp\bigl( - \tfrac14 s C \beta^{1/3}\bigr)\, \widehat \E\Bigl[ \e^{s C_1^\eps \beta Y_1} \Bigr],\quad s>0.
\ee
Using Lemma~\ref{lem:ehatperturbed}, we see that for $s \downarrow 0$,
\be
\limsup_{\beta \to \infty}\frac{1}{\beta^{1/3}}\log \widehat \E\Bigl[ \e^{s C_1^\eps \beta Y_1} \Bigr]< O(s)
\ee
and therefore, for some $\eps_1>0$,
\be
\lim_{C\to \infty} \sup_{0 <\eps \leq \eps_1} \limsup_{\beta \to \infty} \frac{1}{\beta^{1/3}} 
\log \widehat\P\Bigl( C_1^\eps Y_1 > \tfrac14 C\beta^{-2/3} \Bigr) = - \infty. 
\ee
The other three probabilities in~\eqref{apron} are treated in a similar way, and so it suffices to show that, for $s\in \R$ with $|s|$ small enough,
\be 
\label{err1}
\limsup_{\beta\to\infty} \frac{1}{\beta^{1/3}} 
\log \widehat{\mathbb E}\bigl[\e^{s Y_2} \bigr] < \infty, 
\,
\limsup_{\beta\to\infty} \frac{1}{\beta^{1/3}} 
\log \widehat{\mathbb E}\bigl[\e^{s Y_3^{(m)}} \bigr] <\infty,
\,
\limsup_{\beta\to\infty} \frac{1}{\beta^{1/3}} 
\log \widehat{\mathbb E}\bigl[\e^{s\beta^{1/2}Y_4^{(m)}} \bigr]<\infty.
\ee


\paragraph{First estimate in \eqref{err1}.}
Write Use \eqref{compmom} to compute
\be
\label{eq:compmom2}
\widehat{\E}\bigl(\e^{sY_2}\mid N,(\Theta_i)_{i=1}^N \bigr)
=\bigl(1-2s\bigr)^{-(N-1)/2}=\exp\Big(\tfrac{N-1}{2} \log(1-2s)^{-1}\Big)
\leq \e^{N[s + O(s^2)]} \quad \widehat{\mathbb P}\text{-a.s.}
\ee
Taking the expectation, we have
\be
\widehat{\mathbb E}\bigl(\e^{s Y_2} \bigr)
\leq \widehat{\mathbb E}\bigl(\e^{[s+O(s^2)]N} \bigr)
\ee
and therefore, using Lemma~\ref{lem:ehatperturbed}, we see that for $s \downarrow 0$,
\be
\label{Y1bd}
\limsup_{\beta\to\infty}\frac{1}{\beta^{1/3}} 
\log \widehat{\mathbb E}\bigl(\e^{s Y_2} \bigr)\leq O(s).
\ee


\paragraph{Second estimate in \eqref{err1}.} 
Note that $Y_3^{(m)} = \sum_{i=1}^N (m+\overline{B_{T_i}})^2 \Theta_i \leq \sum_{i=1}^N 2(m^2+\overline{B_{T_i}}^2)\Theta_i = 4\pi m^2 + 2Y_3^{(0)}$. The term $4\pi m^2 = O(\beta^{-1/3})$ is harmless. Write
\be
Y_3^{(0)} = E_1 +  \int_0^{2\pi} \d t\, B_{t}^2,
\qquad E_1 = \sum_{i=1}^N \overline{B_{T_i}}^2 \Theta_i -  \int_0^{2\pi} \d t\, B_{t}^2.
\ee
Estimate
\be
\widehat{\mathbb E}\bigl(\e^{sY_3^{(0)}} \bigr)
\leq \widehat{\mathbb E}\bigl(\e^{2s E_1}\bigr)^{\tfrac12} \,
\widehat{\mathbb E}\bigl(\e^{2s \int_0^{2\pi} \d t\, B_{t}^2}\bigr)^{\tfrac12}.
\ee
By Lemma~\ref{lem:stochdiscr1},
\be
\label{iotaest}
\widehat{\mathbb E}\bigl(\e^{2s E_1}\bigr) \leq 
\widehat{\mathbb E}\Bigl(\e^{|s|O\big(\sqrt{Y_1}\log(1/Y_1)\big)} \Bigr).
\ee
Since $Y_1\geq N(2\pi/N)^3$, we have $\log(1/Y_1)=O(\log N) = O(N)$. Since $Y_1 = O(\eps)$, this gives
\be
\widehat{\mathbb E}\bigl(\e^{2s E_1} \bigr)
\leq \widehat{\mathbb E}\bigl(\e^{|s| O(\sqrt{\eps} \log N)} \bigr) = e^{|s| O(\sqrt{\eps}\,)\beta^{1/3}},
\ee
where we use Lemma \ref{lem:ehatperturbed}. By Lemma~\ref{lem:varadhanprep},
\be 
\label{err3}
\widehat{\mathbb E}\bigl[\e^{2s \int_0^{2\pi} \d t\, B_t^2}\bigr] 
= \prod_{k\in\N} \Bigl( 1- \frac{4s}{k^2}\Bigr)^{-1},
\ee
which is finite for $s<\frac14$. 


\paragraph{Third estimate in \eqref{err1}.}
Note that $Y_4^{(m)} = \sum_{i=1}^N (m+\overline{B_{T_i}}) \Theta_i = 2\pi m + Y_4^{(0)}$. The term $2\pi m = O(\beta^{-1/6})$ is again harmless. We have
\be
\widehat{\mathbb E}\bigl(\e^{s \beta^{1/2}Y_4^{(0)}} \bigr)
= \widehat{\mathbb E}\left(\exp\left[s \beta^{1/2}\sum_{i=1}^N  \overline{B_{T_i}}\,\Theta_i\right]\right).
\ee
Use Lemma~\ref{lem:stochdiscr3} to bound this from above by
\be
\widehat{\mathbb E}\left(\exp\left[\tfrac{1}{32\pi} s^2\beta Y_1\right]\right).
\ee
Now use \eqref{Y1bd} to get
\be
\label{Y4bd}
\limsup_{\beta\to\infty}\frac{1}{\beta^{1/3}} \log
\widehat{\mathbb E}\bigl(\e^{s \beta^{1/2}Y_4^{(0)}} \bigr)
\leq O(s).
\ee
This completes the proof of~\eqref{err1}, hence of Lemma~\ref{lem:extreme}, and hence of the first term in~\eqref{keysplit}.
\end{proof}


\paragraph{4. Conclusion.} 
Combining~\eqref{keysplit}, \eqref{splitlb1}, \eqref{splitlb2} and Lemma~\ref{lem:extreme}, and choosing $C$ large enough, we get 
\be
\begin{aligned}
\label{lbkeyalt}
&\liminf_{\beta\to\infty} \frac{1}{\beta^{1/3}} \log \widehat \E\Bigl[\int_{|m|\leq  \beta^{-1/6}} \d m\,  
\e^{O(\eps) (\beta Y_1 + Y_2  + N)} 
\1_{\{Z^{(m)}\in \Upsilon^{\LB'}_{\beta,\eps}  \}}\Bigr] \\
&\qquad \geq \liminf_{\beta \to \infty}\inf_{|m|\leq \beta^{-1/6}}  \frac{1}{\beta^{1/3}}  
\log \widehat \P \bigl(Z^{(m)} \in \mathcal O\bigr) +  O(\eps).
\end{aligned}
\ee
This completes the proof of Proposition~\ref{prop:lowerkey}. 
\end{proof}


\subsection{Three technical conditions}
\label{conds}

Under the following conditions Conjecture \ref{thm:zoom2} is true.

\begin{enumerate} 
\item[(C1)] 
The limit 
\be \label{cond1a}
 \lim_{\beta\to \infty} \frac{1}{\beta^{1/3}} \log \widehat \P(Z^{(0)}\in \mathcal O) = - 2\pi G_\kappa p^{*} 
\ee
exists for some $p^{*} > 0$ that does not depend on $\kappa$.
\item[(C2)] 
The change from $Z^{(0)}$ to $Z^{(m)}$ does not affect (C1) when $m$ is not too large: 
\be
\lim_{\beta\to \infty} \frac{1}{\beta^{1/3}} \sup_{|m| = O(\beta^{1/6})}\Bigl| 
\log \frac{\widehat \P(Z^{(m)}\in \mathcal O)}{\widehat \P(Z^{(0)}\in \mathcal O)}\Bigr| = 0. 
\ee
\item[(C3)] 
For $\delta>0$ sufficiently small, 
\be
\label{chicond}
\limsup_{\beta\to \infty}\frac{1}{\beta^{1/3}} \log \frac{\widehat \E[\exp(\frac12 (1+\delta) \chi) 
\1_{\{ Z^{(m)}\in \mathcal O\}}]}{\widehat \P(Z^{(m)}\in \mathcal O)} =0, 
\ee
uniformly in $|m| = O(\beta^{1/6})$. 
\end{enumerate} 

\noindent
Condition (C1) comes from the fact that for each of the $N \asymp \beta^{1/3}$ boundary points there is a constraint in terms of the two neighbouring boundary points that must be satisfied in order for the corresponding unit disk to touch the boundary of the critical droplet. The constant $p^{*}$ is related to the free energy of an \emph{effective interface model}. Condition (C2) says that the constraint imposed by condition (C1) is not affected by \emph{small dilations} of the critical droplet, and implies that the free energy of the effective interface model is Lipschitz under small perturbations. Condition (C3) says that the first Fourier coefficient of the surface of the critical droplet is small. The term $\chi$ represents an energetic and entropic reward for the droplet boundary to fluctuate away from $\partial B_{\Rc}$. We require that this reward, which may be thought of as a background potential in the effective interface model, does not affect the microscopic free energy of the droplet.

\begin{remark}
{\rm Proposition~\ref{prop:upperkey}, when combined with conditions (C1) and (C3), yields 
\be
\limsup_{\beta \to \infty}\frac{1}{\beta^{1/3}} \log \mathcal I^\mathrm{UB}(\kappa,\beta; C,\eps) 
\leq  2\pi G_\kappa (\tau^{*} - p^{*}) + O(\eps)
\ee
with $p^{*}$ given in~\eqref{cond1a}. Combining this estimate with Corollary~\ref{cor:recapI}, we get
\be
\limsup_{\beta \to \infty}\frac{1}{\beta^{1/3}}  \log \Bigl( \e^{\beta\, I^*(\pi \Rc^2)}
\mu_\beta\Bigl( \cV_{C\beta^{-2/3}}\cap \cD_\eps(0) \cap 
\cC_{C\beta^{-2/3}}(0) \Bigr)\Bigr) \leq 2\pi G_\kappa (\tau^*-p^{*}) + O(\eps).
\ee
Combining this estimate with Lemma~\ref{lem:centering}, we find
\be
\limsup_{\beta \to \infty}\frac{1}{\beta^{1/3}}  \log \Bigl( \e^{\beta\, I^*(\pi \Rc^2)}	
\mu_\beta\Bigl( \cV_{C\beta^{-2/3}}\cap \cD_\eps \Bigr)\Bigr)
\leq 2\pi G_\kappa (\tau^*-p^{*}) + O(\eps).
\ee
Contributions from $V_{C\beta^{-2/3}}\cap \cD_\eps$ are bounded by Lemma~\ref{lem:ng}, and are negligible. The upper bound in Conjecture~\ref{thm:zoom2} follows after letting $\eps\downarrow 0$ with $C$ fixed. The lower bound in Conjecture~\ref{thm:zoom2} follows from Corollary~\ref{cor:recapI}, Proposition~\ref{prop:lowerkey} and conditions (C1) and (C2).} 
\end{remark} 


\subsection{Proof of the moderate deviation bounds}
\label{sec:moddevbds}

In this section we give the proof of Theorem~\ref{thm:zoom1}.


\paragraph{Upper bound.}

\begin{proof}
We estimate the right-hand side of~\eqref{ubkey} by dropping the indicator. We decompose $\chi$ into three parts
\be
\chi =E_1 +  E_2 + E_3, 
\ee
with (recall \eqref{foucoeff})
\begin{equation}
\label{newdefE}
\begin{aligned}
E_1 &=  \sum_{i=1}^N \overline{B_{T_i}}^2 \Theta_i -  \int_0^{2\pi} \d t\, B_{t}^2,\\
E_2 &= \int_0^{2\pi} \d t\, B_{t}^2- A_1 ^2 - A_1^{*2},\\
E_3 &= A_1 ^2 + A_1^{*2} - D_1^2- {D_1^*}^{2}.
\end{aligned}
\end{equation} 
We first look at conditional expectations. Note that $Y_1$ depends on $N$ and $\Theta_i$, $1 \leq i \leq n$, alone. A repeated application of the Cauchy-Schwarz inequality yields 
\begin{equation} 
\label{split2}
\begin{aligned}
&\widehat \E\Bigl[ \e^{\tfrac12 [1+O(\eps)]\chi} \mid N,(\Theta_i)_{i=1}^N\Bigr]\\
&\leq \widehat \E\Bigl[ \e^{\tfrac 3 2 [1+O(\eps)] E_1} \mid N,(\Theta_i)_{i=1}^N\Bigr]^{1/3}
\, \widehat\E\Bigl[ \e^{\tfrac 3 2  [1+O(\eps)]E_2} \mid N,(\Theta_i)_{i=1}^N\Bigr] ^{1/3}
\,\widehat\E\Bigl[\e^{\tfrac 3 2 [1+O(\eps)]E_3} \mid N,(\Theta_i)_{i=1}^N\Bigr] ^{1/3}.
\end{aligned}
\end{equation}
By Lemma~\ref{lem:stochdiscr1}, for every $s\in (-1,1)$ with $|s|<1/\sqrt{2\pi Y_1}$, 
\be\label{eq:e5esti}
\log \widehat \E\Bigl[\e^{\frac{1}{2} s\,E_1}\,\Bigr|\, N,(\Theta_i)_{i=1}^N \Bigr]
\leq |s| \sqrt{Y_1}\log(1/Y_1) \quad \widehat{\mathbb P}\text{-a.s.}
\ee
By Lemma~\ref{lem:stochdiscr2}, for every $|s|\leq 1/\sqrt{2\pi Y_1}$,
\be\label{e3}
\log \widehat \E\Bigl[\e^{\frac{1}{2} s (A_1^2 - D_1^2)} \,\Bigr|\, N,(\Theta_i)_{i=1}^N \Bigr]
\leq |s| \sqrt{Y_1}\log(1/Y_1) \quad \widehat{\mathbb P}\text{-a.s.}
\ee
A similar estimate holds for $A_1^{*2} - D_1^{*2}$. Via Cauchy-Schwarz, it follows that
\be\label{eq:e6esti}
\log \widehat \E\Bigl[\e^{\frac{1}{2} s E_3} \,\Bigr|\, N,(\Theta_i)_{i=1}^N \Bigr] 
\leq 2 |s| \sqrt{Y_1}\log(1/Y_1) \quad \widehat{\mathbb P}\text{-a.s.}
\ee
as long as $2|s|\leq 1/\sqrt{2\pi Y_1}$. We would like to apply the estimates in \eqref{eq:e5esti} and~\eqref{eq:e6esti} with $s= 3[1+O(\eps)]$. From the a priori estimates in Corollary~\ref{cor:apriori} we know that $Y_1= O(\eps)$. Hence $1/\sqrt{2\pi Y_1} \geq c'/\sqrt{\eps}$ for some $c'>0$. Thus, $|s|\leq c'/\sqrt{4\eps}$ is sufficient to ensure the condition $2|s|<1/\sqrt{2\pi Y_1}$. Therefore we see that $s= 3 [1+O(\eps)]$ satisfies the bound $2|s|\leq 1/\sqrt{2\pi Y_1}$, so that~\eqref{eq:e5esti} and~\eqref{eq:e6esti} are valid. 

In order to get rid of $Y_1$ in the right-hand side of \eqref{eq:e6esti}, we use two estimates: the a priori estimate $Y_1 = O(\eps)$ and the bound $Y_1\geq N (2\pi/N)^3$, which gives $\log(1/Y_1) = O(\log N) = O(N)$. Therefore 
\be\label{eq:e6esti2}
\log \widehat \E\Bigl[\e^{\frac{1}{2} s E_3} \,\Bigr|\, N,(\Theta_i)_{i=1}^N \Bigr] 
\leq O(\eps \log N) = O(\eps N) \quad \widehat{\mathbb P}\text{-a.s.}
\ee
A similar bound holds for $E_1$. We now estimate the term with $E_2$. The tilt by $\e^{Y_0 - \beta C_1 Y_1}$ affects the angular point process only, so it is still true under $\widehat{\mathbb {P}}$ that the distribution of $(B_t)_{t\in [0,2\pi]}$ is a mean-centred Brownian bridge independent from $N$ and the $\Theta_i$, $1 \leq i \leq N$. Since $s= 3[1+O(\eps)]<4$, by picking $\eps$ small enough, we get from Lemma~\ref{lem:varadhanprep} that 
\be\label{eq:ey22esti}
 \log \widehat \E\Bigl[ \e^{\tfrac 3 2 [1+O(\eps)]E_2}\Bigr]
= -\sum_{k=2}^\infty \log \Bigl(1- \frac{3}{k^2} [1+O(\eps)]\Bigr) = O(1).
\ee
Note that this upper bound does not grow with $\beta$. Hence, combining \eqref{eq:e5esti}, \eqref{eq:ey22esti} and \eqref{eq:e6esti2}, inserting into \eqref{split2} and taking expectations, we find
\be
\log \widehat \E\Bigl[ \e^{\tfrac12 [1+O(\eps)]\chi}\Bigr]\leq 
\log  \widehat \E\Bigl[\e^{O(\eps) N} \Bigr] \quad \widehat{\mathbb P}\text{-a.s.}
\ee
Via Lemma \ref{lem:ehatperturbed} and the observation that $Y_1 \geq 0$, we obtain that
\be
\limsup_{\beta\to\infty} \frac{1}{\beta^{1/3}} \log \widehat \E\Bigl[\e^{\tfrac12 [1+O(\eps)]\chi}\Bigr] \leq O(\eps).
\ee
The upper bound in Theorem~\ref{thm:zoom1} follows after dropping $\1_{\{Z^{(m)} \in \mathcal O\}}$ and the supremum over $|m| \leq c \beta^{1/6}$, and letting $\eps\downarrow 0$ with $C$ fixed.  
\end{proof}


\paragraph{Lower bound.}

\begin{proof}
We work directly with the surface integrals and skip the auxiliary random variables. By Corollary~\ref{cor:recapI}, for any $n\in\N$, 
\be\label{eq:lowereasy1}
\begin{aligned}
&\e^{\beta I^*(\pi \Rc^2)} \mu_\beta \Bigl(\cV_{C\beta^{-2/3}}\Bigr) \geq \e^{\beta I^*(\pi \Rc^2)}
\mu_\beta\Bigl( \cV_{C\beta^{-2/3}}\cap \cD_\eps (0)\Bigr)\\
&\quad\geq \Big[1-O(\e^{-\kappa\beta})\Bigr]\, (\kappa \beta)^n \int_{[0,2\pi]^n} 
\d \t\, \1_{\{0 \leq t_1<\cdots < t_n < 2\pi\}} \int_{\R^n} \d \r\, \e^{-\beta \cH(\z)} \,
\1_{\cV_{C\beta^{-2/3}}\cap \cD_\eps (0)}(\z),
\end{aligned}
\ee
where 
\be
\cH(\z) = \bigl(|S(\z)| - \kappa |S(\z)^-|\bigr) - \pi\bigl(\Rc^2 - \kappa (\Rc-1)^2\bigr)
\ee
with $\z = (z_1,\ldots,z_n)$ and $z_i = (r_i \cos t_i, r_i \sin t_i)$. We write $r_i = \Rc-1+\rho_i$, and restrict the integral to the domain 
\be
M_{n,\eps} = \Bigl\{ (\t,\rr)\colon\, \max_{1 \leq i \leq n} \bigl|t_i - \tfrac{2\pi}{n}(i-\tfrac12)\bigr| \leq \tfrac13 \tfrac{2\pi}{n},\, 
\max_{1 \leq i \leq n} |\rho_i| \leq \eps\Bigr\}.
\ee 

\begin{lemma} 
\label{lem:zinM}
Pick $n = \lfloor A \beta^{1/3}\rfloor$ and $\eps = B \beta^{-2/3}$ with $B \leq \tfrac{\Rc-1}{A^2}$ (where $A,B \in (0,\infty)$ are constants to be optimised over later). Then $M_{n,\eps} \subseteq \cO \cap \cD_\eps(0)$ for $\beta$ sufficiently large.
\end{lemma}

\begin{proof}
To prove that $\z \in \cO $, i.e., $\z$ is a connected outer contour, employing Proposition~\ref{prop:locality} we need to show that every triplet $(z_{i-1},z_i,z_{i+1})$ is extremal. Actually, we will show that the intersection of the ray $\ell_i$ (halfline starting at the origin and passing through $z_i$) is intersecting the circle $\partial B(z_i)$ in a point $p \in A_{\Rc,\eps}\setminus B(z_{i-1})\cup B(z_{i+1})$. Clearly, it suffices to show that $p\notin B(z_{i+1})$.

Consider the most extremal case (as illustrated in Fig.~\ref{fig:12}): $\rho_i$ attains the minimum allowed value $\rho_i = -B \beta^{-2/3}$, $\rho_{i+1}$ attains the maximum allowed value $\rho_{i+1} = B \beta^{-2/3}$, the angle $\alpha$ between the rays $\ell_i$ and $\ell_{i+1}$ is minimal, i.e., $\alpha= \frac{2\pi}{n}(1-\frac23)$, and the $y$-coordinate of $z_i$ is at least as large as the $y$-coordinate $z_{i+1}$, i.e.,
\be
\label{E:subvert}
(\Rc-1-B \beta^{-2/3}) \ge (\Rc-1+B \beta^{-2/3}) \cos \alpha
\ee
(in Fig.~\ref{fig:12} these $y$-coordinates are equal and the halfline $\ell$ is parallel to $\ell_i$). Without loss of generality, we may assume that $z_i=(0,\Rc-1-B \beta^{-2/3})$ and $z_{i+1}=(r \sin\alpha, r \cos\alpha)$, where $r=\Rc-1+B \beta^{-2/3}$. Let $\partial B(z_i)\cap \partial B(z_{i+1})=\{v, v_i\}$ with $v\in B_{\Rc-1-\eps}(0)$ and let $\ell$ be the halfline starting at $v$ and passing through $v_i$. Because $\alpha\ge \frac{2\pi}{3 n}\ge \frac{2\pi}{3 A\beta^{1/3}}$, a sufficient condition for the inequality in \eqref{E:subvert} is
\be
(\Rc-1)\frac{1-\cos\alpha}{1+\cos\alpha}\ge(\Rc-1)\frac{\alpha^2}{4}
\ge (\Rc-1)\frac14 \left(\frac{2\pi}{3 A\beta^{1/3}}\right)^2\ge B \beta^{-2/3},
\ee
which yields the sufficient condition $B \le \frac{\Rc-1}{A^2}$, as claimed.

\begin{figure}[htbp]
\centering
\begin{overpic}[width=8cm,scale=.25,percent]
{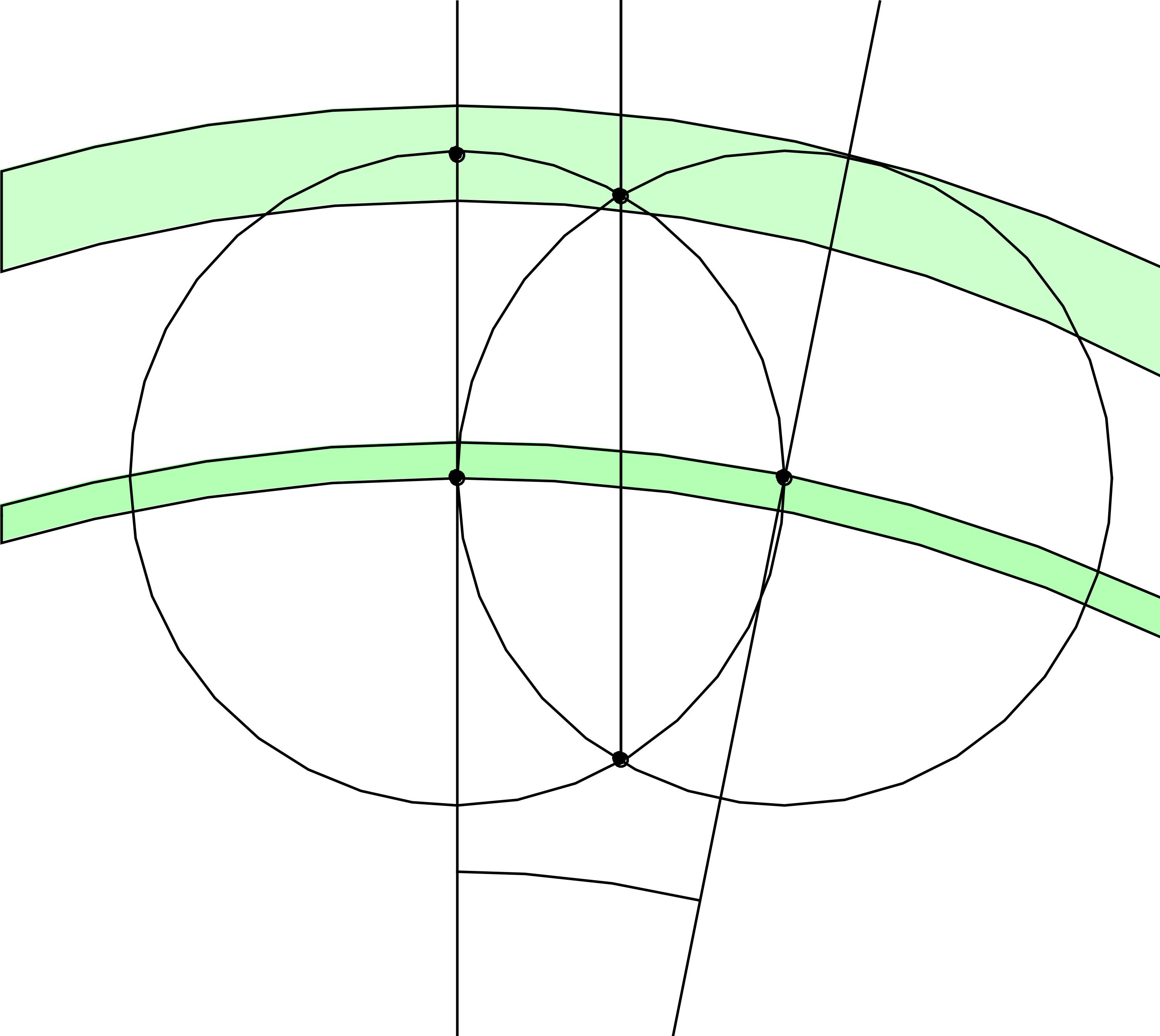}
\put(101,60) {$\scriptstyle A_{\Rc,\eps}$}
\put(101,34.5) {$\scriptstyle A_{\Rc-1,B \beta^{-2/3}}$}
\put(36,45) {$\scriptstyle z_{i}$}
\put(68.5,49) {$\scriptstyle z_{i+1}$}
\put(48.5,7){$\scriptstyle \alpha$}
\put(54,75) {$\scriptstyle v_{i}$}
\put(36,74) {$\scriptstyle p$}
\put(52.5,21) {$\scriptstyle v$}
\put(36,85) {$\scriptstyle \ell_{i}$}
\put(51,85) {$\scriptstyle \ell$}
\put(68,85) {$\scriptstyle \ell_{i+1}$}
\put(8,26) {$\scriptstyle \partial B(z_i)$}
\put(88,26) {$\scriptstyle \partial B(z_{i+1})$}
\end{overpic}
\caption{\small Illustration of the proof of Lemma~\ref{lem:zinM}.}
\label{fig:12}
\end{figure}

To prove that  $\z \in \cD_\eps(0)$, we refer to the claim in Lemma~\ref{lem:extreme}(a), where we show that if $\theta_i \le \frac{2}{\sqrt{\Rc(\Rc-1)}} \sqrt{\eps}$, then the boundary cusp $v_i$  belongs to $A_{\Rc,\eps}$. The above bound on the angular increment $\theta_i $ is clearly satisfied once $\beta$ is sufficiently large, because $\theta_i \leq \tfrac{5}{3} \tfrac{2\pi}{n} \sim \beta^{-1/3}$, as will be discussed in the next paragraph.
\end{proof}

Returning to the proof of the lower bound in Theorem~\ref{thm:zoom1}, we check that the volume constraint is satisfied as well (recal the orders of magnitude of the relevant quantities from Definition~\ref{def:yzsums}). Note that for $(t,\rho)\in M_{n,\eps}$ the angular increments are bounded as 
\be\label{eq:lowereasy2}
\tfrac{1}{3} \tfrac{2\pi}{n} \leq \theta_i \leq \tfrac{5}{3} \tfrac{2\pi}{n} 
\ee
and $\frac{2\pi}{n} = [1+o(1)]\,\frac{2\pi}{A \beta^{1/3}}$, as $\beta \to \infty$. For $\z \in M_{n,\eps}$,
\be\label{eq:lowereasy3}
y_1 (\z) \leq 2\pi \overline{K}^2\eps, \quad
y_2(\z) \leq  2\pi K^2\eps, \quad
y_3 (\z) \leq  2\pi \eps^2, \quad
\abs{y_4(\z)} \leq 2\pi \eps,
\ee
with $K = 2\sqrt{\Rc(\Rc-1)}$ and $\overline{K} = 4/\sqrt{\Rc(\Rc-1)}$. By Proposition~\ref{prop:expansion},
\be
\begin{aligned}
\cH(\z) &= C_1^\eps y_1(\z) + C_2^\eps y_2(\z) - C_2 y_3(\z),\\
S(\z)-\pi\Rc^2 &= - C_1^\eps y_1(\z) + C_2^\eps y_2(\z)  + C_3 y_3(\z) + C_4 y_4(\z).
\end{aligned}
\ee
with (recall \eqref{C123def})
\be 
C_1 = \frac1{24}\Rc^2(\Rc-1), \quad  C_2 = \frac{1}{2(\Rc-1)}, \quad C_3 = \frac12, \quad C_4=\Rc,
\ee 
where $C_k^\eps = C_k\, [1+O(\eps)]$, $k=1,2$. Hence
\be
\cH(\z) \leq [1+O(\beta^{-2/3})]\,D \beta^{-2/3}
\ee
with $D = 2\pi B (C_1\overline{K}^2 + C_2 K^2) = 2\pi B(\tfrac23 \Rc + 2\Rc) = 2\pi B\,\tfrac83\Rc$, and 
\be
\abs{S(\z)-\pi\Rc^2} \leq [1+O(\beta^{-2/3})]\,E\beta^{-2/3}
\ee
with $E = 2\pi B (C_1\overline{K}^2+C_2K^2+C_4) = 2\pi B\,\tfrac{11}{3}\Rc$. Consequently, by \eqref{eq:lowereasy1},
\be
\begin{aligned}
\e^{\beta I^*(\pi \Rc^2)} \mu_\beta \Bigl(\cV_{C\beta^{-2/3}}\Bigr)
& \geq \Big[1-O(\e^{-\kappa\beta})\Bigr]\,(\kappa \beta)^n\,\bigl(\tfrac23 \tfrac{2\pi}{n}\bigr)^n \bigl(B\beta^{-2/3}\bigr)^n\, 
\e^{- [1+O(\beta^{-2/3})]\,D \beta^{1/3}}\\
& = \bigl[1+o(1)\bigr]\, \bigl( \kappa\, 4\pi A^{-1}B\bigr)^n\,\e^{- D\beta^{1/3}},
\end{aligned}
\ee
provided $E \leq C$ to match the volume constraint. Therefore, recalling from \eqref{USkapdef} that $\Rc = \frac{\kappa}{\kappa-1}$, we get
\be
\liminf_{\beta\to \infty}\frac{1}{\beta^{1/3}} \log \Bigl( \e^{\beta I^*(\pi \Rc^2)} 
\mu_\beta \Bigl(\cV_{C\beta^{-2/3}}\Bigr) \Bigr) \geq L_\kappa(A,B)
\ee
with 
\be
L_\kappa(A,B) = - 2\pi B\,\tfrac83\tfrac{\kappa}{\kappa-1} + A \log(\kappa\, 4\pi A^{-1} B).
\ee 
This function must be maximised over $A,B \in (0,\infty)$ for fixed $\kappa \in (1,\infty)$ subject to the constraints $2A^2B \leq \frac{1}{\kappa-1}$ and $2\pi B \tfrac{11}{3} \frac{\kappa}{\kappa-1} \leq C$. For fixed $B$ the supremum over $A$ is taken at $A$ solving $0 = \frac{\partial}{\partial A} L_\kappa(A,B)$, which gives $\log(\kappa\, 4\pi A^{-1} B) = 1$, and so $\max_{A} L_\kappa(A,B) = L_\kappa(B)$ with
\be
L_\kappa(B) =  2\pi B\,(-\tfrac83\tfrac{\kappa}{\kappa-1} + \tfrac{2}{\e}\,\kappa),
\ee    
provided $B^3 \leq \frac{\e^2}{32\pi^2}\frac{1}{\kappa^2(\kappa-1)}$ to match the first constraint. For $\kappa \leq 1+\tfrac{4e}{3}$, the last factor is $\leq 0$ and the supremum over $B$ is taken at $B=0$, which gives $\max_B L_\kappa(B) = 0$. For $\kappa > 1+\tfrac{4e}{3}$ the supremum is taken at some $B>0$ depending on $C,\kappa$, which gives $\max_B L_\kappa(B) > 0$. This completes the proof of Theorem~\ref{thm:zoom1}. 
\end{proof}


\appendix


\section{Proof of two key lemmas}
\label{appA}

In this appendix we give the proof of Lemmas~\ref{lem:admissible} and \ref{L:isope}.  

\begin{proof}[Proof of Lemma~\ref{lem:admissible}.] 
We indicate the proper references to the literature. Part (\ref{Sreach}) is delicate.

\medskip\noindent
(\ref{F-+}) (Matheron~\cite[Chapter 1.5]{Ma}) Note that $x\in (F^-)^+$ is equivalent to $B(x)\cap F^-\neq \emptyset$, which is equivalent to the existence of a $z\in B(x)\cap F^-$. Hence, $x\in B(z)$ since $z\in B(x)$ and $B(z)\subset F$ since $z\in F^-$. In \cite{Ma}, sets $F$ such that $(F^-)^+=F$ are called open w.r.t.\ $B(0)$ (rather than admissible). 

\medskip\noindent
(\ref{adm-comp}) (Matheron~\cite[Chapter 1.5]{Ma}). For any $S\in\mathcal S$, we have $S=F^+$ with $F=S^-$. On the other hand, if $S=F^+$, then $((F^+)^-)^+\supset F^+$ since $(F^+)^-\supset F$ and hence $(S^-)^+\supset S$. The inclusion $((F^+)^-)^+\subset F^+$, which amounts to $(S^-)^+\subset S$, was proven in 1. 

\medskip\noindent
(\ref{contFtoF+}) For the first claim, see Matheron~\cite[Proposition 1.5.1]{Ma} and Schneider and Weil~\cite[Theorem 12.3.5]{SW}, for the second claim, see Kampf~\cite[Lemma 9]{Ka}.

\medskip\noindent
(\ref{connS-toS}) See Matheron~\cite[Chapter 1.5]{Ma}.

\medskip\noindent
(\ref{convS-toS}) See Matheron~\cite[Proposition 1.5.3]{Ma}. In \cite{Ma} , sets $F$ such that $(F^+)^-=F$ are called closed w.r.t.\ $B(0)$.

\medskip\noindent
(\ref{FtoS}) This follows from the fact that finite sets are dense in $\mathcal{F}$ in combination with the first claim in (\ref{contFtoF+}).

\medskip\noindent 
(\ref{Sreach}) Use Federer~\cite[Theorem 5.9]{Fe1}, which assures that the $\reach$ is conserved under taking limits of sets with respect to the Hausdorff metric. It therefore suffices to consider $S\in \mathcal S^{\textrm{fin}}$ with $S=h(\gamma)$, where the finite set $\gamma$ is sufficiently dense so that condition (C) is satisfied for $S=h(\gamma)$.

We will prove that, for such $\gamma$, $\reach(h(\gamma)^-)\geq 1$. First, observe that the boundaries $\partial h(\gamma)$ and $\partial h(\gamma)^-$ are unions of circular arcs (of radius 1). Given that $h(\gamma)$ satisfies condition (C), the set $h(\gamma)\setminus h(\gamma)^-$ splits into connected components, each bordered by two Jordan curves: one connected component of the boundary $\partial h(\gamma)$ and one connected component of the boundary $\partial h(\gamma)^-$. For each component of $\partial h(\gamma)$, we label the centres of the arc circles in such a way that two consecutive arcs belong to two consecutive centres, with periodic boundary condition. The associated component of $\partial h(\gamma)^-$  is a union of circular arcs passing through the centres. The centre of the arc circle connecting two consecutive centres is the point on the boundary $\partial h(\gamma)$ in the intersection of the arcs with these centres. See Fig.~\ref{fig:cleansausage} for an illustration.

Let us assume that $\reach(h(\gamma)^-)< 1$. Then there exist a point $x\in h(\gamma) \setminus h(\gamma)^-$ and two distinct points $y_1,y_2$ in a connected component $\sigma$ of the boundary $ \partial h(\gamma)^-$ such that $\dist(x,h(\gamma)^-) = \dist(x,y_1)= \dist(x,y_2)=r<1$. (To belong to two distinct components of $ \partial h(\gamma)^-$ contradicts assumption (C).) Given that $\dist(x,h(\gamma)^-)=r$, the interior $B_r(x)^0$ of the disk $B_r(x)$ does not contain any point from $h(\gamma)^-$, i.e., $B_r(x)^0 \cap h(\gamma)^-=\emptyset$. In addition, there are at most finitely many points of $h(\gamma)^-$ in $\partial B_r(x)$, all of which belong to $\partial h(\gamma)^-$. 

The admissibility of $h(\gamma)$ means that every unit disk  with a centre on $\partial h(\gamma)^-$ is fully contained in $h(\gamma)$. We will draw a contradiction with this statement from the fact that a Jordan curve $\sigma$ avoiding $B_r(x)^0$ and containing two distinct points on its boundary necessarily indents too sharply to be consistent with admissibility of $h(\gamma)$ and condition (C). To show the contradiction, consider a line $\ell$ through the point $x$ that separates the points $y_1$ and $y_2$ into opposite half planes determined by $\ell$. Without loss of generality, we may assume that the points $\{y,y'\}= \ell\cap \partial B_r(x)$ do not belong to $h(\gamma)^-$. As a result, both $B(y)$ and  $B(y')$ contain points not in $h(\gamma)$ or, since $B(y_1)\cup B(y_2) \subset h(\gamma)$, there exist points $w\in B(y)\setminus (B(y_1)\cup B(y_2))$ and $w'\in B(y')\setminus (B(y_1)\cup B(y_2))$ such that $w,w'\not\in h(\gamma)$. Now, $B(w)$ and $B(w')$ cannot contain any point from $h(\gamma)^-$, and hence the Jordan curve $\sigma$ must avoid $B_r(x)^0 \cup B(w)\cup B(w') \cup \{y,y'\}$ with $w,w'\not\in h(\gamma)$, which yields the contradiction with condition (C).

Indeed, if $\sigma$ is the outer boundary of the set $h(\gamma)^-$ with a single outer component of $\T\setminus h(\gamma)^-$, then in contradiction with condition (C) this component contains two components of $\T\setminus h(\gamma)$, since the points $w$ and $w'$ belong to different components of $\T\setminus (h(\gamma)^-\cup   B(y_1) \cup  B(y_2))$. Otherwise, the Jordan curve $\sigma$ is the inner boundary of the set $h(\gamma)^-$ and surrounds the set $B_r(x)^0\cup B(w)\cup B(w') \cup \{y,y'\}$. However, the region encircled by $\sigma$ contains two different components of $\T\setminus h(\gamma)$, one containing $w$  and the other containing $w'$.

\medskip\noindent
(\ref{noholes}) 
Without loss of generality, assume that $x=0$.
According to the definition \eqref{E:dH} of the Hausdorff metric, the condition \eqref{E:epsBonnesen} implies that $B_{R-\eps}(0)\subset S\subset B_{R+\eps}(0)$
and thus 
\be 
\label{E:S-inout}
B_{R-1-\eps}(0)\subset S^- = \{y \in S: B(y) \subset S\} \subset B_{R-1+\eps}(0).
\ee
We will prove connectedness and simple connectedness of the set $S^-$ by showing that every segment $\ell_e=\{te\colon\,t\in[0,R+\eps\}$ in the direction of any unit vector $e\in\R^2$, intersects the set $S^-$ in a closed interval:  $S^-\cap\ell_e=\{te\colon\,t\in[0,T(e)]\}$ with $T(e)\in[R-1-\eps,R-1+\eps]$. If the contrary were true, then there would be a direction $e$ and two points $y_1,y_2 \in \{te\colon\,t\in[R-1-\eps,R-1+\eps]\}\subset \ell_e$, $\abs{y_1}< \abs{y_2}$ such that $y_1\not\in S^-$ and $y_2\in S^-$. The fact that $y_2\in S^-$ implies that $B_{R-\eps}(0)\cup B(y_2) \subset S$, while $y_1\not\in S^-$ implies that $B(y_1)\cap S^{\rm{c}}\neq\emptyset$, and hence, in view of the preceding inclusion, $B(y_1)\cap (B_{R-\eps}\RK{(0)}\cup B(y_2))^{\rm{c}}\neq\emptyset$. However, this cannot happen when $B(y_1)\setminus B(y_2)\subset B_{R-\eps}(0)$. Nevertheless, this is exactly what happens when $R-1\ge \eps/(1-\eps)$. Indeed, $B(y_1)\setminus B(y_2)\subset B_{R-\eps}(0)$ once $\partial B(y_2)\cap \ell \subset B_{R-\eps}(0)$, where $\ell$ is the line through $y_2$ orthogonal to $\ell_e$. For $z\in \partial B(y_2)\cap \ell$ we have 
\be
z^2=y_2^2+1\le (R-1+\eps)^2 +1 \le (R-\eps)^2
\ee
when $\frac{R-1}{2R-1}\ge \eps$, which is equivalent with $R-1\ge \frac{\eps}{1-2\eps}$.

\medskip\noindent
(\ref{SLip}) For the first claim, see Ambrosio, Colesanti and Villa~\cite[Proposition 3]{ACV}. For the second claim, it suffices to note that for $S\in \mathcal S^{\textrm{fin}}$ the boundary $\partial S^-$ is a finite union of arcs.
\end{proof} 

\begin{proof}[Proof of Lemma~\ref{L:isope}.]
We use the claim about the stability of the Brunn-Minkowski inequality, first proven in qualitatively by Christ~\cite{Ch} and then quantitatively by Figalli and Jerrison~\cite{FJ}. Actually, for our purposes the qualitative version \cite{Ch} suffices, since we are only using it as a springboard for more accurate bounds to be considered later.

Adapted to our setting, the claim is that there exist positive function $\widetilde \xi(\delta)$ with $\lim_{\delta\to 0}\widetilde \xi(\delta)=0$, such that if 
\be
\label{E:BMstab}
\abs{S}^{1/2}\le \abs{S^-}^{1/2}+\sqrt\pi +\delta \max\big(\abs{S^-}^{1/2},\sqrt\pi\,\big),
\ee
with sufficiently small $\delta$, then there exist a compact convex set $K\subset \T$ such that 
\be
\label{E:K-S}
S^-\subset K^-, \qquad \abs{K^-\setminus S^-}< \pi R^2 \tilde\xi(\delta).
\ee
In the same vein, a (properly scaled and shifted) disk $D$ satisfies 
\be
\label{E:K-D}
D^-\subset K^-, \qquad \abs{K^-\setminus D^-}< \pi R^2 \tilde\xi(\delta).
\ee
Suppose, without loss of generality that the centre of the disk $D^-$ is at the origin and write $r$ for its radius, so that $D^-=B_{r}(0)$. 
Note that, when rescaling $B_1(0)$ corresponding to the original formulation in \cite{Ch} to $D^-$, we might need to readjust the constant $ \tilde\xi(\delta)$ taking into account that the areas of all sets considered above belong to a fixed compact interval $[\eps_0,L^2/4]$. To verify the assumption in \eqref{E:BMstab}, we rewrite \eqref{E:Bonnassumpt} as
\be
\label{E:defect}
\big(|S|-\kappa|S^-|\big) - \big(\pi R^2 - \kappa \pi (R-1)^2\big) = \kappa\bigl(\abs{S\setminus S^-} -2\pi R+\pi)\bigr)
\le  \pi\kappa\eps
\ee
and use an equivalent formulation of \eqref{E:BMstab}, namely,
\be
\label{E:eqBMstab}
\abs{S\setminus S^-}- 2\sqrt\pi\,\abs{S}^{1/2}
+ \pi\leq 2 \abs{S^-}^{1/2}\delta \max\big(\abs{S^-}^{1/2},\sqrt\pi\,\big)
+ \delta^2 \max(\abs{S^-},\pi).
\ee
Indeed, \eqref{E:defect} with $\abs{S}=\pi R^2$ implies that the LHS of \eqref{E:eqBMstab} is bounded by $\pi\eps$ which, with the choice $\delta=\sqrt\eps$ is bounded by the right-hand side of \eqref{E:eqBMstab} and thus also \eqref{E:BMstab}. 

Note that \eqref{E:K-S}, \eqref{E:K-D},  \eqref{isope} and \eqref{E:Bonnassumpt} (via \eqref{E:defect}) jointly imply the bounds 
on the volumes $\abs{S^-}$ and $\abs{D^-}$,
\be
\label{E:|S^-|}
\pi (R-1)^2 -\pi\eps \le \abs{S^-}\le \pi (R-1)^2
\ee
and
\be
\label{E:D->}
\pi\bigl((R-1)^2 -R^2 \tilde\xi(\sqrt{\eps}\,)-\eps \bigr)\le \abs{K^-} -\pi R^2 \tilde\xi(\sqrt{\eps}\,)
\le\abs{D^-}\le \abs{K^-}\le \pi\bigl((R-1)^2 +R^2 \tilde\xi(\sqrt{\eps}\,)\bigr).
\ee
Taking into account that $R-1>\eps_0$ and thus $\frac{R}{R-1}\le 1+\frac1{\eps_0}\le \frac2{\eps_0}$ once $\eps_0<1$,
we see that the above bounds imply the restrictions on the diameter $r$ of the disk $D^-$,
\be
\label{E:r}
(R-1)\Bigl(1-4\frac{\tilde\xi(\sqrt{\eps}\,)+\eps}{\eps_0^2}\Bigr)\le r\le (R-1)\Bigl(1+4\frac{\tilde\xi(\sqrt{\eps}\,)}{\eps_0^2}\Bigr).
\ee
Let $\xi(\eps)=  2 L \frac{\tilde\xi(\sqrt{\eps}\,)^{2/3}+\eps}{\eps_0^2}$. We will first show that $S\subset B_{R+\xi}(0)$. Using the  convexity of $K^-$, the condition in \eqref{E:K-D} implies that $K^-\subset (D^- + B_{\zeta}(0))$ with  $\zeta=(\frac{R}{R-1})^{1/3} R (4 \tilde\xi(\sqrt\eps))^{2/3}$. Indeed, if $x\in K^-\setminus D^-$, then the set $K^-$ contains the union of  $D^-$ and the wedge bordered by $\partial D^-$ and the tangents through $x$ to $D^-$. The  volume of this wedge is $r\sqrt{x^2-r^2}-r^2\arctan(\frac{\sqrt{x^2-r^2}}{r})$. Asymptotically, for small $x-r$, this equals $\frac23{\sqrt{2 r}}(x-r)^{3/2}$, implying, in view of the lower bound in \eqref{E:r} for sufficiently small $\eps$, 
\be
(x-r)^{3/2}\le \frac{3}{2\sqrt2}\frac{\pi R^2 \tilde\xi(\sqrt\eps)}{\sqrt{r}}
\le \frac{4 R^2 \tilde\xi(\sqrt\eps)}{\sqrt{R-1}}
\ee
and thus $x\le r + \zeta$. With the lower bound on $r$ from \eqref{E:r}, we get
\be
x \le r+ \Bigl(\frac{R}{R-1}\Bigr)^{1/3} R (4 \tilde\xi(\sqrt\eps))^{2/3}
\le R-1+ \Bigl(1+\frac1{\eps_0}\Bigr)^{1/3} R (4 \tilde\xi(\sqrt\eps))^{2/3} 
+4\frac{\tilde\xi(\sqrt{\eps}\,)+\eps}{\eps_0^2}.
\ee
Using that $R<L/3$ and $\tilde\xi(\sqrt{\eps}\,),\eps,\eps_0<1$, we get
\be
\Bigl(1+\frac1{\eps_0}\Bigr)^{1/3} R (4 \tilde\xi(\sqrt\eps))^{2/3} +4\frac{\tilde\xi(\sqrt{\eps}\,)+\eps}{\eps_0^2}
\le 2 L \frac{\tilde\xi(\sqrt{\eps}\,)^{2/3}+\eps}{\eps_0^2}\le \xi,
\ee
implying that $K^-\subset B_{R-1+\xi}(0)$ and
\be 
\label{E:Sin}
S\subset K\subset B_{R+\xi}(0).
\ee

On the other hand, we get $B_{R-\xi}(0)\subset S$ by strengthening the second claim in \eqref{E:K-S} for an admissible $S\in\mathcal S$. We can argue that $S\supset B_{r+1-s}(0)=B_{r+1}(0)\ominus B_s(0)$ with $s(4 R^2 \tilde\xi(\sqrt\eps\,)^{2/3}$. Indeed, if $x\in (\T\setminus S) \cap B_{r+1-s}(0)$, then $B(x)\cap B_{r-s}(0) \cap S^-=\emptyset$, while 
\be
\label{E:cap}
\abs{B(x)\cap K^-}
\ge\abs{B(x)\cap B_{r}(0)} \geq  \frac43 \sqrt{2}s^{3/2}.
\ee
To get the above inequality, we first note that if $\abs{x}\le  r+1-s$, then $\abs{B(x)\cap B_{r}(0)} \ge \abs{B(y)\cap B_{r}(0)} $ where $y= \frac{(r+1-s)x}{\abs{x}}$ and, in addition, the intersection of the last two disks is larger than the area of the circular segment that is cut off the disk $B(y)$ by the chord passing through the two points in $\partial B(y)\cap \partial B_{r}(0)$. The area of this circular segment is given as the area of the corresponding circular sector minus the triangle with vertices $y$ and the endpoints of the chord, $\arccos(1-s)-(1-s)\sqrt{1-(1-s)^2}= \frac43 \sqrt2 s^{3/2} + O(s^5/2)$. The inequality \eqref{E:cap} with $s=\bigl(4 R^2 \tilde\xi(\sqrt\eps\,)\bigr)^{2/3}$ is in contradiction with the second inequality in \eqref{E:K-S}, and thus such point $x$ does not exist and $S\supset B_{r+1-s}(0)$ as was claimed above.

For the claim about connectedness and simple connectedness of $S^-$ we then refer to Lemma~\ref{lem:admissible}(\ref{noholes}).
\end{proof}


\end{document}